\documentclass[10pt, reqno]{amsart}
\usepackage[UKenglish]{babel}

\usepackage{tikz-cd}
\usepackage[utf8]{inputenc}
\usepackage{amsmath}
\usepackage{amssymb}
\usepackage{amsthm}
\usepackage{mathtools,dsfont}
\usepackage{url}
\usepackage{mdwlist}

\usepackage{graphicx}
\usepackage{float}

\usepackage{doi,comment}

    \usepackage[numbers,sort&compress]{natbib}

\usepackage{xparse,aliascnt,bookmark}
\usepackage{hyperref}
\usepackage{thmtools, thm-restate}

\addto\extrasUKenglish{
    
}

\hypersetup{
    colorlinks,
    linkcolor={red!50!black},
    citecolor={blue!50!black},
    urlcolor={blue!80!black}
}

\numberwithin{equation}{section}

\def \Lp {\mathrm L}  % L^p space L
\def\Ln {\mathcal L}  % Lesbegue measure L
\def \R {\mathbb R}
\def \N {\mathbb N}
\def\transop {\sigma}
\def \Im {\mathfrak{Im}}

\newcommand*{\blank}{{\makebox[1ex]{\textbf{$\cdot$}}}}

\declaretheorem[numberwithin=section,name=Theorem]{theorem}
\declaretheorem[sibling=theorem,name=Lemma]{lemma}
\declaretheorem[sibling=theorem,name=Corollary]{corollary}
\declaretheorem[sibling=theorem,name=Proposition]{proposition}

\declaretheorem[sibling=theorem,name=Conjecture]{conjecture}

\theoremstyle{remark}
\declaretheorem[name=Remark,numbered=no]{remark}
\declaretheorem[name=Remarks,numbered=no]{remarks}

\theoremstyle{definition}
\declaretheorem[sibling=theorem,name=Definition]{definition}

\DeclareMathOperator{\Ai}{Ai}

\title[Enhanced area-laws at vanishing magnetic fields]{Logarithmically enhanced area-laws for fermions in vanishing magnetic fields in dimension two}
\author[P.~Pfeiffer and W.~Spitzer]{Paul Pfeiffer and Wolfgang Spitzer}
\address{Fakultät für Mathematik und Informatik, FernUniversität in Hagen, Universitätsstraße 1, 58097, Hagen, Germany}
\date{\today}
\email{paul.pfeiffer@fernuni-hagen.de}
\email{wolfgang.spitzer@fernuni-hagen.de}

\allowdisplaybreaks
\keywords{Area laws, Landau Hamiltonian, asymptotic analysis}
\subjclass[2010]{Primary 47G30, 35S05; Secondary 45M05, 47B10, 47B35}

\begin{document}

\begin{abstract}
We consider fermionic ground states of the Landau Hamiltonian, $H_B$, in a constant magnetic field of strength $B>0$ in $\mathbb R^2$ at some fixed Fermi energy $\mu>0$, described by the Fermi projection $P_B\coloneqq 1(H_B\le \mu)$. For some fixed bounded domain $\Lambda\subset \R^2$ with boundary set $\partial\Lambda$ and an $L>0$ we restrict these ground states spatially to the scaled domain $L \Lambda$ and denote the corresponding localised Fermi projection by $P_B(L\Lambda)$. Then we study the scaling of the Hilbert-space trace, $\mathrm{tr} f(P_B(L\Lambda))$, for polynomials $f$ with $f(0)=f(1)=0$ of these localised ground states in the joint limit $L\to\infty$ and $B\to0$. We obtain to leading order  logarithmically enhanced area-laws depending on the size of $LB$. Roughly speaking, if $1/B$ tends to infinity faster than $L$, then we obtain the known enhanced area-law (by the Widom--Sobolev formula) of the form $L \ln(L) a(f,\mu) |\partial\Lambda|$ as $L\to\infty$ for the (two-dimensional) Laplacian with Fermi projection $1(H_0\le \mu)$. On the other hand, if $L$ tends to infinity faster than $1/B$, then we get an area law with an $L \ln(\mu/B) a(f,\mu) |\partial\Lambda|$ asymptotic expansion as $B\to0$. The numerical coefficient $a(f,\mu)$ in both cases is the same and depends solely on the function $f$ and on $\mu$. The asymptotic result in the latter case is based upon the recent joint work of Leschke, Sobolev and the second named author~\cite{Leschke2021} for fixed $B$, a proof of the sine-kernel asymptotics on a global scale, and on the enhanced area-law in dimension one by Landau and Widom. In the special but important case of a quadratic function $f$ we are able to cover the full range of parameters $B$ and $L$. In general, we have a smaller region of parameters $(B,L)$ where we can prove the two-scale asymptotic expansion $\mathrm{tr} f(P_B(L\Lambda))$ as $L\to\infty$ and $B\to0$.  

\end{abstract}
\maketitle
{\small
\tableofcontents
}

\section{Introduction}

In recent years, there has been a lot of efforts devoted to entanglement entropy (EE). 
The motivation of the present work is to understand the transition between a (strict) area law and an enhanced area law for the EE of fermionic ground states. 

Even in the simplest situation when there are no particle interactions present, the entanglement (or rather local) entropy of ground states is a complicated and interesting function of the defining parameters. Also, the asymptotic behaviour (for large domains) of this entropy is related to the asymptotic behaviour of Szeg\H{o}-type asymptotics of Toeplitz or Wiener--Hopf operators which has been studied since more than a century. It was Harold Widom who conjectured in 1990 (see \cite{Widom1990}) a formula in the higher-dimensional setting and proved a special case. That conjecture led D.~Gioev and I.~Klich~\cite{Gioev2006} to conjecture in 2006 the asymptotic expansion of the EE of ground states of the ideal Fermi gas. In 2013, A.V.~Sobolev~\cite{sobolev2013pseudo} proved Widom's conjecture which in turn paved the way to prove the conjecture on the EE by Leschke, Sobolev and one of the present authors in \cite{Leschke2014}. 

In 2021, the same authors proved the area law for the EE of ground states of the ideal Fermi gas in a constant magnetic field in the two-dimensional case, see \cite{Leschke2021}. In the three-dimensional case the situation is different from the start since the spectrum of the Landau Laplacian is now purely absolutely continuous. We proved a logarithmically enhanced area-law recently in \cite{PfeifferSpitzer2022}. 

Some connections between the appearance of a strict area-law versus a logarithmically enhanced area-law are obvious. For example, if the off-diagonal integral kernel of the Fermi projection (characterizing the ground state) is decaying fast (exponentially, say) then an area law holds. On the other hand, purely absolutely continuous spectrum does not guarantee a logarithmically enhanced area-law.

Let us recall some more mathematical results that add to the understanding of EE of non-interacting Fermi gases. In \cite{Pfirsch2018}, Pfirsch and Sobolev treat a periodic (electric) potential $V$ in dimension one and prove a logarithmically enhanced area-law. What is particularly interesting is that the second-order term (or ``surface" term of the order $\ln(L)$) is the same as for the Laplacian, that is, with $V=0$. The higher dimensional case remains an open problem. 

Stability of the enhanced area-law by a local (compactly supported) perturbation $V$ was proved by M\"uller and Schulte in \cite{muller2020stability,muller2021stability}. Motivated by these papers, the first named author of the present paper proved the stability of the area law for the two-dimensional Landau Hamiltonian by allowing a perturbation on the magnetic potential and a perturbation by an electric potential, see \cite{pfeiffer2021stability}. 

There are also results on the EE of random systems described by an Anderson-type Hamiltonian. They concern the surprising logarithmic enhancement of EE in the one-dimensional dimer model at a certain Fermi energy proved by M\"uller, Pastur and Schulte \cite{Muller2020}. In a more general case, Pastur and Slavin~\cite{PasturSlavin14} and Elgart, Pastur and Shcherbina~\cite{Elgart2017} proved an area law for the EE at a Fermi energy in the localisation regime for an Anderson model on the lattice $\mathbb Z^d$. However, their formula for the leading coefficient is not very explicit and it is not known how it depends on the disorder parameter. No rigorous result is available when the Fermi energy lies in the delocalisation regime of a random Hamiltonian, but this question touches on the notoriously difficult problem of the existence of such a regime in the first place.

We continue here the study of the local or entanglement entropy of ground states of the ideal Fermi gas in a constant magnetic field in dimension two. To this end, we fix some Fermi energy $\mu>0$ and denote by $P_B \coloneqq 1(H_B\le \mu)$ the spectral (or Fermi) projection of the Landau Hamiltonian 
\begin{equation}\label{def Landau} H_B \coloneqq (-i\nabla - a)^2 \,,
\end{equation}
where $i$ is the imaginary unit and $\nabla = (\partial_{x_1},\partial_{x_2})$ is the gradient. We choose the symmetric gauge $a(x) = (a_1(x), a_2(x)) \coloneqq (x_2 ,-x_1 )B/2$ for the vector potential $a: \R^2 \to \R^2$ generating the constant magnetic field (vector) perpendicular to the plane with Cartesian coordinates $x = (x_1, x_2)$. The strength of magnetic field is given by the real number $B>0$. On (a suitable domain of) $\mathsf{L}^2(\R^2)$, $H_B$ acts as a (positive) self-adjoint operator.

For some (bounded) Borel set $\Lambda\subset \R^2$ with Lebesgue volume $|\Lambda|$, we consider the localised Fermi projection
\begin{equation}
P_B(\Lambda) \coloneqq 1_\Lambda \, P_B \, 1_\Lambda\,,
\end{equation}
where $1_\Lambda$ is the multiplication operator with the indicator function of $\Lambda$. For our asymptotic results we assume in addition that $\Lambda$ is an open domain (that is, $\Lambda$ has only finitely many connected components) with some ``smoothness" properties of the boundary, $\partial\Lambda$. The latter may be piecewise $\mathsf{C}^2$-smooth in our first main result, piecewise $\mathsf{C}^1$-smooth in our second main result and a polygon  or $\mathsf{C}^2$-smooth in our third main result.

For some (suitable) ``test" function $f$ we are then interested in the Hilbert-space trace $\mathrm{tr} f(P_B(\Lambda))$. The most relevant cases are the quadratic polynomial $f(t) = t(1-t)$ related to particle number fluctuations and $f(t) = -t\ln(t)-(1-t)\ln(1-t)$ related to the von Neumann EE. For fixed $\Lambda$, this trace is too complicated but it is of interest to study the behaviour for large domains. To this end, we introduce a scaling parameter $L>0$ and consider, for fixed $\Lambda$ and fixed $\mu$, the function $( L,B) \mapsto \mathrm{tr} f(P_B(L\Lambda))$ as we let $L\to\infty$. This has been completely analysed in \cite{Leschke2021} for fixed $B$,  namely, the following area law has been proved (under the condition that $\partial\Lambda$ is $\mathsf{C}^3$-smooth),
\begin{equation} \label{1.3}
\mathrm{tr} f(P_B(L\Lambda)) = L^2 B \frac{|\Lambda|}{2\pi} (n+1) f(1) + L\sqrt{B} \,|\partial\Lambda|\, \mathsf{M}_{\le n}(f) + o(L)\,,
\end{equation}
see \cite[Theorem 2]{Leschke2021}. Here, $n\coloneqq \lfloor(\mu/B - 1)/2\rfloor$ is the number of Landau levels below $\mu$ and the coefficient $\mathsf{M}_{\le n}(f)$ is defined in \eqref{def M_K}. Hence, if $f(1)=0$, then the leading contribution as $L\to\infty$ is of the order $L |\partial\Lambda|$, which is the reason why it is called an area law. We speak of an enhanced area-law if the leading term is larger. The most prominent example is when there is an extra factor of $\ln(L)$. Such a logarithmically enhanced area-law is present for the (free) Laplacian (set $B=0$ in the above Hamiltonian $H_B$), as was conjectured by Gioev and Klich in \cite{Gioev2006} and proved in \cite{Leschke2014}. 

The concrete purpose of the present paper is to study the transition from an area law to a logarithmically enhanced area-law as $B$ vanishes and the number of Landau levels $n$ tends to infinity. The results may also be interpreted as a high energy limit where the magnetic strength $B$ is kept fixed and the Fermi energy (or the number of Landau levels) and the scaling parameter $L$ tend to infinity.

We cannot use the above result \eqref{1.3} from the constant $B$ case  directly as we have no control over lower order error terms, which depend, in general, on $n$ and  might, a priori, blow up as $n\to\infty$. The joint limit $B\to0$ and $L\to\infty$ (for fixed $\mu$) depends crucially on the product $BL$. %We will obtain the same logarithmically enhanced area-law of the order $L\ln(L)$ as for the two-dimensional Laplacian if $BL\to0$. On the other hand, if $BL\to\infty$ then we obtain the logarithmically enhanced area-law of the order $L\ln(\mu/B)$. %The extra $\ln(\mu/B)$ is in fact the result of a one-dimensional logarithmically enhanced area-law. 
We venture to state the following conjecture, which we will prove in certain relevant circumstances, in particular only for polynomials $f$. 
\begin{conjecture} \label{main} For any H\"older-continuous function $f$ with H\"older exponent strictly bigger than $0$, which satisfies $f(0)=f(1)=0$, any bounded Lipschitz domain $\Lambda$, and any $\mu>0$ we have the asymptotic expansion
 \begin{align} \label{main eq}
 \operatorname{tr} f(1_{L\Lambda} 1(H_{B}\le\mu) 1_{L \Lambda} )= \begin{cases}
   \lvert \partial \Lambda \rvert  \, \frac{2\sqrt{\mu}}{\pi} \, \mathsf{I}(f) L \ln(\sqrt\mu L) + o(\sqrt\mu L \ln(\sqrt\mu L))\, &\text{ if } BL< \sqrt \mu \, ,   \\
   \lvert \partial \Lambda \rvert  \, \frac{2\sqrt{\mu}}{\pi} \, \mathsf{I}(f) L \ln(\mu/B) + o(\sqrt\mu L \ln(\mu/B))\, &\text{ if } BL\ge \sqrt \mu \, ,
 \end{cases}
 \end{align}
as $L \to \infty$ and $B\to 0$, where we defined the functional
\begin{align} \label{def I(f)}
f\mapsto\mathsf{I}(f)  \coloneqq \frac 1 {4 \pi^2} \int_0^1 \frac {f(t)}{t(1-t)} \,\mathrm{d}t \,.
\end{align}

\end{conjecture}
In order to explain the numerical factors in this formula we recall the asymptotic formula for the Laplacian in two dimensions. As mentioned, it was proved that for a H\"older-continuous function $f$ with $f(0)=f(1)=0$ (using the notation of \cite[(7)]{Leschke2014})
\begin{align} 
\operatorname{tr} f(1_{L\Lambda} 1(-\Delta\le\mu) 1_{L \Lambda} )=  J(\partial\Gamma,\partial\Lambda) \, \mathsf{I}(f) \sqrt\mu L \ln(\sqrt\mu L) + o(\sqrt\mu L \ln(\sqrt\mu L))\,, 
 \end{align}
as $L \to \infty$, where $\Gamma \coloneqq \{p\in\R^2 : p^2\le\mu\}$, $\hbar=1$, and
\begin{align} J(\partial\Gamma,\partial\Lambda) = \frac{2}{(\frac12)!}\left(\frac{\mu}{4\pi}\right)^{\frac{1}{2}} |\partial\Lambda| = \frac{2\sqrt{\mu}}{\pi}\,|\partial\Lambda| \, ,
\end{align}
with $(1/2)! \coloneqq \sqrt \pi /2$ being defined by Euler's gamma function.

This transition of area laws is similar in spirit to the transition that happens for the free Laplacian as one lets the temperature $T$ go to zero in the study of the EE of equilibrium states. As proved by A.V.~Sobolev \cite{Sobolev2017} in arbitrary spatial dimension and prior by Leschke, Sobolev and one of the present authors in \cite{Leschke2017} in dimension one, there are two regions for $(T,L)$: if $TL\to 0$, then there is an enhanced area-law of the order $L^{d-1}\ln(L)$. On the other hand, if $TL\to\infty$, then we have an enhanced area-law of the order $L^{d-1}\ln(T_0/T)$, for some temperature $T_0>0$. To draw the connection between the two scenario one could identify $T$ with $B$. 

However, from a technical point of view, the two cases differ in the sense that there is a full-fledged pseudo-differential calculus for the study of the Laplacian, or more general, of translation-invariant operators. This is not so much the case for the magnetic Laplacian $H_B$ but we can rely on well-established (asymptotic) properties of Hermite and Laguerre polynomials.

At the end of this section, we will argue by scaling that it suffices to consider $\mu=2$, only. Moreover, we find it more convenient to switch from $B$ to $\mu/(2B)\approx K\in\N$ and hence let $K\to\infty$. Our main results in this paper are therefore formulated for $\mu=2$ and in the joint limit $L\to\infty,K\to\infty$.

Our first main result is \autoref{pf thm} which deals with particle number fluctuations (that is, the function $f(t)=t(1-t)$) and the full parameter set of $K$ and $L$. We can handle the quadratic test function $f$ because the phase $\exp(i x\wedge y/(2K))$, which appears in the integral kernel of $P_B$, cancels in the computation of the trace and is out of our way. Nevertheless, it is an important case and we believe that is yields the right picture in the general case. Therefore, we venture to state \autoref{main}. We discuss the quadratic case in Section \ref{section: pf}. 

Our second main result \autoref{thm large K} is the logarithmically enhanced area-law of the order $L\ln(L)$ if $K$ is much larger than $L$ and the Landau Hamiltonian $H_B$ is ``close" to the (free) two-dimensional Laplacian. Of course, for any finite $K$ (or strictly positive $B$), the spectrum of $H_B$ is never anything like that of the Laplacian and the off-diagonal integral kernel $P_B(x,y)$ decays exponentially to 0 as the distance $\|x-y\|$ tends to infinity. But the rate is given by $1/K$ and that goes to 0 in the end and we do get the convergence to the integral kernel of $1(-\Delta\le \mu)$. More precisely, we are able to prove the enhanced area-law under the condition $L\le C K^{2/5}$ but, as just said, we believe that this holds true up to the transition line $K = C L$.  We should note that we use (but do not reprove) the known result for the Laplacian ($1/K=0$ in a way) and consider the case with small $1/K$ as a perturbation.

Our third main result is \autoref{thm small K} and deals with the region $K\ll L$. Here, due to the slow vanishing of the magnetic field, we are more in the regime of a constant magnetic field and the area law is of the order $L \ln(K)$. Interestingly, the $\ln(K)$ is a result of an enhanced area-law for the one-dimensional Laplacian where $K$ is the effective scaling parameter. Distilling the one-dimensional Laplacian is the result of the so-called sine-kernel asymptotics for Hermite and Laguerre functions. The difficulty here is that we need this asymptotics on a global scale, which takes up some space to prove. We succeed to prove \autoref{main} (almost) over the full range of parameters $K\le C L$ when the domain $\Lambda$ is a polygon; in fact, we have to assume $K\le C L/\ln(L)$. When $\partial\Lambda$ is $\mathsf{C}^2$-smooth we lose control over some error terms and end up with the restriction $K^2\le C L$. 

We return to some open questions in \autoref{concluding remarks}.

\subsection{Some notations and preliminary definitions}

A \emph{domain} $\Lambda$ is a (non-empty) bounded, open set in the two-dimensional Euclidean space $\R^2$ with finitely many connected components. It is called $\mathsf{C}^r$-smooth or piecewise $\mathsf{C}^r$-smooth if the boundary $\partial\Lambda = \bar{\Lambda}\setminus\Lambda$ is a $\mathsf{C}^r$-smooth curve, respectively a piecewise $\mathsf{C}^r$-smooth curve, for some $r\in\N$. $\Lambda$ is called Lipschitz if the boundary is Lipschitz continuous. The surface area $|\partial\Lambda|$ is the one-dimensional Hausdorff measure of $\partial\Lambda$. 

By $D_R(x)$ we denote the open disk of radius $R>0$ at the centre $x\in\R^2$ and by $D_R(S)\coloneqq \bigcup_{x\in S} D_R(x)$ the $R$-neighbourhood of a set $S\subset \R^2$.

We denote the set of natural numbers by $\N \coloneqq \{1,2,\ldots\}$ and $\N_0 \coloneqq \N \cup\{0\}$ the set including $0$.

The parameter $L$ is a positive real number which scales the domain $\Lambda$ and goes to infinity in our asymptotic results. The parameter $K$ is another positive (in most cases natural) number and the inverse of the magnetic field strength, which also tends to infinity (in most statements). It determines the Fermi projection $P_K$, see \eqref{def P_K}. 

The indicator function of a set $I\subset\R^m$ is denoted by $1_I$ and our notation does not distinguish between this function and the multiplication operator by this function. The identity operator is denoted by $\mathds{1}$.  

We use the standard big-$\mathcal O$ and little-$o$ notation. That is, for two functions $f$ and $g>0$, $f=\mathcal O(g)$ if $|f(L)|\le C g(L)$ for some (finite) constant $C$ and sufficiently large $L$ and $f=o(g)$ if $\limsup |f|/g(L) = 0$. In the latter case, we also write $f\ll g$ or $g\gg f$.  In a series of estimates the specific value of a constant may change from line to line without changing its name. Constants are always finite real numbers and usually strictly positive.

\subsection{Reduction to the case $\mu=2$, introduction of $K$ and $L$}
As is well-known, the spectrum of the Landau Hamiltonian $H_B$ of \eqref{def Landau} equals $B(2 \mathbb N_0+1)$. Each eigenvalue is infinitely degenerate. Let $\Pi_{\ell,B}$ be the projection onto the eigenspace with eigenvalue $B(2\ell+1)$ for $\ell \in \mathbb N_0$. For some fixed (Fermi energy) $\mu>0$ we work with the spectral projection $1(H_B\le \mu) = \sum_{\ell =0}^\nu \Pi_{\ell,B}$ with $\nu \coloneqq \lfloor(\mu/B-1)/2\rfloor$, also called Fermi projection. 

The expressions we are interested in are only dependent on the eigenvalues of the operator $1_{\Lambda} 1(H_B \le \mu) 1_{\Lambda}$ for some domain $\Lambda \subset \R^2$. For any $\lambda \in \R^+$, this operator is unitarily equivalent to  $1_{ \lambda^{-1} \Lambda} 1(H_{\lambda^2 B } \le \lambda^2 \mu ) 1_{\lambda^{-1}\Lambda}$. We define $K\coloneqq \lfloor(\mu/B-1)/2\rfloor+1$.  
We can then assume without loss of generality that $B=1/K$, as both sides of \eqref{main eq} are invariant under this scaling. For the rescaled $\mu'$ we observe $K-1= \lfloor (\mu' K-1)/2 \rfloor = \lfloor (2K-1)/2 \rfloor$ and this implies $1(H_{1/K} \le \mu' )= 1(H_{1/K} \le 2)$. This shows $\mu'=2 + \mathcal O(1/K)$ and thus, replacing $\mu'$ by $2$ on the right-hand side of \eqref{main eq} only changes the leading term of the right-hand side by an additive error term of order $L \ln(\min(L,K) ) /K = \mathcal O(L)$. This is why, in the following, we always assume 
\begin{align} \mu =2 \quad\mbox{ and } \quad B=1/K\; \mbox{ for } \; K \in \mathbb N\,.
\end{align} 
Finally, we redefine the projection
\begin{equation} \label{def P_K}
P_K\coloneqq 1(H_B\le 2) = \sum_{\ell=0}^{K-1} \Pi_{\ell,1/K}\,.
\end{equation}

\section{Preliminary asymptotic results on the integral kernel of the Fermi projection}\label{kernel asymp section}

This section starts with the integral kernel of the Fermi projector $P_K$ and collects various estimates on this projector, which are needed throughout the paper. Most of them are probably well-known and we list them here for completeness and the convenience of the reader. We do not claim any novelty.

We introduce the function $F_K^{(\alpha)}$ on $\R^+$, which is related to the (generalised) Laguerre polynomial of degree $K-1$. Then we study its asymptotic properties as $K$ becomes large. This is split into two subsections, one is devoted to small arguments $x$, that is, to $x\le 1/2$ and the other one to large arguments, that is, to $x\ge 1/2$. The main results on the asymptotic expansion of the translation invariant part, $G_K$, of the integral kernel of $P_K$ are collected in \autoref{GK for t<1/2 K} and \autoref{GK for t>1/2 K}. The last subsection contains an integral bound on $G_K$, which is of immediate use in the next section on particle number fluctuations.

In this section, we study the integral kernel of the Fermi projection and in particular, how it behaves asymptotically for small magnetic fields. We will see in which specific way it converges to the free projector. This kernel is given for both $x=(x_1,x_2)$ and $y=(y_1,y_2)$ in $\mathbb R^2$ by
\begin{align}
P_K(x,y) &=1(H_{1/K}\le2)(x,y) 
\\
& = \frac 1 {2 \pi K}\exp \left(-\frac{1}{4K} \lVert x -y \rVert^2 + i\frac {1} {2K} x \wedge y \right) \Ln_{K-1}^{(1)} \left( \frac 1 {2K}  \lVert x-y \rVert^2 \right) \\
&= \exp \left( i\frac {1} {2K} x \wedge y \right) G_K(\lVert x-y \rVert/\sqrt 8)  \,, \label{pl kernel}
\end{align}
where $x\wedge y\coloneqq x_1y_2-x_2y_1$, 
\begin{align*} \Ln_{K-1}^{(\alpha)}(t)\coloneqq \sum_{j=0}^{K-1} \frac{(-1)^j}{j!} {K-1+\alpha\choose K-1-j} t^j\,,\quad t\ge0
\end{align*}
is the (generalised) Laguerre polynomial of degree $K-1$ (for any $\alpha\in\R$), see \cite[(5.1.6)]{Szego}  and 
 \begin{align}
 G_K(t) \coloneqq \frac 1 {2 \pi K} \exp \left(-4K \frac{t^2}{2K^2}  \right) \Ln_{K-1}^{(1)} \left( 4K \frac{t^2}{K^2} \right),\quad t\ge0\,.
 \end{align}
We need to be a bit more general and define for $\alpha \in \mathbb R$ and $\nu \coloneqq 4K+2 (\alpha-1)$ the function
 \begin{align}
 F_K^{(\alpha)}(x) \coloneqq 2^\alpha \sqrt \nu x^{\frac 1 2 \alpha+ \frac 1 4 } \lvert 1-x\rvert ^{\frac 1 4} \exp(- \nu x/2) \Ln_{K-1} ^{(\alpha)} (\nu x) \,,\quad x > 0\,.
 \end{align}
The definition of this function is based on the asymptotic analysis of $G_K$, which can be understood using $\alpha=1$. It can be found in \cite[\href{http://dlmf.nist.gov/18.15.iv}{(18.15.(iv))}]{NIST:DLMF}. We will consider the equations (18.15.17) to (18.15.23) in \cite{NIST:DLMF}. In particular, the equations (18.15.19) and (18.15.22) provide us with the asymptotics for $G_K(t)$. We start by identifying the variables and parameters. According to (18.15.17) and trivialities, we see that
 \begin{align}
 \alpha=1\,,\quad  n=K-1\,, \quad \nu=4K\,, \quad x= \frac{t^2}{K^2}\,.
 \end{align}
Later on, we will also use the cases $\alpha= \pm \frac 12 $ to find the asymptotic expansion of the Hermite polynomials.
 
This allows us to rewrite $G_K$ as
 \begin{align} \label{G to F}
 G_K(t) = \frac 1 { 8 \pi t^{3/2} \left \lvert 1- \frac {t^2} {K^2} \right\rvert ^{1/4}  } F_K^{(1)} \left( \frac{t^2}{K^2} \right).
 \end{align}
 
We intend to show
\begin{align}
F_K^{(\alpha)}(x)  = 
	\begin{cases}
	\sqrt{ \nu \xi(x)} J_\alpha(\nu \xi(x)) + \mathcal O( 1/K)  &\text{ if } 0 \le x \le 1/2\,, \\
	\sqrt{2/ \pi } \cos ( \nu \xi(x)+g(\alpha)) + \mathcal O(1/(1+K\sqrt x))  & \text{ if } 0 \le x \le 1/2\,, \\
	\sqrt{2/ \pi } \cos ( \nu \xi(x)+g(\alpha)) + \mathcal O(1/(1+K(1-x)^{3/2}))  & \text{ if } 1/2 \le x \le 1\,, \\
	0+ \mathcal O(\exp(-K (x-1)^{3/2}/10^4)  & \text{ if } 1 \le x \le 3/2\,, \\
	0 + \mathcal O(\exp(-K x/10^4)  & \text{ if } 3/2 \le x \,.
	\end{cases}
\end{align}
 The factors $10^4$ are obviously placeholders and can depend on $\alpha$, but not on $K$. The function $\xi$ is given in \eqref{def xi eq} below (or \cite[\href{http://dlmf.nist.gov/18.15.E18}{(18.15.18)}]{NIST:DLMF}) and $g(\alpha)=-\pi\alpha /2- \pi/4$ is affine linear. 
 \begin{align} \label{def xi eq}
\xi(x) \coloneqq \frac 12 \left( \sqrt{x-x^2} + \operatorname{arcsin}(\sqrt{x})\right),\quad x\in [0,1]\, .
\end{align}
Let us also define $\eta \colon [-1,\infty) \to \mathbb C$ by $\eta(0)=0$ and for any $t \in (-1,\infty)$
\begin{align} \label{def of eta}
\eta'(t) \coloneqq \begin{cases} \sqrt{1-t^2}\, , \quad & \text{ if } -1 \le t \le 1\,, \\ i \sqrt{t^2-1} \,, \quad& \text{ if } t \ge 1\,. \end{cases}
\end{align}
An explicit formula for $\eta$ is given by
\begin{align} \label{eta expicit}
\eta(t)= \begin{cases} 1/2\left( t \sqrt{1-t^2} + \arcsin (t)\right) , \quad & \text{ if } -1 \le t \le 1\,, \\
i/2 \left( t \sqrt{t^2-1} - \operatorname{arccosh}(t) \right) + \pi/4\,, \quad & \text{ if } t\ge 1\,.
\end{cases}
\end{align}
We see that this function satisfies 
\begin{align} \label{xi in eta}
\eta(\sqrt t) = \xi(t)\quad\mbox{for any } t\in [0,1] \mbox{ and } \eta(1)=\pi/4\,.
\end{align}
At this point, let us also define the function $\zeta$
\begin{align} 
\zeta(s)\coloneqq - \left[3/2(\pi/4- \eta(\sqrt s))\right]^{2/3}\,,\quad s \ge 0\,. \label{def zeta}
\end{align}
In the last equation, the terms inside the brackets is either in $\mathbb R^+$ or in $(-i) \mathbb R^+$  and thus the two thirds power has a unique real value with the branch chosen by $(-i)^{2/3}=-1$. Thus, for $s \ge 1$, we get
\begin{align} \label{zeta for x>0 good sign}
\zeta(s)= -\left[ 3/2 (-i)\, \Im (\eta (\sqrt s)) \right]^{2/3} = \left[ 3/2 \,\Im (\eta (\sqrt s)) \right] ^{2/3}\,,
\end{align}
where $\Im$ refers to the imaginary part.
\begin{lemma} \label{eta props}
For $h \in (-1,1)$, the function $\eta$ satisfies the following properties:
\begin{align}
\eta(h)&=h+ \mathcal O(h^3)\,, \label{eta at 0eq}\\
|\eta(h)| &\ge (\pi/4) |h|\,, \label{eta lower linear bound} \\ 
3/2 \big(\pi/4- \eta(1-h)\big)&=\sqrt 2 \Big( \sqrt h^3 - \frac{3 }{20} \sqrt h^5 + \mathcal O(\sqrt{\lvert h\rvert}^{7})\Big)\,, \label{eta at 1eq}
\end{align}
where $\sqrt h = i \sqrt{-h}$ for $h \le 0$. Furthermore, $\eta^{-1}$ is $2/3$-Hölder continuous on $\eta([-1,2])$.
\end{lemma}
\begin{proof}
The expansion for $\eta(h)$ is trivial. The lower bound for $\eta(h)/h$ holds for positive $h$, as $\eta$ is concave on $[0,1]$ and thus, the line from $(0,0)$ to $(1,\pi/4)$ lies below the graph of $\eta$. For negative $h$, the claim follows as $\eta(-h)=-\eta(h)$.

 The expansion for $\eta(1-h)$ follows from $\eta(1)=\pi/4$ and
\begin{align}
\eta'(1-h)=\sqrt{1-(1-h)^2}=\sqrt{h(2-h)}=\sqrt 2  \sqrt h - \frac{ \sqrt 2} {4} \sqrt h^3 + \mathcal O(|h|^{5/2})\,,
\end{align}
where we choose the branch of the square root, which satisfies $\sqrt{-t}=i \sqrt{t}$ for any $t\ge 0$. Integrating this equation leads to
\begin{align}
3/2 \big(\pi/4- \eta(1-h)\big)=\sqrt 2 \Big( \sqrt h^3 - \frac{3 }{20} \sqrt h^5 + \mathcal O(|h|^{7/2})\Big)\,.
\end{align}
The inverse $\eta^{-1}$ is locally Lipschitz continuous away from $\eta(1)=\pi/4$, as $\eta'$ only vanishes at $1$. From the expansion of $\eta(1-h)$, we can thus conclude that $\eta^{-1}$ is $2/3$-Hölder continuous on $\eta([-1,2])$.
\end{proof}

With these preparations we discuss in the following two subsections pointwise estimates on $G_K$. In the last subsection we present an estimate on the integral of $|G_K(t)|^2$ and a simple integral identity. All these estimates will be useful in later sections.

\subsection{The case \texorpdfstring{$x  \le  \frac 12 $}{x ^^e2 ^^89 ^^a4 1/2}}
In this subsection, we will establish the necessary understanding of \cite[\href{http://dlmf.nist.gov/18.15.E19}{(18.15.19)}]{NIST:DLMF}. We would like to establish an asymptotic expansion up to a sufficiently small error term.

As \cite[\href{http://dlmf.nist.gov/18.15.E19}{(18.15.19)}]{NIST:DLMF} reduces the Laguerre polynomial asymptotic to the Bessel functions $J_\alpha$, we take a look at their asymptotics.
\begin{proposition} \label{Bessel asymptotic}
For $s\ge 0$, we have
\begin{align} \label{BesselJ asymp eq}
\sqrt s J_{\alpha}(s) = 
\begin{cases} 
	\sqrt {2/ \pi} \cos( s - \alpha \pi /2  - \pi /4)  & \text{ if } \alpha = \pm 1/2 \,, \\
	\sqrt {2/ \pi} \cos( s - \alpha \pi /2  - \pi /4)  +\mathcal O(1/(1+s)) & \text{ if } \alpha =1\,.
\end{cases}
\end{align}
\end{proposition}
\begin{proof}
The first part is just \cite[\href{http://dlmf.nist.gov/10.16.E1}{(10.16.1)}]{NIST:DLMF} and the second part is \cite[\href{http://dlmf.nist.gov/10.17.i}{(10.17.3)}]{NIST:DLMF}.
\end{proof}

\begin{lemma} \label{?? for x<1/2}
For $0 \le x \le 1/2$ and $\alpha \in \{\pm 1/2, 1\}$, we have
\begin{align}
F_K^{(\alpha)}(x)=\sqrt{ \nu \xi(x)} J_\alpha(\nu \xi(x)) + \mathcal O\left(1 /K\right).
\end{align}
In particular, for $\alpha=\pm 1/2$, we get
\begin{align}
F_K^{(\alpha)}(x)=\sqrt {2/ \pi} \cos( \nu \xi(x) - \alpha \pi /2  - \pi /4)  + \mathcal O\left(1 /K\right).
\end{align}
Furthermore, if $\alpha=1$, as $x \to 0$, the error term is at most $\mathcal O(\sqrt K x^{3/4})$.
\end{lemma}
\begin{proof}
The asymptotic expansion \cite[\href{http://dlmf.nist.gov/18.15.E19}{(18.15.19)}]{NIST:DLMF} with $\nu=4K+2 (\alpha-1)$ and $M=1$ can be solved for $F_K^{(\alpha)}$ and provides us with
\begin{align}
F_K^{(\alpha)}(x) = \sqrt \nu \left( \sqrt{\xi(x)}  J_\alpha(\nu \xi(x) )+  {\xi(x)^{-\frac 1 2}} J_{\alpha+1} (\nu \xi(x) ) \frac{B_0(\xi(x))}{\nu} + \sqrt {\xi(x)} \,\mathrm{env}J_\alpha(\nu \xi(x)) \mathcal O\left( \frac 1 \nu \right) \right).
\end{align}
We are left to establish upper bounds for the last two summands. We need to bound the coefficient $B_0$. For that, we recall \eqref{xi in eta} and \eqref{eta at 0eq} to observe
\begin{align}
\xi(x)&=\eta(\sqrt x)= \sqrt x+ \mathcal O\left(x^{\frac 3 2} \right). \label{xi asymptotic eq}
\end{align}
Hence, the coefficient described in \cite[\href{http://dlmf.nist.gov/18.15.E20}{(18.15.20)}]{NIST:DLMF},
\begin{align}
B_{0}(\xi(x))\coloneqq -\frac{1}{2}\left(\frac{1-4\alpha^2}{8}+\xi(x)\left(\frac{1-x}{x}\right)^{\frac{1}{2}}\left(\frac{4 \alpha^2-1}{8}+\frac{1}{4}\frac{x}{1-x}+\frac{5}{24}\left(\frac{x}{1-x}\right)^{2}\right)\right),
\end{align}
satisfies $\lvert B_0(\xi(x)) \rvert \le C x\le C \xi(x)^2$ for $0 \le x \le \frac 12$ with a constant $C$ depending on $\alpha$.

 Using \cite[\href{https://dlmf.nist.gov/2.8\#iv.p5}{(2.8.32--34)}]{NIST:DLMF} and  \cite[\href{http://dlmf.nist.gov/10.17.i}{(10.17.3--4)}]{NIST:DLMF}, we see that for $s \ge 1$, we have
\begin{align} \label{BesselJ upper bound large argument}
\max\{ \lvert J_\alpha(s) \rvert,\lvert J_{\alpha+1}(s) \rvert, \lvert \text{env}J_{\alpha}(s) \rvert  \}  &\le C /\sqrt s 
\end{align}  
and for $ 0\le s \le 1$, \cite[\href{http://dlmf.nist.gov/10.2.E2}{(10.2.2)}]{NIST:DLMF} tells us that
\begin{align} \label{BesselJ upper bound small argument}
\max\{\lvert J_\alpha(s) \rvert,  \lvert J_{\alpha+1}(s) \rvert, \lvert \text{env}J_{\alpha}(s) \rvert  \}  &\le C s^\alpha \,.
\end{align}
Thus, as $\alpha \ge -1/2$, \eqref{BesselJ upper bound large argument} holds for any $s \ge 0$. Hence, for any $0 \le x \le 1/2$, we observe
\begin{align}
\sqrt \nu & \left( {\xi(x)^{-\frac 1 2}} J_{\alpha+1} (\nu \xi(x) ) \frac{B_0(\xi(x))}{\nu} + \sqrt {\xi(x)}\, \text{env}J_\alpha(\nu \xi(x)) \,\mathcal O\left( \frac 1 \nu \right) \right) \\
&= \mathcal O \left(   {\xi(x)^{-\frac 1 2}}  \frac {\sqrt \nu}  {\sqrt {\nu \xi(x)}}  \frac{\xi(x)^2}{\nu} + \sqrt {\nu \xi(x)} \frac 1 {\sqrt{\nu \xi(x)}}  \frac 1 \nu     \right) \\
& =\mathcal O((1+ \xi(x))/\nu)= \mathcal O(1/K)\,.
\end{align}
This proves the first claim. The second one follows by \eqref{BesselJ asymp eq}. Let $\alpha=1$. If $\nu\xi(x)\ge 1$, we have $\sqrt x \ge C/K$, which implies $\sqrt Kx^{3/4} > C/K$ and thus we may assume $\nu \xi(x) \le 1$. Then, we get 
\begin{align}
\sqrt \nu &\left( {\xi(x)^{-\frac 1 2}} J_{\alpha+1} (\nu \xi(x) ) \frac{B_0(\xi(x))}{\nu} + \sqrt {\xi(x)}\, \text{env}J_\alpha(\nu \xi(x)) \,\mathcal O\left( \frac 1 \nu \right) \right) \\
& = \mathcal O \left(   {\xi(x)^{-\frac 1 2}}  \sqrt \nu (\nu \xi(x))^1   \frac{\xi(x)^2}{\nu} + \sqrt {\nu \xi(x)} (\nu \xi(x))^1  \frac 1 \nu     \right) \\
&= \mathcal O(\sqrt \nu \xi(x)^{3/2} )= \mathcal O(\sqrt K x^{3/4})\,.
\end{align}
This was the third claim.
\end{proof}

This brings us to the first main result on the translation-invariant part of the integral kernel of the Fermi projection in case the argument $t$ is small compared to $K$.
 
\begin{theorem} \label{GK for t<1/2 K}
We have for $K\in\mathbb N$, we have
\begin{align} 
G_K(t)= \begin{cases} \frac{ \cos(\omega_K(t) - 3 \pi/4) } {4 \sqrt 2  \pi^{\frac 3 2} t^{\frac 3 2} \left( 1- \frac {t^2} {K^2} \right)^{\frac 1 4}   }+\mathcal O \left( \frac 1 {t^{\frac 5 2} }\right) \quad &\text{ if } 1 \le t \le \sqrt{\frac 1 2} K\,, \\
\frac{  J_1(4t)} {4\pi t} + \mathcal O \left( \frac{(1+ t)^{\frac{3}{2}}}{K^2} + \frac 1 {K (1+t)^{\frac 32} } \right) & \text{ if } 0 \le t \le K^{2/3}\,,
\end{cases}
\end{align}
where 
\begin{align}%\label{omegaK}
\omega_K(t) \coloneqq 4K \xi\left(\frac{t^2}{K^2}\right) =4K \eta(t/K)=2 \left( t \sqrt {1- \frac{t^2} {K^2} } + K \operatorname{arcsin}\left(\frac tK \right) \right).
\end{align}
\end{theorem}

The regions mentioned in this theorem overlap for $K>1$.  The second case provides the way, in which the kernel converges to the free kernel, that is, to the integral kernel of the two-dimensional Laplacian.
\begin{proof} As $\nu=4K$, we have
\begin{align}\label{omegaK} 
\omega_K(t)=\nu \xi(t^2/K^2)\,.
\end{align}
Since $t \le K/\sqrt 2$, we have $1/2 \le 1- t^2/K^2 \le 1$. We use $\nu =4K$, \eqref{G to F} and \autoref{?? for x<1/2} and see
\begin{align}
G_K(t) &= \frac 1 {8 \pi t^{3/2}  \left( 1- \frac {t^2} {K^2} \right) ^{1/4}  } \left( \sqrt{ \omega_K(t)} J_1(\omega_K(t)) + \mathcal O\left(\min\{ 1 /K ,  \sqrt K (t^2/K^2)^{3/4} \} \right) \right)\\
&=  \frac {\sqrt{ \omega_K(t)} J_1(\omega_K(t)) } {8 \pi t^{3/2}  \left( 1- \frac {t^2} {K^2} \right) ^{1/4}  }  + \mathcal O\left(\min\{ 1 /(Kt^{3/2})  ,  1 /K \} \right) \\
&=  \frac {\sqrt{ \omega_K(t)} J_1(\omega_K(t)) }  {8 \pi t^{3/2}  \left( 1- \frac {t^2} {K^2} \right) ^{1/4}  }  + \mathcal O\left(1 /(K(1+t^{3/2}))   \right).
\end{align}
As $K>t$, we have $1 /(K(1+t^{3/2}))  \le 1/t^{5/2}$. Thus, the error term is currently good enough for both results. Due to \eqref{omegaK} and \eqref{xi asymptotic eq}, we have
\begin{align}
\omega_K(t)=4t + \mathcal O(t^3/K^2). \label{omegaK asymp eq1}
\end{align}
We also have $\omega_K(t) \ge Ct$, see \eqref{eta lower linear bound}, which implies $1/(1+\omega_K(t)) \le C/(1+t)$. 

We consider the case $t \ge 1$. According to \autoref{Bessel asymptotic}, we observe
\begin{align}
G_K(t) =& \frac {\sqrt{2 /\pi} \cos(\omega_K(t) - 3 \pi/4 ) + \mathcal O( 1/(1+\omega_K(t))) }  {8 \pi t^{3/2}  \left( 1- \frac {t^2} {K^2} \right) ^{1/4}  }  + \mathcal O\left(1 /(K(1+t^{3/2}))   \right) \\
=& \frac{ \cos(\omega_K(t) - 3 \pi/4) } {4 \sqrt 2  \pi^{\frac 3 2} t^{\frac 3 2} \left( 1- \frac {t^2} {K^2} \right)^{\frac 1 4}   }+\mathcal O \left( \frac 1 {t^{\frac 5 2} }\right).
\end{align}
This finishes the proof of the first claim.

We are left with the case $0 \le t \le K^{\frac 2 3}$. In this case, we want to fully eliminate the dependency of the leading term on $K$. We consider the function $s \mapsto \sqrt s J_1(s)$ and want to study its derivative. For that, we note that $J_1'(s)=-J_2(s)+J_1(s)/s$ (see \cite[\href{https://dlmf.nist.gov/10.6\#E2}{(10.6.2)}]{NIST:DLMF}). We get
\begin{align}
\left( \sqrt s J_1(s) \right)' =  \frac{ J_1(s) }{2 \sqrt s} + \sqrt s J_1'(s) =\frac{ J_1(s) }{2 \sqrt s}  -\sqrt s J_2(s) + \frac{ J_1(s) }{ \sqrt s} = \frac{ 3J_1(s)-2s J_2(s) }{ 2\sqrt s}   \,.
\end{align}
With \eqref{BesselJ upper bound large argument} and \eqref{BesselJ upper bound small argument}, we see that $\left( \sqrt s J_1(s) \right)' $ is bounded independently of $s$. Thus, using \eqref{omegaK asymp eq1}, we arrive at 
\begin{align}
\sqrt {\omega_K(t) } J_1(\omega_K(t) )= \sqrt{4t} J_1(4t) + \mathcal O( t^3/K^2)\,.
\end{align}
Using \eqref{BesselJ upper bound large argument} and \eqref{BesselJ upper bound small argument} again, we also see that
\begin{align} \label{BesselJ1 upper bound eq1}
\frac{ \lvert J_1(4t) \rvert } t \le \frac C {(1+t)^{3/2}}\,.
\end{align}
Let us now deal with the denominator. We observe
\begin{align}
\frac 1 {\left( 1- t^2/K^2\right)^{1/4} } =& 1 + \mathcal O(t^2/K^2)\,.
\end{align}
 Combining these results, we arrive at
\begin{align}
G_K(t)&=\frac {\sqrt{ \omega_K(t)} J_1(\omega_K(t)) }  {8 \pi t^{3/2}  \left( 1- \frac {t^2} {K^2} \right) ^{1/4}  }  + \mathcal O\left(1 /(K(1+t^{3/2}))   \right) \\
&=\frac{ \sqrt{4t} J_1(4t) + \mathcal O(t^3/K^2) } { 8 \pi t^{3/2} } \left(1+ \mathcal O(t^2/K^2) \right) +\mathcal O\left(1 /(K(1+t^{3/2}))   \right) \\
&= \frac { J_1(4t)} { 4 \pi t} \left( 1+ \mathcal O(t^2/K^2) \right) +\mathcal O \left( \frac { t^{3/2} } {K^2} + \frac 1 {K(1+t)^{3/2}} \right) \\
&= \frac { J_1(4t)} { 4 \pi t}  +\mathcal O \left(\frac{t^2}{K^2 (1+t)^{3/2}}+ \frac { t^{3/2} } {K^2} + \frac 1 {K(1+t)^{3/2}} \right) \\
&= \frac { J_1(4t)} { 4 \pi t}  +\mathcal O \left(\frac { t^{3/2} } {K^2} + \frac 1 {K(1+t)^{3/2}} \right).
\end{align}
The first part of the supposed error term is only smaller than the order of the main term $\mathcal O(1/t^{\frac 32} )$, if $t< K^{\frac 23 }$. 
\end{proof}

 \subsection{The case \texorpdfstring{$x  \ge   \frac 12 $}{x ^^e2 ^^89 ^^a5 1/2}}
 In this section, we intend to understand \cite[\href{http://dlmf.nist.gov/18.15.E22}{(18.15.22)}]{NIST:DLMF} sufficiently well. 
 
 In accordance with \cite[\href{http://dlmf.nist.gov/18.15.E21}{(18.15.21)}]{NIST:DLMF}  (see \eqref{eta expicit}, \eqref{def zeta}, and \eqref{3.55}) we have
 \begin{align}
 \zeta(x) = \begin{cases} 
 	- \left( \frac 3 4 \left ( \arccos(\sqrt{x}) - \sqrt{x-x^2}\right) \right)^{\frac  2 3} \quad & \text{ if } 0\le x \le 1\,, \\
 	\left( \frac 3 4  \left ( \sqrt{x^2-x}- \operatorname{arccosh}(\sqrt {x}) \right) \right)^{\frac 2 3} \quad & \text{ if } x\ge 1\,.
 	\end{cases} 
\end{align}
Therefore, $\zeta$ is negative on $(0,1)$ and positive on $(1,\infty)$ with a (unique) zero at $1$.

We first want to put this in relation to $\xi$ and present the following simple  
\begin{lemma} \label{cos lem}
For $0 \le x \le 1$ and  $\nu = 4K +2(\alpha-1)$ with $K\in\mathbb N$ and $\alpha\in\R$, we have 
\begin{align}
\cos(\nu \xi(x) -\alpha \pi/2 - \pi/4 )= (-1)^{K-1} \cos( 2\nu /3  (-\zeta(x))^{3/2}-\pi /4)\,.
\end{align}
\end{lemma}
\begin{proof}
We first consider
\begin{align}
\xi(x) + 2/3 (-\zeta(x))^{3/2} &= \frac 12 \left( \sqrt{x-x^2} + \arcsin(\sqrt{x})\right) + \frac 1 2 \left ( \arccos(\sqrt{x}) - \sqrt{x-x^2}\right) \\
&= \frac 1 2 \left( \arcsin(\sqrt{x}) +  \arccos(\sqrt{x})\right) \\
&= \frac \pi 4\,.\label{3.55}
\end{align}
We proceed to add the two arguments inside the cosines. Thus, we observe
\begin{align}
\nu \xi(x) -\alpha \pi/2 - \pi/4 + 2\nu /3  (-\zeta(x))^{3/2}-\pi/4 & = (4K+ 2(\alpha-1)) \pi/4  - \alpha \pi/2 - \pi /2 \\
& = (K-1)\pi\,.
\end{align}
Let $s,t$ be the arguments of the cosines. We just showed $s+t=(K-1) \pi$. Thus, $t=(K-1) \pi-s$, which implies $\cos(t)=(-1)^{K-1} \cos(-s)= (-1)^{K-1} \cos(s)$, which is the claim.
\end{proof}

The following expansion to second order is only needed to prove \autoref{F0 bounded}. Beyond that, the rougher estimates in \autoref{zeta increasing} are sufficient. 

\begin{lemma} \label{zeta at 1}
The function $\zeta$ satisfies the expansion
\begin{align}
\zeta(1+s)= \frac {s}{2^{\frac 2 3}} \left(1-s/5\right) + \mathcal O(s^3)
\end{align}
for $\lvert s \rvert \le \frac 12 $. 
\end{lemma}

\begin{proof}

We recall \eqref{eta at 1eq}, that is, 
\begin{align}
3/2 (\pi/4- \eta(1-h))=\sqrt 2 \left( \sqrt h^3 - \frac{3 }{20} \sqrt h^5 + \mathcal O(|h|^{7/2})\right),
\end{align}
which implies
\begin{align}
\zeta((1-h)^2)=-2^{1/3} h\left(1 - \frac h {10} + \mathcal O(h^2) \right).
\end{align} 
If we insert $h=1-\sqrt{1+s}=-s/2+s^2/8+\mathcal O(s^3)$, we arrive at the claimed expansion.
\end{proof}

We now need to study $\zeta$ for $x$ further away from $1$.

\begin{lemma}  \label{zeta increasing}
On $\mathbb R^+$, the function $\zeta$ is strictly increasing. Furthermore, there are positive constants $C_1,C_2$, such that for any $x> \frac 32$, we have
\begin{align} \label{zeta increasing eq1}
C_1 \le \frac {\zeta(x)^{\frac 3 2} } x \le C_2 \,,
\end{align} 
and for $x \in \left( \frac 12, \frac 32 \right)$, we have
\begin{align} \label{zeta increasing eq2}
C_1 \le \frac{ \zeta(x)} {x-1} \le C_2\,.
\end{align}
\end{lemma}
\begin{proof}
We know that $\zeta(1)=0$  and $\zeta'(1)=2^{- \frac 2 3} >0$ (see \autoref{zeta at 1}). Then, we would like to estimate the derivative of $x \mapsto \lvert \zeta(x)\rvert ^{\frac 3 2}$ for $x\not=1$ by using the derivative of $\eta$, see \eqref{def of eta}. For $0<x<1$, it is given by 

\begin{align}
\left((- \zeta(x))^{\frac 3 2} \right)' = \frac 3 2  \left (  \frac{\pi}{4} - \eta(\sqrt{x})\right)' = -\frac 34 \eta'(\sqrt{x})\frac{1}{\sqrt{x}} = -\frac {3}{4} \sqrt{1-x} \frac{1}{\sqrt{x}}  < 0 \,.
\end{align}
Similarly, for $x>1$, due to \eqref{zeta for x>0 good sign}, it is given by
\begin{align}
\left( \zeta(x)^{\frac 3 2} \right)' =  \frac 3 2  \,\left (\Im \left( \eta(\sqrt{x}) \right) \right)'  = \frac{3}{2} \,\Im ( \eta'(\sqrt{x}) )\frac{1}{2\sqrt{x}}  =  \frac {3}{4} \,\sqrt{x-1}\, \frac{1}{\sqrt{x}}    > 0 \,.
\end{align}
Altogether this implies that $\zeta$ is strictly increasing and that \eqref{zeta increasing eq2} holds.

The last equation also shows that for $x\ge 3/2$,
\begin{equation} \frac{\sqrt{3}}{4} \le \left( \zeta(x)^{\frac 3 2} \right)' = \frac{3}{4}\sqrt{1-x^{-1}} \le \frac{3}{4}\,.
\end{equation}
As $\zeta$ is strictly increasing and thus $\zeta(\frac 3 2)>0$, we have proved \eqref{zeta increasing eq1}. This is the end of the proof.
\end{proof}

Next, we  consider the coefficient described in \cite[\href{http://dlmf.nist.gov/18.15.E23}{(18.15.23)}]{NIST:DLMF}, and define the function $F_0$,
\begin{align}
F_0(x) \coloneqq - \frac 5 {48 \zeta(x)^2} + \sqrt{ \frac{x-1} {x\zeta(x)} } \left( \frac 12 - \frac 1 8 - \frac 1 4 \frac x {x-1} + \frac 5 {24} \left( \frac x {x-1} \right)^2 \right), \quad x\in (0,\infty)\setminus\{1\}\,,
\end{align}
and $F_0(1)\coloneqq \lim_{x\to1} F_0(x)$, see the next lemma. This function has, of course, nothing to do with the function $F_K^{(\alpha)}$.

\begin{lemma} \label{F0 bounded}
There is a constant $C>0$ such that $F_0(x)$ for any $x> \frac 1 2$ satisfies
\begin{align}
\lvert F_0(x) \rvert \le C\,.
\end{align}
\end{lemma}
\begin{proof}
Due to \autoref{zeta increasing} and $\zeta(1)=0$, we know that $\zeta$ has a unique zero at $x=1$. Let $A \subset [\frac 12 , \infty)$ be a closed subset with $1 \not \in A$.  We see that $F_0$ is bounded on $A$, as the functions $x \mapsto \frac 1 {\zeta(x)}, x \mapsto \frac x {x-1}, x \mapsto \frac {x-1}x$ are all bounded on $A$. We employ \autoref{zeta at 1} and see that as $s \to 0$, we have
\begin{align}
\frac 1 {\zeta(1+s)^2} =& \frac 1 { \left( 2^{-\frac 2 3} s \left(1-\frac s 5\right)  + \mathcal O(s^3) \right)^2 } 
= \frac {2^{\frac 4 3} } {s^2}  \frac 1 {\left(1-\frac s 5  + \mathcal O(s^2) \right)^2}
= \frac {2^{\frac 4 3} } {s^2} \left( 1 + \frac {2s} 5 \right) + \mathcal O(1)\,.
\end{align}
Similarly, we observe
\begin{align}
\sqrt{ \frac s {(1+s) \zeta(1+s) }} &= \sqrt{ \frac s {(1+s)  \left( 2^{-\frac 2 3} s \left(1-\frac s 5\right)  + \mathcal O(s^3) \right)}} \\
&=   \frac 1 {\sqrt{(1+s)  \left( 2^{-\frac 2 3}  \left(1-\frac s 5\right)  + \mathcal O(s^2) \right)}} = 2^\frac 1 3 \left( 1 - \frac {2s} 5 \right) +\mathcal O(s^2)\, .
\end{align}
We conclude
\begin{align}
2^{-\frac  13 }& F_0(1+s)\\
&=2^{-\frac 1 3} \left( - \frac 5 {48 \zeta(1+s)^2} + \sqrt{ \frac{s} {(1+s)\zeta(1+s)} } \left( \frac 12 - \frac 1 8 - \frac 1 4 \frac {1+s} {s} + \frac 5 {24} \left( \frac {1+s} {s} \right)^2 \right) \right) \\
&= - \frac 5 {24s^2} - \frac 1 {12s} + \left(1- \frac{2s} 5 \right) \left( - \frac 1 {4s} + \frac {5(1+2s)}{24s^2} \right) + \mathcal O(1) \\
&=-   \frac 5 {24s^2} - \frac 1 {12s}  -\frac 1 {4s} + \frac 5 {24s^2} + \frac{ 5} {12s} -\frac 1 {12s} + \mathcal O(1) = \mathcal O(1)\,\,. 
\end{align}
Thus, $F_0$ can be defined at $1$ as a limit. Moreover,  we have shown that $F_0$ is bounded on some neighbourhood of $1$ and on any closed subset of $[\frac 1 2, \infty)$ not containing $1$, which implies that $F_0$ is bounded on $[\frac 12, \infty)$, which was the claim.
\end{proof}

As \cite[\href{http://dlmf.nist.gov/18.15.E22}{(18.15.22)}]{NIST:DLMF} reduces our $F_K^{(\alpha)}$ to the Airy function, its derivative and its envelope, we should now take a look at the asymptotics of these functions.

\begin{proposition} \label{Airy asymp}
Let $s >0$ and let $\mathcal A (s)\coloneqq \frac 2 3 s ^{3/2}$ (see \cite[\href{http://dlmf.nist.gov/9.7.E1}{(9.7.1)}]{NIST:DLMF}). Then, there are (positive) constants $C$ (which may vary from line to line) such that  
\begin{align}
\Ai(-s) s^{1/4} &= 1/\sqrt \pi \cos(\mathcal A(s) -\pi/4)+\mathcal O(1/(1+\mathcal A(s)))\,, \label{Ai(-s)}
\\ \lvert \Ai'(-s) s ^{-1/4} \rvert &\le C(1+s^{-1/4}) \,, \label{Ai'(-s)}
\\ 0\le \operatorname{env} \Ai(-s)s^{1/4}  &\le C \,, \label{envAi(-s)}
\\ \lvert \Ai(s) \rvert s^{1/4} &\le C \exp(-\mathcal A(s))\, , \label{Ai(s)}
\\ \lvert \Ai'(s) s ^{-1/4} \rvert &\le C \exp(-\mathcal A(s)) (1+s^{-1/4}) \,,\label{Ai'(s)}
\\ 0\le \operatorname{env} \Ai(s)s^{1/4}  &\le C \exp(-\mathcal A(s)) \, .\label{envAi(s)}
\end{align}
\end{proposition}
\begin{proof}
All of these follow from  \cite[\href{http://dlmf.nist.gov/9.7.ii}{(9.7.5--11)}]{NIST:DLMF} for the asymptotics of the Airy functions $\Ai,\operatorname{Bi}$ and \cite[\href{http://dlmf.nist.gov/2.8\#iii.p4}{(2.8.19--21)}]{NIST:DLMF} for the definition of the envelope $\text{env} \Ai$, which can be expressed in $\Ai$ and $\operatorname{Bi}$.
\end{proof}

We are now prepared to take on \cite[\href{http://dlmf.nist.gov/18.15.E22}{(18.15.22)}]{NIST:DLMF}. 

\begin{lemma} \label{Laguerre to Airy}
Let $\alpha\in \{ \pm 1/2 ,1 \}$ and $\nu = 4K+2(\alpha-1)$. For $1/2 \le x \le 1$, we have
\begin{align}
F_K^{(\alpha)}(x)&= \sqrt{2/\pi} \cos(\nu \xi(x)-\alpha \pi /2 - \pi/4) + \mathcal O(1/(1+K(1-x)^{3/2} )\,,
\end{align}
for $1 \le x \le3/2$, it holds
\begin{align}
\left \lvert F_K^{(\alpha)}(x) \right \rvert \le C\exp(-\beta K (x-1)^{3/2})\,,
\end{align}
and in the case $x>3/2$, we get
\begin{align}
\left \lvert F_K^{(\alpha)}(x) \right \rvert \le  C \exp(-\beta K x)
\end{align}
for some constants $\beta>0$ and $C<\infty$.
\end{lemma}
\begin{proof}
We note that $\zeta(x)/(x-1)=\lvert \zeta(x) \rvert / \lvert x-1 \rvert$ due to \autoref{zeta increasing}. Solving \cite[\href{http://dlmf.nist.gov/18.15.E22}{(18.15.22)}]{NIST:DLMF} for $F_K^{(\alpha)}(x)$ provides us with
\begin{align}
F_K^{(\alpha)}&(x)(-1)^{K-1} \sqrt 2^{-1}\\
&=  \sqrt \nu  \lvert \zeta(x) \rvert^{\frac 1 4}\left( \frac {\Ai(\nu^{2/3} \zeta(x))}{ \nu^{1/3}} + \frac{ \Ai'( \nu^{2/3} \zeta(x)) }{\nu^{5/3} }F_0(x) + \text{env}\Ai(\nu^{2/3} \zeta(x) ) \mathcal O(\nu^{-4/3}) \right) \\
&=\left( \nu^{2/3} \lvert \zeta(x) \rvert \right)^{1/4} \Ai(\nu^{2/3} \zeta(x)) 
\\
&+ \mathcal O \left(\frac{\Ai'( \nu^{2/3} \zeta(x)) }{\left( \nu^{2/3} \lvert \zeta(x) \rvert \right)^{1/4}} \frac{ \sqrt{\lvert \zeta(x) \rvert }} {\nu}  + \left( \nu^{2/3} \lvert \zeta(x) \rvert \right)^{1/4}  \text{env}\Ai(\nu^{2/3} \zeta(x)) \nu^{-1}      \right).
\end{align}
We have regrouped the terms such that we get the left-hand sides of the formulas in \autoref{Airy asymp}. 

For $1/2\le x \le 1$, due to \autoref{zeta increasing}, we get $-\infty  < \zeta(1/2) \le \zeta(x)\le 0$ and thus we can bound the contribution of the $\Ai'$ and $\text{env}\Ai$ term by $C/K$ using \eqref{Ai'(-s)} and \eqref{envAi(-s)}. For the $\Ai$ term, we use \eqref{Ai(-s)}  with $s = -\nu ^{2/3}\zeta(x) > 0$ and \autoref{cos lem}, which yield
\begin{align}
F_K^{(\alpha)}(x) &=\sqrt 2  (-1)^{K-1} \left( \nu^{2/3} \lvert \zeta(x) \rvert \right)^{1/4} \Ai(\nu^{2/3} \zeta(x)) + \mathcal O (1/K) \\
&=\sqrt {2/ \pi }   (-1)^{K-1}  \cos(2\nu/3 (-\zeta(x))^{3/2} -\pi /4 ) + \mathcal O(1/(1+K\lvert \zeta(x) \rvert^{3/2})  + 1/K) \\
&= \sqrt {2/ \pi } \cos(\nu \xi(x) -\alpha \pi/2 - \pi/4 ) + \mathcal O(1/(1+K(1-x)^{3/2}))\,.
\end{align}

For $x \ge 1$, we use \eqref{Ai(s)}, \eqref{Ai'(s)} and \eqref{envAi(s)} to see
\begin{align}
\left \lvert F_K^{(\alpha)}(x) \right \rvert &\le C \exp(- \mathcal A ( \nu^{2/3} \zeta(x))) \left( 1+ \frac {\sqrt {\lvert \zeta(x) \rvert} } {\nu} + \frac 1 \nu \right) \\
&\le C \exp(- 2\nu /3 \,\zeta(x)^{3/2} ) ( 1 +\sqrt { \zeta(x)} )\\
&\le C \exp(- \nu/2  \,\zeta(x)^{3/2} ) \sup_{s>0} \left( \exp(- \nu s^3/6) \, (1+s) \right) .
\end{align}
As $\nu=4K+2(\alpha-1) \ge 4-3=1$, the supremum at the end is bounded independently of $\nu$. To get a better understanding of the exponential decay, we need \autoref{zeta increasing}. 
For $1 \le x \le 3/2$, we get
\begin{align}
\left \lvert F_K^{(\alpha)}(x) \right \rvert &\le C \exp(- \nu/2\, \zeta(x)^{3/2} )  \\
&\le C \exp(-\nu/2 \, C_1^{3/2} (x-1)^{3/2} ) \\
&\le  C \exp(-\beta K (x-1)^{3/2})\,,
\end{align}
where $\beta>0$ is some constant and where we used $\nu \ge 4K-3\ge K$.

For $3/2 \le x < \infty$, we use the other estimate in \autoref{zeta increasing} and see
\begin{align}
\left \lvert F_K^{(\alpha)}(x) \right \rvert &\le C \exp(- \nu/2 \,\zeta(x)^{3/2} )  \\
&\le  C \exp(-\nu/2 \, C_1 x ) \\
&\le  C \exp(-\beta K  x)\,,
\end{align}
with a possibly different $\beta>0$.

\end{proof}

This leads to the following second main result on $G_K(t)$ in case the argument $t$ is large compared to $K$,

\begin{corollary} \label{GK for t>1/2 K} There exists a constant $\beta>0$ and a constant $C<\infty$ independent of $K$ such that we have the following estimates for the function $G_K$ from \eqref{G to F},
\begin{align}
\lvert G_K(t) \rvert \le   \begin{cases} C\exp(- \beta t) \quad &\text{ if } t> \sqrt {\frac 3 2} K\,, \\ \frac C {K^{\frac 3 2}  \left \lvert 1- \frac t K \right \rvert^{\frac 1 4}}   \quad &\text{ if } \sqrt{ \frac 12} K  \le t \le \sqrt {\frac 3 2} K\,.
 \end{cases}
\end{align}
\end{corollary}
\begin{proof}
We recall $x=t^2/K^2$ and  \eqref{G to F}, which states that
\begin{align}
G_K(t) = \frac 1 { 8 \pi t^{3/2} \left\lvert 1- \frac {t^2} {K^2} \right\rvert ^{1/4}  } F_K^{(1)} \left( \frac{t^2}{K^2} \right).
\end{align}
If $t> \sqrt{3/2} K$, we have $x>3/2$. Thus, the fraction in front is bounded by $C/K^{3/2}\le C$ and \autoref{Laguerre to Airy} tells us that 
\begin{align}
F_K^{(1)}(t^2/K^2) \le C \exp(-\beta K x ) \le C \exp(-\beta K\sqrt x)=C \exp(-\beta t)\,.
\end{align}
This finishes the case $t>\sqrt{3/2} K$. In the case $\sqrt {1/2} K \le t \le \sqrt{3/2}K$, \autoref{Laguerre to Airy} implies that $\lvert F_K^{(1)}(t^2/K^2) \rvert \le C$ and thus, we can bound 
\begin{align}
\lvert G_K(t) \rvert \le \frac { C } {t^{3/2} \left\lvert 1- \frac {t^2} {K^2} \right\rvert ^{1/4}  } \le \frac{C} {K^{3/2} \left\lvert 1- \frac {t} {K} \right\rvert ^{1/4}  } .
\end{align}
\end{proof}

\subsection{An integral bound and an identity on the kernel $G_K$}

We start with an integral estimate, which will help to calculate some Hilbert--Schmidt norms.
\begin{lemma} \label{GK square int}
There are constants $C<\infty$ and $\beta>0$, independent of $K$ such that for any $R\ge 0$
\begin{align}
\int_R^\infty \lvert G_K(t) \rvert^2 t \,\mathrm d t \le \begin{cases} \frac C {1+R}   \quad &\text{ if } R< 2K\,,  \\ C \exp(-\beta R) & \text{ if } R \ge 2K\,. \end{cases} 
\end{align}
\end{lemma}
\begin{proof}
Combining \autoref{GK for t>1/2 K} and \autoref{GK for t<1/2 K}, we obtain the universal upper bound
\begin{align}
\lvert G_K(t) \rvert \le \begin{cases} \frac C { (1+t)^{3/2}  \left \lvert 1- \frac t K \right \rvert ^{1/4} } \quad &\text{ if } t < 2K \,, \\ C \exp(-\beta t) & \text{ if } t\ge 2K\,. \end{cases}
\end{align}
For $R\ge 2K$, we simply observe
\begin{align}
\int_R^\infty \lvert G_K(t) \rvert^2 t\, \mathrm d t \le C \int_R^\infty \exp(-2\beta t) t\, \mathrm dt \le C \exp(-\beta R)\,.
\end{align}
First, we assume $1 < R <2K$ and estimate
\begin{align}
\int_R^{2K} \lvert G_K(t) \rvert ^2 t \,\mathrm d t &\le  \int_R^{2K}   \frac {C t } { t^{3}  \left \lvert 1- \frac t K \right \rvert ^{1/2} }   \,\mathrm d t  \\ 
&= \int_{R/K} ^2    \frac C { (Ks)^{2}  \left \lvert 1-  s \right \rvert ^{1/2} }   K \,\mathrm d s \\
&\le  \frac C K \int_{R/K} ^2 \frac 1 {s^2 \lvert 1-s^2 \rvert^{1/2} }\,\mathrm d s \\
&= \frac C K \left( \frac {\sqrt 3} 2  + \frac {\sqrt { 1- R^2/K^2}} { R/K} \right) \le \frac C K + \frac C R \le \frac C {1+R}\,.
\end{align}
Finally, for $0\le R \le 1$, it suffices to see that
\begin{align}
\int_0^1  \lvert G_K(t) \rvert ^2 t \,\mathrm d t \le \int_0^1 C t \,\mathrm d t \le C\,.
\end{align}
This concludes the proof, as we can split the integral on $(R,\infty)$ into the at most three parts $(R,1), (\max\{1,R\}, 2K)$ and $(\max\{2K,R\},\infty)$ and the upper bound is decreasing in $R$.
\end{proof}

We finish this subsection with a simple and yet useful identity of the localised Fermi projection.

\begin{lemma} \label{easy box estimate}
For every $E \subset \mathbb R^2$ measurable and bounded, we have
\begin{align}
\left \lVert 1_{E} P_K \right \rVert_2^2=\operatorname{tr} 1_E P_K 1_E  = \frac {\lvert E \rvert } {2 \pi}\,.
\end{align}
\end{lemma}
\begin{proof}
This follows from 
\begin{align}
 \operatorname{tr} 1_E P_K 1_E &=\int_E  P_K(x,x)\,\mathrm d x  =\int_E G_K(0) \,  \mathrm d x =\lvert E \rvert \frac 1 {2 \pi K} \Ln_{K-1}^{(1)}(0)= \frac {\lvert E \rvert } {2 \pi}\,,
\end{align}
where we used $\Ln_{K-1}^{(1)}(0)=(K-1)+1=K$, see \cite[\href{http://dlmf.nist.gov/18.6.E1}{(18.6.1)}]{NIST:DLMF} with $(2)_{K-1}=((K-1)+1)!$, see \cite[\href{http://dlmf.nist.gov/5.2.iii}{(5.2.iii)}]{NIST:DLMF}.
\end{proof}

\section{On the particle number fluctuations}\label{section: pf}

This section is devoted to the study of the asymptotic expansion of the trace of $f(P_K(L\Lambda))$ when $f$ is a quadratic function. Since we assume $f(0)=f(1)=0$, we  may restrict to $f(t)=t(1-t)$. It is related to  the fluctuations of the local particle number in the ground state and hence of physical interest. 

\begin{theorem} \label{pf thm}
Let $\Lambda\subset \mathbb R^2$ be a piecewise $\mathsf{C}^2$-smooth domain. Then, with the above test function $f(t) = t(1-t)$,
\begin{align}
\operatorname{tr} f(P_K(L\Lambda) ) = 
\operatorname{tr} \left( 1_{L \Lambda} P_K 1_{L\Lambda^\complement} P_K 1_{L \Lambda} \right)= \frac 1 {\sqrt 2 \pi^3}\begin{cases} L \lvert \partial \Lambda \rvert \ln(K) +\mathcal O(L)\quad &\text{  if } K<L\,, \\ L \lvert \partial \Lambda \rvert  \ln(L) +\mathcal O(L)\quad &\text{  if } K>L\,. \end{cases}
\end{align}
\end{theorem}

\begin{remark} In line with our \autoref{main} the value $\frac 1 {\sqrt 2 \pi^3}$ equals $\frac{2\sqrt{2}}{\pi} \mathsf{I}(t\mapsto t(1-t))$.
\end{remark}

\begin{proof}
The first trace identity is rather obvious, see the short proof following \eqref{tr(f).1}. As we are then just calculating the Hilbert--Schmidt norm of $ 1_{L \Lambda} P_K 1_{L\Lambda^\complement}$, we have (see \eqref{pl kernel})
\begin{align}
\operatorname{tr} \left( 1_{L \Lambda} P_K 1_{L\Lambda^\complement} P_K 1_{L \Lambda} \right)& =\int_{L \Lambda} \mathrm d x \int_{L \Lambda^\complement} \mathrm d y \, \lvert  G_K(\| x- y\|/\sqrt 8) \rvert^2  \\
&= \int_0^\infty \mathrm d s \, \lvert  G_K(s/ \sqrt 8) \rvert^2 F(s)\label{4.3}\,.
\end{align}
The last step relies on changing to polar coordinates in $y$ and Fubini, where
\begin{align}\label{def F}
F(s) \coloneqq s \int_{L \Lambda} \mathrm d x  \int_{0}^{2 \pi} \mathrm d \theta\, 1_{L \Lambda^\complement}(x+ s (\cos(\theta),\sin(\theta))) \,  .
\end{align}
Trivially, we have
\begin{align}
F(s) \le s \int_{L \Lambda}\mathrm d x\, 2 \pi =2 \pi s L^2 \lvert \Lambda \rvert =\mathcal O(sL^2 )\,.
\end{align}
As $\partial\Lambda$ is piecewise $\mathsf{C}^2$-smooth, we have the expansion 
\begin{align}\label{expansion F(s)}
F(s) =  2s^2 L \lvert \partial \Lambda \rvert + \mathcal O(s^3) \,.
\end{align} 
See \autoref{Appendix C} for a proof. Thus, for any $s>0$, we observe
\begin{align} 
F(s) \le C \min \big\{ 2sL^2  , s^2 (L+s) \big\} =\begin{cases} 2C sL^2    &\text{ if } s >L \,, \\ Cs^2(L+s) \le 2C  s^2 L & \text{ if } s \le L\,,\end{cases}
\end{align}
or equivalently,
\begin{align}
F(s) \le C sL \min\{s,L\}\,.
\end{align}
We want to replace the integral in \eqref{4.3} over $\mathbb R^+$ by an integral over $(\sqrt 8, \min\{  K,L\})$. Let us consider the resulting error terms. The first one is trivial, the second and third one rely upon \autoref{GK square int}.
\begin{align}
\int_0^{\sqrt 8}  \mathrm d s \, \lvert  G_K(s/ \sqrt 8) \rvert^2 F(s) &\le C \int_0^{\sqrt 8} \mathrm d s\,   s^2 L \le CL\,, \\
\int_L^\infty  \mathrm d s \lvert  G_K(s/ \sqrt 8) \rvert^2 F(s) &= C \int_{L/\sqrt 8} ^\infty \mathrm d t \,  \lvert G_K(t) \rvert^2  L^2 t \le CL\,,  \\
\int_K^{10K}   \mathrm d s \,\lvert  G_K(s/ \sqrt 8) \rvert^2 F(s) &\le CK \int_{K/\sqrt 8}^{10K /\sqrt 8 } \mathrm d t  \,\lvert G_K(t) \rvert^2 L t \le CL\,, \\
\int_{10K}^L \mathrm ds \,\lvert G_K(s/\sqrt 8) \rvert^2 F(s) &\le C \int_{10K/\sqrt 8}^\infty \mathrm d t \, \exp(-\beta t) t^2 L \le CL \exp(-\beta K)\,.
\end{align}
Thus, as the integrand is always positive, we get
\begin{align}
\int_0^\infty \mathrm d s \,\lvert  G_K(s/ \sqrt 8) \rvert^2 F(s) &= \int_{\sqrt 8}^{\min\{K,L\}} \mathrm d s \, \lvert  G_K(s/ \sqrt 8) \rvert^2  \left( 2s^2 L \lvert \partial \Lambda \rvert + \mathcal O(s^3) \right)+ \mathcal O(L) \\
&= \sqrt 8 \int_{1}^{\min\{K,L\}/\sqrt 8 } \mathrm d t\, \lvert  G_K(t) \rvert^2  \left( 16 t^2 L \lvert \partial \Lambda \rvert + \mathcal O(t^3) \right) + \mathcal O(L) \label{4.14}\,.
\end{align}
Next, we utilize \autoref{GK for t<1/2 K}. As $1<t < K/\sqrt 8$ the singularities of the denominator are outside the integration domain. As the singularity around $t =K$ is substantially away from the domain, we can easily bound that factor. Thus, we observe
\begin{align}
\lvert G_K(t) \rvert^2  &=  \left \lvert \frac{ \cos(\omega_K(t)-3\pi/4) } {4 \sqrt 2  \pi^{\frac 3 2} t^{\frac 3 2} \left( 1- \frac {t^2} {K^2} \right)^{\frac 1 4}   }+\mathcal O \left( \frac 1 {t^{\frac 5 2} }\right)  \right \rvert^2 \\
&= \frac 1 {32 \pi^3 t^3} \left( \cos(\omega_K(t)-3\pi/4) \left( 1+ \mathcal O \left( \frac {t^2}{K^2} \right) \right) + \mathcal O \left( \frac 1 t \right) \right)^2 \\
&=  \frac { \cos(\omega_K(t)-3\pi/4) ^2}  {32 \pi^3 t^3} +\mathcal O \left(  \frac 1 {t K^2} + \frac 1 {t^4} \right).\label{4.17}
\end{align}
We estimate the next batch of error terms of \eqref{4.14} and \eqref{4.17}, respectively,
\begin{align}
\int_1^{\min\{K,L\}/\sqrt 8}  \mathrm d t\,\lvert G_K(t) \rvert^2 t^3  &\le C \int_1^{L}  \mathrm d t\, 1 = CL\,, \\
\int_1^{\min\{K,L\}/\sqrt 8} \mathrm d t \, \left( \frac 1 {tK^2} + \frac 1 {t^4} \right) L t^2   &\le C L \left( \frac  {K^2} {K^2} + 1 \right)  =2CL\,.
\end{align}
Thus, we have shown
\begin{align}
\int_0^\infty \mathrm d s \,\lvert  G_K(s/ \sqrt 8) \rvert^2 F(s)  &= \sqrt 8  \int_1^{\min\{K,L\}/\sqrt 8 }  \mathrm d t  \, \frac { \cos(\omega_K(t)-3\pi/4) ^2}  {32 \pi^3 t^3}  16 t^2 L \lvert \partial \Lambda \rvert  +\mathcal O(L) \\
&= \frac { \sqrt 2 L \lvert \partial \Lambda \rvert} { \pi^3}    \int_1^{\min\{K,L\}/\sqrt 8 }  \mathrm d t\, \frac{\cos(\omega_K(t)-3\pi/4)^2} t  +\mathcal O(L)   \\
&= \frac { \sqrt 2 L \lvert \partial \Lambda \rvert} { \pi^3} \int_ {\omega_K(1)} ^{\omega_K( \min\{ K,L\} /\sqrt 8 )} \mathrm d s\, \frac {\cos(s-3\pi/4)^2 }{ \omega_K^{-1}(s) \omega_K'(\omega_K^{-1}(s)) } +\mathcal O(L) \,.
\end{align}
For the last step, we need to show that $\omega_K$ is invertible on the range $(1, K/\sqrt 8)$ and need to estimate the inverse function $\omega_K^{-1}$ and the differential $\omega_K'$. As $\omega_K(t)= 4K\eta(t/K)$ we observe that
\begin{align}
\omega_K'(t) = 4\eta'(t/K) = 4 \sqrt{1- \frac{t^2} {K^2}} = 4 + \mathcal O \left( \frac {t^2} {K^2} \right). \label{omegaK' asmyp eq}
\end{align}
Since $\omega_K'$ and $\omega_K^{-1}$ appear in the denominator, we need to establish lower bounds for both of them. Thus, for $0 \le t \le  K/\sqrt 8$, we see that $4 \ge \omega_K'(t) \ge  4  \sqrt{7/8}$. This ensures that $\omega_K$ is invertible. Since $\omega_K(0)=0$, this also implies
\begin{align}
4t \ge \omega_K(t) \ge 4\sqrt{7/8} \, t\,. \label{omegaK bounds}
\end{align}
Thus, for $0 \le s \le \omega_K(K/\sqrt 8 )$, we have
\begin{align}
4 \ge  \omega_K'(\omega_K^{-1}(s)) \ge 4 \sqrt{7/8}\,, \quad \frac 1 4 s \le \omega_K^{-1}(s) \le \frac 1 4 \sqrt{ \frac 8 7 } s \,.
\end{align}
Using the last two estimates, \eqref{omegaK' asmyp eq} and \eqref{omegaK asymp eq1}, we also obtain
\begin{align}
\frac 1 {\omega_K'(\omega_K^{-1}(s)) } = \frac 1 4  + \mathcal O \left( \frac {s^2} {K^2} \right),\quad
\frac 1 {\omega_K^{-1}(s)} =\frac 4 s + \mathcal O \left( \frac {s} {K^2} \right).
\end{align}
Thus, we conclude
\begin{align}
\int_0^\infty \mathrm d s \, &\lvert  G_K(s/ \sqrt 8) \rvert^2 F(s)  
\\
&= \frac { \sqrt 2 L \lvert \partial \Lambda \rvert} { \pi^3}  \int_ {\omega_K(1)} ^{\omega_K( \min\{ K,L\} /\sqrt 8 )} \mathrm d s \,\frac {\cos(s-3\pi/4)^2 }{ \omega_K^{-1}(s) \omega_K'(\omega_K^{-1}(s)) }  + \mathcal O(L)  \\
&=\frac { \sqrt 2 L \lvert \partial \Lambda \rvert} { \pi^3}  \int_ {\omega_K(1)} ^{\omega_K( \min\{ K,L\} /\sqrt 8 )}  \mathrm d s\, \Big[\frac {\cos(s-3\pi/4)^2 }{ s } + \mathcal O \left( \frac {s} {K^2} \right)\Big] +\mathcal O(L)  \\
&= \frac { \sqrt 2 L \lvert \partial \Lambda \rvert} { \pi^3}  \int_ {\omega_K(1)} ^{\omega_K( \min\{ K,L\} /\sqrt 8 )} \mathrm d s\, \frac {1 - \sin(2s) }{ 2s }  +\mathcal O(L)  \\
&=  \frac  {L \lvert \partial \Lambda \rvert}  { \sqrt 2 \pi^3} \left(  \ln( \min \{K,L\} ) + \mathcal O(1) \right) + \mathcal O(L) \\
&=\frac  {L \lvert \partial \Lambda \rvert}  { \sqrt 2 \pi^3}   \ln( \min \{K,L\} )  + \mathcal O(L) \,. 
\end{align}
We used that $2\cos(a-3\pi/ 4)^2=  1+ \cos(2a-3 \pi /2)= 1- \sin(2a)$  and \eqref{omegaK bounds}.

\end{proof}

\section{The case \texorpdfstring{$K \gg L$}{K  ^^e2^^89^^ab  L}}

\begin{theorem} \label{thm large K}
Suppose that $\Lambda$ has a piecewise $\mathsf{C}^1$-smooth boundary $\partial\Lambda$. Let $L\le C K^{2/5}$. Then, for any polynomial $f$ with $f(0)=0$, we have the asymptotic expansion
\begin{align}
\operatorname{tr} f( 1_{L \Lambda} P_K 1_{L \Lambda})=  L^2 \frac{\lvert \Lambda \rvert}{2\pi} f(1)+  L \ln(L) \lvert  \partial \Lambda \rvert\,\frac{2\sqrt{2}}{\pi}\, \mathsf{I}(t \mapsto (f(t)-tf(1)))+ \mathcal O(L)\,,
\end{align}
 as $L\to\infty$ (and hence $K\to\infty$) with $\mathsf{I}( \blank)$ defined in \eqref{def I(f)}. 
\end{theorem}

\begin{remarks}
\begin{enumerate}
\item 
Our proof relies on the same result for $B=0$, see  \cite[Theorem 2.2]{Sobolev2014}, where the set $\Lambda$ plays the same role as here and $\Omega$ is the Fermi sea at Fermi energy $2$, which is smooth. We will show that in the case $L \le CK^{2/5}$, the magnetic field only yields a small perturbation relative to the free case $H=- \Delta$.
\item 
With some efforts we could improve the result and relax the condition to $L = \mathcal O( \sqrt{K})$. As for the particle number fluctuations with the quadratic function $f$, we believe that the optimal condition is $L \le C K$, but we do not know how to prove this.
\end{enumerate}
\end{remarks}

\begin{proof} For the integral kernel of the Fermi projection of the Laplacian we have the explicit expression  in terms of the Bessel function $J_1$,
\begin{align} \label{Pinfty def}
P_{\infty}(x,y) \coloneqq 1(-\Delta \le 2) (x,y) = \begin{cases}\frac{ J_1(\sqrt 2 \| x-y \|) } { \sqrt 2 \pi \| x-y \|} \,,\quad &x,  y\in\mathbb R^2\, \text{ with } x\neq y\, , \\
\frac 1 {2 \pi} \, , \quad  &x=y \in \R^2 \, . \end{cases}
\end{align}
The latter is also the pointwise limit of the integral kernel of $P_K$ as $K\to\infty$, see \autoref{GK for t<1/2 K}. Formula \eqref{Pinfty def} can be derived by a simple Fourier transformation, the use of polar coordinates and the very definition of the Bessel function $J_1$. 

We consider the polynomials $f(t)=t^m$ for $m \in \mathbb N$ with $m \ge 2$; the case $m=1$ is covered by \autoref{easy box estimate} without mentioning $P_\infty$.

According to the above remark it suffices to show
\begin{align}
\operatorname{tr}  (1_{L \Lambda} P_K 1_{L \Lambda})^m = \operatorname{tr} ( 1_{L \Lambda} P_\infty 1_{L \Lambda})^m + \mathcal O(L)\,.
\end{align}
To this end, we use \autoref{tele sum lem2} with $A_1 = 1_{L \Lambda} P_K  1_ {L \Lambda}$ and $A_2 = 1_{L \Lambda} P_\infty  1_ {L \Lambda}$. Both operator norms are bounded by $1$. Let $R_0>0$ satisfy $\Lambda \subset D_{R_0}(0)$. We begin by estimating
\begin{align}
\lVert 1_{L \Lambda} (P_K - P_\infty) 1_ {L \Lambda} \rVert_2^2  = \int_{L \Lambda } \mathrm d x \int_{L \Lambda } \mathrm d y \, \lvert P_K(x,y)- P_\infty(x,y) \rvert ^2\,. 
\end{align}
For this, we need the assumption $L\le C K^{0.4}$. For $x,y \in L \Lambda$, according to \eqref{pl kernel} and \autoref{GK for t<1/2 K}, we have
\begin{align}
 P_K(x,y) &= \exp \left( i\frac {1} {2K} x \wedge y \right) G_K(\lVert x-y \rVert/\sqrt 8) \\
&= \left( 1+ \mathcal O\left( \frac{\|x-y \| \| x \|} {K} \right) \right) \left( P_\infty (x,y ) + \mathcal O\left( L^{1.5} K^{-2} + K^{-1}L^{-1.5} \right) \right) \\
&= P_\infty (x,y) + \mathcal O \left(   \frac{\|x-y \| \| x \|} {K} \frac 1 {(1+ \|x-y \|)^{1.5}} + L^{1.5} K^{-2} + K^{-1}L^{-1.5} \right) \\
&= P_\infty (x,y) + \mathcal O \left(   \frac{ L}  {K(1+ \| x-y \|)^{0.5}} + L^{1.5} K^{-2} + K^{-1}L^{-1.5} \right).
\end{align}
Thus, we observe 
\begin{align}
\lVert 1_{L \Lambda} &(P_K - P_\infty) 1_ {L \Lambda} \rVert_2^2 \\
&=\int_{L \Lambda } \mathrm d x \int_{L \Lambda } \mathrm d y \,\lvert P_K(x,y)- P_\infty(x,y) \rvert ^2 \\
& \le C \int_{D_{R_0L}(0)} \mathrm d x \int_{D_{R_0L}(0)} \mathrm d y  \, \left(   \frac{ L}  {K(1+ \| x-y \|)^{0.5}} + L^{1.5} K^{-2} + K^{-1}L^{-1.5} \right)^2 \\
& \le C \int_{D_{R_0L}(0)} \mathrm d x \int_{D_{R_0L}(0)} \mathrm d y  \, \left(   \frac{ L^2}  {K^2(1+ \| x-y \|)} + L^{3} K^{-4} + K^{-2}L^{-3} \right) \\
&\le C \left( L^5 /K^2 + L^7/K^4 + L^1/K^2 \right) \le C\,.
\end{align}

Using the last estimate, \eqref{Pinfty def} and \eqref{BesselJ1 upper bound eq1}, we can conclude
\begin{align}
\max \left( \lVert 1_{L \Lambda} P_K  1_ {L \Lambda} \rVert_2 , \lVert 1_{L \Lambda}  P_\infty 1_ {L \Lambda} \rVert_2 \right)^2  
&\le  \left( \lVert 1_{L \Lambda} P_\infty  1_ {L \Lambda} \rVert_2 + C \right)^2 \\
& \le2C^2+ 2   \int_{L \Lambda } \mathrm d x \int_{L \Lambda } \mathrm d y  \,  \lvert P_\infty(x,y) \rvert ^2 \\
&\le  C+\int_{L \Lambda } \mathrm d x \int_{L \Lambda } \mathrm d y \,\frac C {(1+ \| x-y \|)^3}  \\
&\le   C+ \int_{D_{LR_0}(0) } \mathrm d x \int_{\mathbb R^2} \mathrm d y\, \frac C {(1+ \| x-y \|)^3} \\
&\le  C L^2  \,.
\end{align}
The error term in \eqref{K>L op algebra est} is therefore of the order $\mathcal O( \sqrt {CL^2} C)= \mathcal O(L)$ and we arrive at 
\begin{align}
\operatorname{tr}  (1_{L \Lambda} P_K 1_{L \Lambda})^m = \operatorname{tr} ( 1_{L \Lambda} P_\infty 1_{L \Lambda})^m + \mathcal O(L)\,,
\end{align}
which, as we explained before, finishes the proof. 
\end{proof}

\section{The case \texorpdfstring{$K \ll L$}{K ^^e2^^89^^aa L}: Overview}
In this section we discuss the asymptotic expansion when $K$ grows significantly slower than $L$. We first need to list some technical conditions. 
\begin{restatable}{condition}{LKLconditions} \label{LKL conditions}
We say a triple $(\Lambda,K,L)$, consisting of a domain $\Lambda$, an integer $K \ge 3$ and a real number $L\ge 100$ satisfies 
\begin{description}
\item[Condition A] $\Lambda$ is a polygon (see \autoref{def polygon} for a formal definition) and $K< C L/\ln(L)$ for some (finite) constant $C$, or
\item[Condition B] $\Lambda$ is a $\mathsf{C}^2$-smooth domain and $K^2 \le C L$ for some (finite) constant $C$.
\end{description}
\end{restatable}

Our third main result of this paper is the following theorem for polynomials $f$. It verifies our \autoref{main} under the above Condition.

\begin{theorem} \label{thm small K}
Let the triple $(\Lambda,K,L)$ satisfy either one of the conditions in \autoref{LKL conditions}.  Then, for any polynomial $f$ with $f(0)=f(1)=0$, we have the asymptotic expansion
\begin{align} \label{thm small K eq}
\operatorname{tr} f( 1_{L \Lambda} P_K 1_{L \Lambda})= L \ln(K) \lvert  \partial \Lambda \rvert\,\frac{2\sqrt{2}}{\pi}\,\mathsf{I}(f)+ \mathcal O(L \ln \ln (K))\,,
\end{align}
as $K\to\infty$ (and hence $L\to\infty$). The coefficient $\mathsf{I}(f)$ is defined in \eqref{def I(f)}.
\end{theorem}
The proof of this theorem spans the remaining chapters of this paper. 
Our approach essentially approximates the boundary curve $\partial \Lambda$ by a straight line and this leads to an $\mathcal O(K^2)$ error term for $\mathsf{C}^2$-smooth domains.  When we already have straight boundary lines as in polygons then our error term is much smaller, or put differently, we can allow for a larger $K$. The difference in results may only be due to our methods of proof.

The proof of these theorems is rather long, so we begin with a short summary. We will consider the polynomials $f(t)=(1-t)^mt$ with $m\in\mathbb N$. For them, we get
\begin{align} \label{tr(f).1}
\operatorname{tr} f( 1_{L \Lambda} P_K 1_{L \Lambda}) = &  \operatorname{tr}  1_{L\Lambda}  (P_K 1_{L\Lambda^\complement} )^m P_K 1_{L\Lambda} \, .
\end{align}
To see this identity, we note that the eigenvalues of $1_{L \Lambda} P_K 1_{L \Lambda}$ agree with the eigenvalues of $P_K 1_{L \Lambda} P_K$, as both are jut the squares of the principal values of $P_K 1_{L \Lambda}$. Thus, we get
\begin{align*}
\operatorname{tr} f( 1_{L \Lambda} P_K 1_{L \Lambda}) &= \operatorname{tr} f(P_K 1_{L\Lambda} P_K) \\
&=\operatorname{tr} (\mathds{1}-  P_K 1_{L \Lambda} P_K)^m  P_K 1_{L \Lambda} P_K \\
&=\operatorname{tr} 1_{L \Lambda}P_K (\mathds{1}-  P_K 1_{L \Lambda} P_K)^m  P_K 1_{L \Lambda} \\
&=\operatorname{tr}1_{L \Lambda} P_K (P_K-  P_K 1_{L \Lambda} P_K)^m  P_K 1_{L \Lambda} \\
&=\operatorname{tr} 1_{L \Lambda}P_K (P_K(\mathds{1}-1_{L \Lambda}) P_K)^m  P_K 1_{L \Lambda} \\
&=\operatorname{tr}  1_{L\Lambda}  (P_K 1_{L\Lambda^\complement} )^m P_K 1_{L\Lambda} \, .
\end{align*}
We have used that $P_K 1_{L \Lambda} P_K$ commutes with $P_K$.

For fixed $K$ and $L\to \infty$, the asymptotics of the trace of this operator\footnote{Apart from the different base of polynomials $f(t)=t^m$} was reduced to an integral depending only on $K$ and $m$ using Roccaforte's formula in \cite{Leschke2021}. They calculated the said integral. The strong (exponential) decay of the integral kernel $P_K$ for fixed $K$ was used to deal with the error terms originating from Roccaforte's formula. As we consider the limit $K \to \infty$, this is not directly possible. We can, however, get to the same integral using some geometric manipulations and operator estimates before switching to the integral. This will allow us to show the same leading term with an error bound, that depends on $K$ in a good manner. After that, we still have to study the asymptotics of the formula in \cite[Theorem 2]{Leschke2021} (for fixed $K$) as $K \to \infty$.

Let us define for any measurable $E,E'\subset \mathbb R^2$  and $m\in\mathbb N$,
\begin{align} \label{def Jm}
\mathcal J_m (E,E';K) \coloneqq \operatorname{tr}1_{E}  (P_K 1_{E'} )^m P_K 1_E \,,
\end{align}
which takes values in $[0,\infty]$, as we shall soon see.

Then, with  $f_m(t)\coloneqq t(1-t)^m$, 
\begin{align}
\operatorname{tr} f_m( 1_{L \Lambda} P_K 1_{L \Lambda}) =   \operatorname{tr}  1_{L\Lambda}  (P_K 1_{L\Lambda^\complement} )^m P_K 1_{L\Lambda} 
=\mathcal J_m(L\Lambda,  L \Lambda^\complement;K) \,.
\end{align}

The proof of \autoref{thm small K} relies on two main steps, which will be proved in the remainder of this paper. The first one is
\begin{restatable}{theorem} {JMdomaincor} \label{Jm domain reduction cor}
Let the triple $(\Lambda,K,L)$ satisfy either one of the conditions in \autoref{LKL conditions}. Let $m \in\N$ and $f_m(t) = t(1-t)^m$. Then, we have the asymptotic expansion
\begin{align}
\operatorname{tr} f_m (1_{L \Lambda} P_K 1_{L \Lambda} ) = \mathcal J_m (L \Lambda, L \Lambda^\complement;K)= L \lvert \partial \Lambda \rvert \,\mathcal J_m  \left( [0,1) \times \mathbb R^-, \mathbb R \times \mathbb R^+;K \right) + \mathcal O(L) \, .
\end{align}
\end{restatable}

This result achieves that the leading term of the asymptotic expansion of $\mathrm{tr} f_m(1_{L\Lambda} P_K 1_{L\Lambda})$ is a product of the surface area of $L\Lambda$ and a term $\mathcal J_m(E,E';K)$ with fixed sets $E$ and $E'$ which depends on $K$ but is independent of $L$ and $\Lambda$.

%\begin{remark} Man könnte das $\mathcal O(L)$ durch $o(L)$ ersetzen, wenn in \autoref{LKL conditions} jeweils eine $o$-Bedingung an $K$ gestellt wird. Das verstehe ich nicht, denn schon mit $K\le C L/\ln(L)$ ist $K=o(L)$.
%\end{remark}

This is proved in Sections \ref{norm estimates section} and  \ref{L dependency section}. The second main step is 
\begin{restatable}{theorem}{JMtoIthm} \label{Jm=I thm}
With $\mathcal J_m$ defined in \eqref{def Jm} and $\mathsf{I}$ defined in \eqref{def I(f)}, we have as $K \to \infty$
\begin{align}
\mathcal J_m([0,1) \times \R^- , \R \times \R^+ ;K)= \frac{2\sqrt{2}}{\pi}\, \mathsf{I}( t \mapsto t(1-t)^m )\, \ln(K) + \mathcal O(\ln \ln(K)) \,. 
\end{align}
\end{restatable}
This is proved in \autoref{dependence on K section} and \autoref{main term is MI proof section}.  With these theorems, we can easily conclude the
\begin{proof}[Proof of \autoref{thm small K}]
Note that any polynomial $f$ with $f(0)=f(1)=0$ can be written as a finite linear combination of the basis polynomials $f_m(t) = t(1-t)^m$ for $m \in \mathbb N$.  Since both sides of \eqref{thm small K} are linear in $f$, it suffices to show the identity for $f=f_m$, which \autoref{Jm domain reduction cor} and \autoref{Jm=I thm} do.
\end{proof}

Now we collect some properties of $\mathcal J_m$.
\begin{lemma}\label{Jm properties}
Let $E,E_1,E_2  , E',E_1',E_2'\subset \mathbb R^2$ be measurable.  Then, 
\begin{enumerate}
\item $\mathcal J_m(E,E';K) \in [0, \infty]$ is well-defined.
\item  $\mathcal J_m( \blank, E';K)$ is additive, that is, if $E_1,E_2 $ are disjoint, we have 
\begin{align}
\mathcal J_m(E_1 \cup E_2 ,E';K)=\mathcal J_m(E_1,E';K) + \mathcal J_m(E_2,E';K)\,.
\end{align}
\item $\mathcal J_m$ satisfies the a-priori Hilbert--Schmidt norm estimate,
\begin{align}
\mathcal J_m(E,E';K)\le \mathcal J_1(E,E';K)  = \lVert 1_E P_K 1_{E'} \rVert_2^2\,.\label{Jm HS estimate 1}
\end{align}
\item $\mathcal J_m$ satisfies a Lipschitz-type estimate, in the sense of 
\begin{align}
\lvert\mathcal J_m(E_1,E_1';K) - \mathcal J_m(E_2,E_2';K) \rvert \le  \lvert E_1 \Delta E_2 \rvert/(2\pi) +  m  \lvert E_1' \Delta E_2' \rvert/(2\pi)\, ,
\end{align}
where $E_1 \Delta E_2$ is the symmetric difference. Consequently, we can always change the sets $E$ and $E'$ in the two arguments by sets of (Lebesgue) measure zero without changing $\mathcal J_m$.
\item For any unitary $\mathcal A \colon \mathbb R^2 \to \mathbb R^2$ which is affine-linear, that is, $\lVert \mathcal A (x) -\mathcal A(y) \rVert= \lVert x-y \rVert$, we have
\begin{align}
\mathcal J_m(E,E';K)=\mathcal J_m(\mathcal A(E), \mathcal A(E');K)\,.
\end{align}
\end{enumerate}
\end{lemma}

\begin{proof}
We have
\begin{align}
&\mathcal J_m(E,E';K)=   \operatorname{tr}1_{E}  (P_K 1_{E'} )^m P_K 1_E \\&= \begin{cases} \operatorname{tr}1_{E}  (P_K 1_{E'} )^{m/2}P_K P_K  (1_{E'} P_K )^{m/2}   1_E \, \, = \lVert 1_{E}  (P_K 1_{E'} )^{m/2}P_K \rVert_2^2\quad &\text{ if $m$ is even\,,} \\
	\operatorname{tr}1_{E}  (P_K 1_{E'} )^{(m+1)/2}  (1_{E'} P_K )^{(m+1)/2}   1_E  = \lVert 1_{E}  (P_K 1_{E'} )^{(m+1)/2}  \rVert_2^2  \quad &\text{ if $m$ is odd\,.} \end{cases}
\end{align} %sub-optimal alignment
Thus, $\mathcal J_m(E,E';K)$ is the trace of a positive operator, which has a well-defined value in $[0,\infty]$.

As $\operatorname{tr}AB=\operatorname{tr}BA$, we have
\begin{align}
\mathcal J_m(E,E';K)=  \operatorname{tr}1_{E'}  P_K 1_E(P_K 1_{E'} )^m \,.
\end{align}
For any disjoint $E_1,E_2$ we have $1_{E_1\cup E_2}=1_{E_1}+1_{E_2}$ and thus, as the trace is linear, we get
\begin{align}
\mathcal J_m(E_1 \cup E_2,E';K)= \mathcal J_m(E_1,E';K)+\mathcal J_m(E_2,E';K)\,.
\end{align}

For the Hilbert--Schmidt norm estimate, we just use that $1_E,1_{E'}$, and $P_K$ are projections (and thus have norm $1$). Then, we get 
\begin{align}
\mathcal J_m (E,E';K) = \lVert 1_E (P_K1_{E'})^m P_K 1_E \rVert_1 \le \lVert 1_E P_K 1_{E'} \rVert_2^2\,,
\end{align}
where the first equality holds by positivity. 

For the Lipschitz-type estimate, we first consider
\begin{align}
\left \lvert \mathcal J_m(E_1,E_1';K)- \mathcal J_m(E_2,E_1';K) \right \rvert &= \lvert \mathcal J_m(E_1 \setminus E_2,E_1';K) - \mathcal J_m(E_2 \setminus E_1,E_1';K) \rvert \\
&\le  \lVert 1_{E_1 \setminus E_2} P_K 1_{E_1'} \rVert_2^2+ \lVert 1_{E_2\setminus E_1} P_K 1_{E_1'} \rVert_2^2 \\
&\le \frac 1 {2\pi}\left (  \lvert E_1 \setminus E_2 \rvert+ \lvert E_2 \setminus E_1 \rvert \right ) = \frac 1 {2\pi} \lvert E_1 \Delta E_2 \rvert\,.
\end{align}
This relies upon \autoref{easy box estimate} and the additivity above. For the estimate in the second component, we see\footnote{The start of this calculation is copied from the proof of \autoref{cutoff box 2}, with some renaming of the sets $E$.}
\begin{align}
\big\lvert \mathcal J_m&(E_2,E_1';K)-\mathcal J_m(E_2,E_2';K) \big\rvert \\
 &= \left \lvert  \sum_{k=1}^m \operatorname{tr} 1_{E_2} (P_K 1_{E_1'})^{k-1} P_K (1_{E_1'}- 1_{E_2'}) (P_K 1_{E_2'})^{m-k} P_K 1_{E_2} \right \rvert \\
& \le  \sum_{k=1}^m \left \lVert  1_{E_2} (P_K 1_{E_1'})^{k-1} P_K (1_{E_1'}- 1_{E_2'}) (P_K 1_{E_2'})^{m-k} P_K 1_{E_2} \right \rVert_1 \\
&\le m \lVert P_K (1_{E_1'}-1_{E_2'})P_K\rVert_1\\
&\le m \left( \lVert P_K 1_{E_1' \setminus E_2'}\rVert_2^2+ \lVert P_K 1_{E_2' \setminus E_1'}\rVert_2^2 \right)=\frac m {2 \pi} \lvert E_1'\Delta E_2'\rvert\,.
 \end{align}
Combining the two estimates with the triangle inequality finishes the Lipschitz-type estimate.

For the last point,  let $\mathcal S_{\mathcal A}$ be the unitary operator on $\Lp^2(\mathbb R^2)$ that maps $f \mapsto  (x \mapsto f(\mathcal A(x)))$. We immediately see that for any measurable $E$, we have
\begin{align}
(\mathcal S_{\mathcal A})^{-1} 1_E \, \mathcal S_{\mathcal A} = 1_{\mathcal A(E)}\,. 
\end{align}
As the set of all affine linear maps, under which the claim holds, is closed under composition, we may assume that $\mathcal A$ is either a translation or a reflection\footnote{Any rotation is the composition of two reflections}. Let $\mathcal A$ be a reflection,  that is, a $2\times2$ orthogonal matrix with determinant $-1$. For any $x,y \in \mathbb R^2$, we observe
\begin{align}
\left( \mathcal S_{\mathcal A} P_K \mathcal S_{\mathcal A}^{-1} \right) (x,y)= P_K(\mathcal A(x), \mathcal A(y))=\overline{P_K(x,y)}\,,
\end{align}
as $\lVert x-y\rVert= \lVert \mathcal Ax- \mathcal Ay \rVert$ and $x \wedge y=- ( \mathcal Ax) \wedge ( \mathcal Ay)$. By Mercer's theorem, we have
\begin{align}
\mathcal J_m(E,E';K)&= \operatorname{tr} 1_E (P_K 1_{E'})^m P_K 1_E \\
&= \int_{E} \mathrm d x \int_{E'} \mathrm dx_1 \int_{E'} \mathrm dx_2 \cdots \int_{E'} \mathrm dx_m \, P_K(x,x_1)P_K(x_1,x_2) \cdots P_K(x_m,x)\,.
\end{align}
Thus, we have $\mathcal J_m(E,E';K)=\overline{\mathcal J_m( \mathcal A(E), \mathcal A(E');K)}$. As $\mathcal J_m(E,E';K)$ is always real, this implies $\mathcal J_m(E,E';K)=\mathcal J_m( \mathcal A(E), \mathcal A(E');K)$.

We now consider the case that $\mathcal A$ is a translation, that is, $\mathcal A(x)=x-x_0$ for some $x_0 \in \mathbb R^2$. Here, we need to use the magnetic translation operator $\mathcal S_{x_0, K} $, which is given by 
\begin{align}
(\mathcal S_{x_0, K} \varphi )(x) \coloneqq \exp(   i x\wedge x_0/(2K) )\, \varphi(x-x_0)\,,\quad x \in \mathbb R^2\,,
\end{align}
for $\varphi \in \Lp^2(\mathbb R^2)$. For any measurable set $E$, it satisfies $\mathcal S_{x_0,K}  1_E \mathcal S_{x_0, K}^{-1} = 1_{E+x_0}$ and commutes with $P_K$. Thus, we observe 
\begin{align}
\mathcal J_m(E,E';K)&= \operatorname{tr} 1_{E'} P_K 1_E (P_K 1_{E'})^m \\
&= \operatorname{tr} \mathcal S_{x_0,K} 1_{E'} P_K 1_E (P_K 1_{E'})^m \mathcal S_{x_0,K}^{-1} \\
&= \operatorname{tr} 1_{E'+x_0} P_K 1_{E+x_0} (P_K 1_{E'+x_0})^m \\
&= \mathcal J_m(E+x_0, E'+x_0;K)\,.
\end{align}
\end{proof}

\section{\texorpdfstring{$K \ll L$}{K ^^e2^^89^^aa L}: Norm estimates} \label{norm estimates section}

In this section we prove norm estimates of various combinations and powers of $1_E$ (for certain sets $E$) and $P_K$, that will be useful later. 

\begin{lemma} \label{cutoff box 1}
Let $K, m \in \mathbb N$. If $m>1$,  then we have arbitrary measurable sets $A_i\subset \R^2$ for $i=2, \dots , m$. Let $R>0$ and $x_0 \in \mathbb R^2$. Then, there are constants $C<\infty$ and $\beta>0$ such that
\begin{align}
\left \lVert 1_{D_R(x_0)} P_K \left( \prod_{i=2}^m 1_{A_i} P_K \right) 1_{D_{(m+1)R}^\complement (x_0)} \right \rVert_2^2 \le  C m^4 \begin{cases} R   \quad &\text{ if } R< 10K\,,  \\  \exp(-\beta R) & \text{ if } R \ge 10K\,. \end{cases} 
\end{align}
If $m=1$, then the product $\prod_{i=2}^m 1_{A_i} P_K$ is interpreted as the identity.
\end{lemma}
\begin{proof}
We will prove this by induction over $m$. We begin with a slight generalisation of the case $m=1$. Let $l \in \mathbb N$. Then, we observe
\begin{align}
\left \lVert 1_{D_{lR}(x_0)} P_K 1_{D_{(l+1)R}^\complement(x_0)} \right \rVert_2^2 
 &= \int_{D_{lR}(x_0)} \mathrm{d}x \int_{D_{(l+1)R}^\complement(x_0)} \mathrm d y   \,\lvert P_K(x,y) \rvert^2 \\
&=  \int_{D_{lR}(0)} \mathrm dx\int_{D_{(l+1)R}(0)^\complement} \mathrm d y \, \lvert G_K(\lVert x-y \rVert/\sqrt 8 ) \rvert ^2  \\
&\le  \pi l^2 R^2 \int_{D_{R}(0)^\complement} \mathrm d y \,\lvert G_K(\lVert y \rVert/\sqrt 8 ) \rvert ^2  \\
&\le  C l^2R^2 \int_{R/\sqrt8 } ^\infty  \mathrm d t\,  2 \pi t \, \lvert G_K(t) \rvert^2   \\
&\le  C l^2  \begin{cases} R   \quad &\text{ if } R< 10K\,,  \\  \exp(-\beta R) & \text{ if } R \ge 10K\,. \end{cases} \label{HS norm 1 proof eq1}
\end{align}
We used \eqref{pl kernel} and \autoref{GK square int} noting that $R> 10K$ implies $R /\sqrt 8 > 2K$.

We continue with the induction step. We have
\begin{align}
&\left \lVert 1_{D_R(x_0)} P_K \left( \prod_{i=2}^{m+1} 1_{A_i} P_K \right) 1_{D_{(m+2)R}^\complement (x_0)} \right \rVert_2  \\
&\le  \left \lVert 1_{D_R(x_0)} P_K  \left( \prod_{i=2}^{m} 1_{A_i} P_K \right)  1_ { D_{(m+1)R}^\complement(x_0)}1_{A_{m+1} } P_K  1_{D_{(m+2)R}^\complement (x_0)} \right \rVert_2  \\
 & + \left \lVert 1_{D_R(x_0)} P_K    \left( \prod_{i=2}^{m} 1_{A_i} P_K \right)  1_{A_{m+1}} 1_ { D_{(m+1)R}(x_0)} P_K 1_{D_{(m+2)R}^\complement (x_0)} \right \rVert_2 \\
 & \le   \left \lVert 1_{D_R(x_0)} P_K  \left( \prod_{i=2}^{m} 1_{A_i} P_K \right)  1_ { D_{(m+1)R}^\complement(x_0)} \right \rVert_2 \left \lVert 1_{A_{m+1} } P_K  1_{D_{(m+2)R}^\complement (x_0)} \right \rVert \\
 & + \left \lVert 1_{D_R(x_0)} P_K    \left( \prod_{i=2}^{m} 1_{A_i} P_K \right)  1_{A_{m+1}} \right \rVert \left \lVert 1_ { D_{(m+1)R}(x_0)} P_K 1_{D_{(m+2)R}^\complement (x_0)} \right \rVert_2 \\ 
 & \le  \left \lVert 1_{D_R(x_0)} P_K  \left( \prod_{i=2}^{m} 1_{A_i} P_K \right)  1_ { D_{(m+1)R}^\complement(x_0)} \right \rVert_2 + \left \lVert 1_ { D_{(m+1)R}(x_0)} P_K 1_{D_{(m+2)R}^\complement (x_0)} \right \rVert_2 \\
&\le  (\sqrt{ C_m} + (m+1) \sqrt{C_1}) \begin{cases} \sqrt  R   \quad &\text{ if } R< 10K\,,  \\ C \exp(- \beta R/2) & \text{ if } R \ge 10K\,. \end{cases} 
\end{align}
The first step relies on the triangle inequality and the fact that multiplication operators commute. The second step relies on  the H\"older inequality and the third step uses that all the operators $1_A, P_K$ have operator norm $1$, as they are projections. 

In the last step, we used the induction hypothesis and \eqref{HS norm 1 proof eq1} with $l=m+1$ (taking the square root of both in the process). Here, $C_1, C_m$ are the $m$-dependent constants in the claim. The recursion $\sqrt {C_{m+1}}= \sqrt{C_m} + (m+1) \sqrt{C_1}$ implies $\sqrt{C_m} = \frac {m^2+m} 2\sqrt {C_1}$ and thus $C_m=\frac {(m^2+m)^2} 4 C_1$. This explains the power $m^4$ in the claim and thus finishes the proof.
\end{proof}

\begin{lemma} \label{cutoff box 2}
For $R>0, x_0 \in \mathbb R^2$, let $E \subset D_R(x_0)$ and let $E' \subset \mathbb R^2$ be measurable. Then, with $\mathcal J_m$ defined in \eqref{def Jm} we have for some $\beta>0$,
\begin{align}
\mathcal J_m (E,E';K)= \mathcal J_m(E, E'\cap D_{(m+1)R}(x_0);K) +  \begin{cases} \mathcal O(m^5R)   \quad &\text{ if } R< 10K\,,  \\ \mathcal O(m^5\exp(-\beta R)) & \text{ if } R \ge 10K \,.\end{cases} 
\end{align}
\end{lemma}
\begin{proof}
Let $E'' \coloneqq E' \cap D_{(m+1)R}(x_0)$. Using the definition of $\mathcal J_m$, we can write the telescope sum
\begin{align}
&\left \lvert \mathcal J_m(E,E';K)-\mathcal J_m(E,E'';K) \right \rvert \\
 &= \left \lvert  \sum_{k=1}^m \operatorname{tr} 1_E (P_K 1_{E'})^{k-1} P_K (1_{E'}- 1_{E''}) (P_K 1_{E''})^{m-k} P_K 1_E \right \rvert \\
& \le  \sum_{k=1}^m \left \lVert  1_E (P_K 1_{E'})^{k-1} P_K (1_{E'}- 1_{E''}) (P_K 1_{E''})^{m-k} P_K 1_E \right \rVert_1 \\
& \le  \sum_{k=1}^m \left \lVert  1_E (P_K 1_{E'})^{k-1} P_K 1_{E'\setminus E''} \right \rVert_2 \left \lVert 1_{E' \setminus E''} (P_K 1_{E''})^{m-k} P_K 1_E \right \rVert_2 \\
& \le  \sum_{k=1}^m  \left \lVert  1_{D_R(x_0)} (P_K 1_{E'})^{k-1} P_K 1_{D_{(m+1)R}(x_0)^\complement} \right \rVert_2 \left \lVert 1_{D_{(m+1)R}(x_0)^\complement} (P_K 1_{E''})^{m-k} P_K 1_{D_R(x_0)} \right \rVert_2 \\
& \le   \begin{cases} C m^5R   \quad &\text{ if } R< 10K\,,  \\ C m^5\exp(-\beta R) & \text{ if } R \ge 10K\,. \end{cases} 
\end{align}
We used that $E' \subset E''$ and thus $1_{E'}-1_{E''}=1_{E' \setminus E''}$ and that $E' \setminus E'' \subset D_{(m+1)R}(x_0)^\complement$ by construction of $E''$. We also used that whenever $A \subset A'$ and $B\subset B'$, for any operator $T$, we have $\lVert 1_A T 1_B \rVert_2 =\lVert 1_A 1_{A'} T 1_{B'} 1_B \rVert_2 \le \lVert 1_{A'} T 1_{B'} \rVert_2$. The last step is of course \autoref{cutoff box 1}. 
\end{proof}

We need some more properties of $\mathcal J_m$, where $E'$ is a half space.
\begin{lemma}  \label{Jm halfspace}
Let $a,b \in [0,\infty)$ and $c \in (0,\infty]$ with $b<c$. Then, with $\mathcal J_m$ defined in \eqref{def Jm} we have 
\begin{align}
\mathcal J_m\left( [0,a) \times (-c,-b) , \mathbb R\times \mathbb R^+ ;K \right) = a \mathcal J_m\left( [0,1) \times (-c,-b) ,\mathbb R \times \mathbb R^+;K \right),
\end{align}
and there is a constant $\beta>0$ such that
\begin{align} \label{a priori stripe bound}
\mathcal J_m \left( [0,a) \times (-c,-b), \mathbb R \times \mathbb R^+ ;K \right) \le \begin{cases}  C a \ln \left( \frac{20K} {b+1} \right) &\text{ if } b<10K\,, \\ Ca \exp(-\beta b) &\text{ if } b\ge 10K\,. \end{cases} 
\end{align}
\end{lemma}

\begin{proof}
We prove the upper bound first. As $\mathcal J_m$ is additive in the first component and non-negative, as shown in \autoref{Jm properties}, it suffices to consider $c=\infty$. We observe that by the Hilbert--Schmidt norm estimate 
\begin{align}
\mathcal J_m \left( [0,a) \times (-\infty, -b), \mathbb R \times \mathbb R^+ ;K\right) &\le \int_ {0}^{a} \mathrm d x_1 \int_{-\infty}^{-b} \mathrm d x_2 \int_{\mathbb R \times \mathbb R^+} \mathrm d y \, \lvert P_K((x_1,x_2),y) \rvert^2  \\
&\le a\int_{-\infty} ^{-b} \mathrm dx_2 \int_{D_{\lvert x_2 \rvert}(0)^\complement} \mathrm d y \,\lvert G_K(\|y\|/\sqrt 8) \rvert ^2 \\
&= 16\pi a \int_{b} ^{\infty} \mathrm d s  \, \int_{s /\sqrt 8} ^\infty \mathrm d t\, t \lvert G_K(t) \rvert^2  \\
&\le C a\begin{cases}  \ln \left( \frac{20K} {b+1} \right) &\text{ if } b<10K\,, \\ \exp(-\beta b) &\text{ if } b\ge 10K\,. \end{cases} 
 \end{align}
The last step relies upon \autoref{GK square int} and an easy estimate for the integral over $s$. The first claim just says that for any $b,c$, the function
\begin{align}
a\mapsto f_{b,c}(a) \coloneqq \mathcal J_m\left( [0,a) \times (-c,-b), \mathbb R \times \mathbb R^+ ;K\right)
\end{align}
is linear. As we have already shown, it has a linear upper bound, and hence it suffices to show that it satisfies the Cauchy functional equation, that is, $f_{b,c}(a)+f_{b,c}(a')=f_{b,c}(a+a')$. For shortness, we write $I\coloneqq (-c,-b)$. We use the additivity and translational invariance of $\mathcal J_m$, as proved in \autoref{Jm properties}. We observe
\begin{align}
f_{b,c}(a+a')&= \mathcal J_m\left([0,a+a') \times I, \mathbb R \times \mathbb R^+ ;K\right) \\
&= \mathcal J_m\left([0,a) \times I, \mathbb R \times \mathbb R^+ ;K\right) + \mathcal J_m\left([a,a+a') \times I, \mathbb R \times \mathbb R^+ ;K\right) \\
&=f_{b,c}(a) + \mathcal J_m\left([0,a') \times I, (\mathbb R -a ) \times \mathbb R^+;K \right) \\
&=f_{b,c}(a) + \mathcal J_m\left([0,a') \times I, \mathbb R  \times \mathbb R^+;K \right) \\
&=f_{b,c}(a)+f_{b,c}(a')\,.
\end{align}
This completes the proof.
\end{proof}

\section{\texorpdfstring{$K \ll L$}{K ^^e2^^89^^aa L}: Proof of \autoref{Jm domain reduction cor}}
%\section{\texorpdfstring{$K \ll L$}{K ^^e2^^89^^aa L}: Factoring out the dependence on $L$ and $\Lambda$: Proof of \autoref{Jm domain reduction cor}}
\label{L dependency section}

The goal of this section is to prove \autoref{Jm domain reduction cor}, which stated
\JMdomaincor*

Let us also recall \autoref{LKL conditions}:
\LKLconditions*
The proof of  \autoref{Jm domain reduction cor} splits into the two cases according to this condition. Under Condition A, it follows from \autoref{reduction lem polygons} below and under Condition B it is implied by \autoref{Jm domain dependence lem} below.

\subsection{Proof of \autoref{Jm domain reduction cor} for polygons $\Lambda$}
The time has come for the formal definition of polygons.

\begin{definition} \label{def polygon} A polygon $\Lambda$ is a bounded, connected Lipschitz domain in $\R^2$ with boundary $\partial\Lambda$ such that there is a finite set $\mathcal V = \{x_1,\ldots, x_V \}$ of points in $\R^2$ with $\partial\Lambda$ being the union of the closed line segments $\gamma_i = [x_i,x_{i+1}], i= 1,\ldots,V$ between $x_i$ and $x_{i+1}$ with the convention $x_{V+1} \coloneqq x_0$. The points in $\mathcal V$ (or the boundary of the edges) of are called corners and the line segments $\gamma_i$ are called edges. The number $V$ is chosen minimal so that none of the interior angles, denoted by $\theta_i$, (at the corner $x_i$) are $\pi$ or $2\pi$. 
\end{definition} 

By definition, a polygon is open and simply-connected.  By the Lipschitz property, two edges can only intersect in a common corner.

All our claims below can be extended easily to polygons that are not simply-connected but this would be more cumbersome from a notational point of view.

\begin{lemma} \label{reduction lem polygons}
Let $\Lambda \subset \mathbb R^2$ be a polygon according to \autoref{def polygon}. Let $K, m\in \mathbb N, L \in \mathbb R$ with $L>2$. Then, as $L \to \infty$ and with $K\le C L/\ln(L))$ (for some finite constant $C$), we have the asymptotic expansion
\begin{align}
\mathcal J_m(L\Lambda, L\Lambda^\complement;K ) & = L  \lvert \partial \Lambda \rvert \, \mathcal J_m \left( [0,1) \times \mathbb R^-, \mathbb R \times \mathbb R^+ ;K \right) + \mathcal O( (\ln(L)^3+K \ln(K) ))\,. \label{reduction lem polygons.result eq}
\end{align}
\end{lemma}

\begin{remark} The condition $K\le C L/\ln(L)$ fits the condition for polygons in the main \autoref{thm small K}. We could extend \autoref{reduction lem polygons} to $K,L$ such that $K/L\to0$ as $L\to\infty$ with essentially the same proof at the expense of a larger error term. 
\end{remark}

\begin{proof}

As we have additivity in the first argument of $\mathcal J_m$, we will decompose $L\Lambda$ as a (disjoint) union of certain sets, called edge sets, corner sets, the essential interior and a (remaining) null set. First, we consider the collar neighbourhood of the boundary $L\Lambda$, namely $D_R(\partial L \Lambda)\cap L\Lambda$, define $E_{\text{int}} \coloneqq L \Lambda \setminus\overline{ D_R(\partial L \Lambda)}$ and call this the essential interior. We want to cover (up to null sets) the collar neighbourhood by (lots of) open squares of length $R$, so that these squares are all inside $L\Lambda$, one side lies on an edge of $L\Lambda$ and there is a ball of radius $R(m+1)$ around the centre of the square, in which $L \Lambda$ looks like a half-space. To enable this condition and prevent squares of different edges from intersecting we stop in a safe distance from the corners. Around each corner of $L\Lambda$ we will cut out some measurable subset of a disk of radius $3 R(m+1)/\varepsilon$ centred at the corner and call the intersection with $L\Lambda$ a corner set. This distance depends also on the angle at a corner and therefore we introduce an $\varepsilon$. The variables $R$ and $\varepsilon$ will be chosen accordingly in the proof. A null set is included to really cover $L\Lambda$ by all these components.

\begin{figure}[H]
  \centering
  \includegraphics[height=5cm]{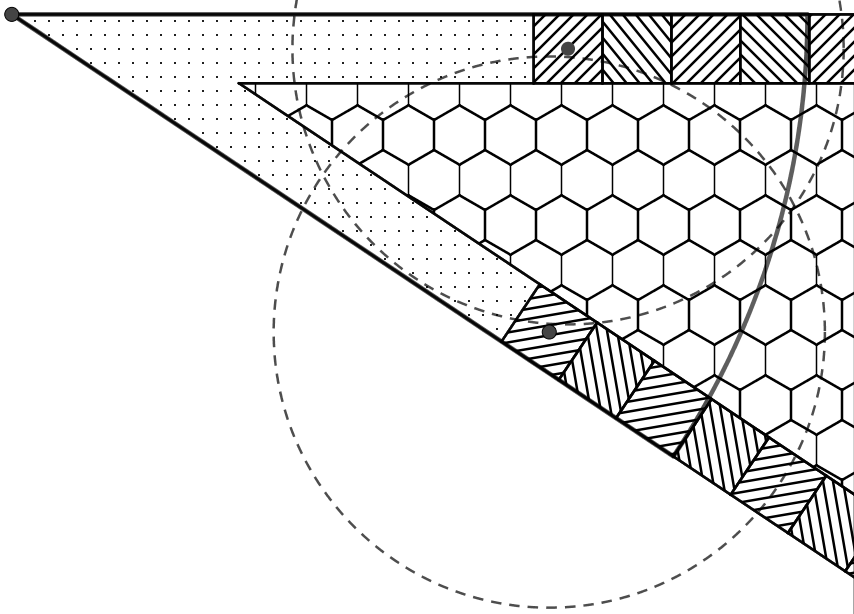} \hspace{1cm}
 \includegraphics[height=5cm]{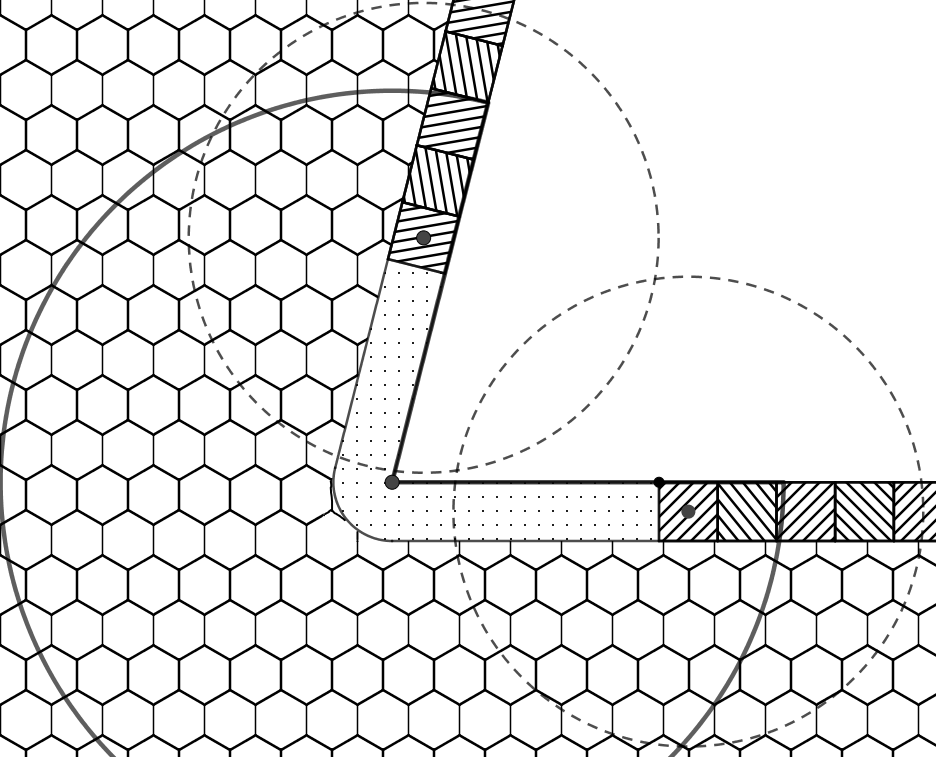}
  \caption{Two examples of a corner set (dotted), some edge sets (hatched) and the essential interior (covered by a hexagonal grid) near a corner with different angles. The dotted circles show the balls of radius $R(m+1)$ around the centre of edge sets and the circular sectors are the disks of radius $R(m+1)/\varepsilon$ around the corners intersected with $\Lambda$.}  
\end{figure}

As $\Lambda$ is a polygon, there is an $\varepsilon \in (0,1)$ with the following \emph{intersection property}: For each $\delta< \varepsilon$, $x \in \mathbb R^2$ such that the disk $D_{\varepsilon \delta }(x)$ intersects multiple edges, it intersects exactly two edges and the corner between these edges is at most $\delta$ away from $x$. One such $\varepsilon$ can be constructed as follows. Let $d_i$ be the smallest (Euclidean) distance between the corner $x_i$ and any edge of the polygon not containing $x_i$. Then we set (recall, $\theta_i$ is the interior angle at $x_i$)
\begin{align}
\varepsilon \coloneqq \min \{ |\sin(\theta_i)|,d_i : i=1,\ldots,V\}\,. 
\end{align}

Since $K < L/ \ln(L)$ and since we only care for the asymptotic behaviour as $L \to \infty$, we can assume 
\begin{align}\label{7.3}
30K(m+1) < \varepsilon^2 L\,.
\end{align}
Let $R \coloneqq \max(10K,  \varepsilon^2 \ln(L)^2 /(3(m+1)) )$. This ensures that $R$ behaves nicely with the domain, as 
\begin{align} \label{7.2} 3R(m+1)<\varepsilon^2 L
\end{align}
and that the error term in \autoref{cutoff box 2}, which is mainly  $\exp(-\beta R)$ decays faster than $L^{-1}$ as $L\to \infty$.

We will now construct a covering of $L \Lambda$ with the properties we just described. Let $\gamma \subset \partial \Lambda$ be an edge of length $|\gamma| = 2 \lambda$. After an isometry, we may assume that $\gamma = (-\lambda, \lambda) \times \{0\}$ and that the outward unit normal vector of $\Lambda$ at $0$ is $+e_2$. Scaling up with $L$, the edge $L\gamma$ is now $(-L \lambda, L \lambda) \times \{ 0\} $. The candidates for our edge sets are now the sets $E_{\gamma,j} \coloneqq (R(j-1/2), R(j+1/2) ) \times (-R,0)$ for some $ j \in \mathbb Z$. 

Let $\lvert j \rvert R \le L \lambda - 2 R (m+1)/\varepsilon$. In order to show that $E_{\gamma,j}$ is an edge set, according to the above definition, we assume, by contradiction, that $D_{R(m+1)}((Rj,-R/2)) \cap \partial (L \Lambda) \setminus L \gamma$ is non-empty and contains the point $x$. Note that $+e_2$ is the outward unit normal vector of $L \Lambda$ at $(Rj,0)$. Then, we have
\begin{align}
\delta \coloneqq \operatorname{dist}(x/L, (jR/L,0) ) / \varepsilon  & \le 1/L \left( \operatorname{dist} (x, (jR,-R/2)) + R/2\right) /\varepsilon \\
&< R(m+2)/(\varepsilon L) < 3 R(m+1)/(\varepsilon L)  < \varepsilon \, .
\end{align}
The last inequality follows by \eqref{7.2}. As the points $x/L \in \partial \Lambda$ and $(jR/L,0) \in \partial \Lambda$ have distance $\varepsilon \delta$, we know that the distance from $(jR/L,0)$ to the next corner, which is $(\pm \lambda,0)$, is at most $\delta<2 R(m+1)/L \varepsilon$, by the above intersection property of $\Lambda$. This contradicts the assumption $\lvert j \rvert R \le L \lambda - 2 R (m+1)/\varepsilon$. Thus, under this assumption, $E_{\gamma, j}$ is indeed an edge set.

For an edge $\gamma$, let $N_\gamma \coloneqq \left \lfloor \left( L \lambda - 2 R (m+1)/\varepsilon \right)/R \right \rfloor$ (with $\lambda = \lvert \gamma \rvert/2$). Thus, we have $2N_\gamma+1$ edge sets along $\gamma$ and we observe that
\begin{align} \label{Ngamma asmyp}
(2N_\gamma+1)R= L \lvert \gamma \rvert + \mathcal O(R ) \, .
\end{align}
Furthermore, together with a Lebesgue null set $E_{\gamma,\text{null}}$, they cover
\begin{align}
E_\gamma \coloneqq E_{\gamma,\text{null}} \cup \bigcup_{j=-N_\gamma}^{N_\gamma} E_{\gamma,j} = [-(N_\gamma+1/2) R,(N_\gamma+1/2)R ] \times [-R,0) \,.
\end{align}

The next claim is that any two edge sets $E_{\gamma,j}, E_{\gamma',j'}$ are disjoint. For $\gamma=\gamma'$, this is obvious.  Assume $\overline{E_{\gamma,j}}\cap\overline{ E_{\gamma',j'}}\neq \emptyset $ with $\gamma \neq \gamma'$. Let  $x$ be the centre  of  $E_{\gamma,j}$. Then, the set $\overline{D_{(1+\sqrt 2 /2) R}(x) }\cap \partial  L \Lambda$ contains points on the edge $L \gamma'$ and thus, as $(1+\sqrt 2 /2)R <2R \le R(m+1)$, we can conclude that $E_{\gamma,j}$ is not an edge set, which is a contradiction.

Let us now define the corner sets. For each $i=1,2, \dots, V$, let the corner set $E_{x_i}$ be defined by
\begin{align}
E_{x_i}\coloneqq D_{3R(m+1)/\varepsilon}(L x_i) \cap \left(  L \Lambda  \setminus \left( E_{\text{int}} \cup \bigcup_{i'=1}^V E_{\gamma_{i'}}  \right) \right) .
\end{align}
We will now show that $L \Lambda$ is covered by the sets $E_{\text{int}}, E_{\gamma_i}$ and $E_{x_i}$ for $i \in \{1,2,\dots,V\}$.

Let $x \in L \Lambda$ with $x \not \in E_{\text{int}}$ and for all $i \in \{1,2,\dots,V\}$, $x \not \in E_{\gamma_i}$. As $x \in L\Lambda \setminus D_R(\partial L \Lambda)$, we know that $\operatorname{dist}(x, \partial L \Lambda) <R$. Let $y \in \partial L\Lambda $ be a closest point to $x$, that is, $\| x-y \| = \operatorname{dist}(x, \partial L \Lambda) <R$. Let $\gamma \subset \partial L$ be an edge such that $y \in L \gamma$ \footnote{There is only one such edge unless $y$ is a corner, in which case, there are two possible $\gamma$.}. Assume without loss of generality that $\gamma = (-\lambda, \lambda) \times \{0\}$.  As $x \not \in E_\gamma$, we can conclude $y \in [-L \lambda, L \lambda] \setminus [-(N_\gamma+1/2) R,(N_\gamma+1/2)R ] $. Thus, the distance of $y$ to the closest corner $x_{i_0}$ is at most $2R(m+1)/\varepsilon +R$ and the distance of $x$ to $x_{i_0}$ is at most $2R(m+1)/\varepsilon +2R\le 3R(m+1)/\varepsilon$. This implies $x \in E_{x_{i_0}}$ and thus
\begin{align}
L \Lambda =E_{\text{int}} \cup \bigcup_{i=1}^V \left(E_{\gamma_i} \cup E_{x_i}  \right) .
\end{align}

Using the Hilbert--Schmidt norm estimate $\mathcal J_m \le \mathcal J_1$ in \autoref{Jm properties} and \autoref{GK square int}, we now estimate
\begin{align} \label{Jm int estimate.1}
\mathcal J_m(E_{\text{int}} , L \Lambda^\complement;K  ) &\le \int_{E_{\text{int}}} \mathrm d x \int_{L \Lambda^\complement} \mathrm dy \,\lvert P_K(x,y) \rvert^2 \\
& \le\int_{E_{\text{int}}} \mathrm d x \int_{D_R( x )^\complement} \mathrm dy \,\lvert G_K(\|y- x\|/\sqrt 8 ) \rvert^2 \\
&= C \lvert E_{\text{int}}  \rvert \int_{R/\sqrt 8} ^\infty \mathrm d s \,s\,\lvert G_K(s) \rvert^2 \\
& \le C L^2 \lvert \Lambda \rvert \exp(-\beta R) \\
&\le C\,. \label{Jm int estimate.end}
\end{align}
In the final step, we used that $R\ge \varepsilon^2 \ln(L)^2/ (3 (m+1))$ and in the step before, we used $R \ge 10K$. The important part is the $\ln(L)^2$, which leads to an annihilation of the polynomial growth.

For a corner set $E_{ x_i }$, let $\mathcal A(\theta_i, r)$ be the circular sector centred at the corner with radius $r$ and opening angle $\theta_i$ between the two edges touching this corner. We observe
\begin{align}
E_{ x_i } \subset \mathcal A(\theta_i, 3R(m+1)/\varepsilon ) \subset L \Lambda\,.
\end{align}
Due to translational and rotational invariance of $\mathcal J_m$, we can assume that the corner is at $0$ and one edge goes in direction $e_1$. We observe
\begin{align}
\mathcal J_m(E_{ x_i  }, L \Lambda^\complement;K ) &\le \mathcal J_1 (E_{ x_i }, L \Lambda^\complement ;K) \\
&\le \mathcal J_1 (\mathcal A(\theta_i, 3R(m+1)/\varepsilon ) , L \Lambda^\complement;K) \\
&\le \mathcal J_1 (\mathcal A(\theta_i, 3R(m+1)/\varepsilon ) , \mathcal A(\theta_i , 3R(m+1)/\varepsilon ) ^\complement;K ) \\
&\le C R \ln(R)\,.
\end{align} 

We used the Hilbert--Schmidt norm estimate and could then enlarge the domains in both arguments of $\mathcal J_1$, as $\mathcal J_1(E,E';K)=\lVert 1_E P_K 1_{E'} \rVert_2^2$. The final step is just an application of \autoref{pf thm} with scaling parameter $R$ instead of $L$. The constant $C$ depends only on the angle $\theta_i $, $m$ and $\varepsilon$.\footnote{The constant is continuous in $\theta_i \in [0,2\pi]$.}

We denote the total number of edge sets by $N$ and get immediately that $N = \sum_{i=1}^V (2N_{\gamma_i} +1)$. Using \eqref{Ngamma asmyp}, we conclude that
\begin{align} \label{edge set number}
RN= \sum_{i=1}^V R(2N_{\gamma_i}+1) = L \lvert \partial \Lambda \rvert + \mathcal O(V R)\, ,
\end{align}
where we recall that $V$ is the number of edges (or corners) of $\Lambda$.

Finally, for an edge set $E_{ \gamma_i,j }$, we translate and rotate it such that $E_{\gamma_i,j }=[0,R] \times [-R,0]$ and $L\partial \Lambda \cap E_{\text{edge}}=[0,R]\times\{0\}$. Let $x_0\coloneqq (R/2,-R/2) \in \mathbb R^2$. Due to one of the defining properties of edge sets, we have
\begin{align}
D_{R(m+1)}(x_0) \cap( L\Lambda)^\complement=D_{R(m+1)}(x_0) \cap \left( \mathbb R \times \mathbb R^+ \right)\,.
\end{align}
Thus, using the translational and rotational invariance of $\mathcal J_m$, as well as applying \autoref{cutoff box 2}, we can conclude
\begin{align}
\mathcal J_m \left(E_{\gamma_i,j } ,L \Lambda^\complement;K \right) = \mathcal J_m\left( [0,R) \times (-R,0) , \mathbb R\times \mathbb R^+;K \right) + \mathcal O(\exp(-\beta R) ) \,,
\end{align}
where the error term is uniform in $i$ and $j$.

We will now use the additivity of $\mathcal J_m$ and that it vanishes on zero sets in combination with \autoref{Jm halfspace} to see that
\begin{align}
\mathcal J_m & \left( [0,R) \times (-R,0),\mathbb R\times \mathbb R^+ ;K\right) \\
 & = \mathcal J_m\left( [0,R) \times (-\infty,0) , \mathbb R\times \mathbb R^+;K \right) - \mathcal J_m\left( [0,R) \times (-\infty,-R) , \mathbb R\times \mathbb R^+;K \right) \\
& = R \mathcal J_m\left( [0,1) \times (-\infty,0) , \mathbb R\times \mathbb R^+;K \right) + \mathcal O(R \exp(-\beta R))\,.
\end{align}
Again, as $R\ge C \ln(L)^2$, we the error term decays at least as $\mathcal O(L^{-1})$ as $L \to \infty$.

We recall that $V$ is the number of corners and $N = \sum_{i=1}^V (2N_{\gamma_i} +1)$ is the number of edge sets. Combining everything we have just shown, we arrive at 
\begin{align}
&\mathcal J_m(L \Lambda, L \Lambda^\complement;K)\\
&= \mathcal J_m(E_{\text{int}}, L \Lambda^\complement;K) + \sum_{i=1}^V  \Big( \mathcal J_m(E_{x_i}, L \Lambda^\complement;K) +  \sum_{j=-N_{\gamma_i}}^{N_{\gamma_i}} \mathcal J_m(E_{\gamma_i,j}, L \Lambda^\complement;K) \Big)  \\
&= \mathcal O(1) +  V \mathcal O(R \ln(R)) + N \left( R  \mathcal J_m\left( [0,1) \times (-\infty,0) , \mathbb R\times \mathbb R^+;K \right) + \mathcal O(L^{-1}) \right) \\
&= NR \mathcal J_m\left( [0,1) \times (-\infty,0) , \mathbb R\times \mathbb R^+ ;K\right) + \mathcal O( V  R \ln(R) +1)  \\
&= (L \lvert \partial \Lambda \rvert + \mathcal O(V R) ) \mathcal J_m\left( [0,1) \times (-\infty,0) , \mathbb R\times \mathbb R^+;K \right) + \mathcal O( V R \ln(R) ) \\
&= L \lvert \partial \Lambda \rvert \mathcal J_m\left( [0,1) \times (-\infty,0) , \mathbb R\times \mathbb R^+;K \right) + \mathcal O( VR( \ln(K)+ \ln(R) ) \,.
\end{align}

We used \autoref{Jm halfspace} to deal with the error term stemming from the expansion of $N$, which comes from \eqref{edge set number}.

The main term already agrees with the claim, but the error term still contains $R$. We now use that $R\ge K$ and $R \approx \max(K, \ln(L)^2)$ to conclude
\begin{align}
V R (\ln(K) + \ln(R)) \le C R \ln(R) \le C (K\ln(K) + \ln(L)^2 \ln ( \ln(L)^2 ) \le C(K\ln(K) + \ln(L)^3)\, .
\end{align}
The $\ln(L)^3$ is obviously not optimal, but it is not a relevant error term for our application.
\end{proof}

\subsection{Proof of \autoref{Jm domain reduction cor} for $\mathsf{C}^2$-smooth domains $\Lambda$}
The proof is similar to the polygonal case. However, instead of estimating the contribution of corners, we need to deal with flattening the $\mathsf{C}^2$-smooth boundary curve. We use rather crude estimates for the error terms, which will lead to the assumption $K^2<L$. The error estimates for some contributions will actually be rather sharp, but we guess that these error terms cancel each other and thus we could allow for a weaker assumption, at best $K\le C L$.

Let us begin with some technical results concerning $\mathsf{C}^2$-smooth domains. They are mainly stated for the convenience of the reader and introduce the notation that is used later on.

\begin{lemma}[$\mathsf{C}^2$-smooth tubular neighbourhood theorem] \label{tnt}
Let $\Lambda$ be a $\mathsf{C}^2$-smooth domain. Let $S_1, \dots, S_r$ be the connected components of $\partial \Lambda$. Then, there are $\varepsilon>0, C_\Lambda< \infty$ and for each $i=1, \dots ,r $ a $\mathsf{C}^1$-smooth function $g_i \colon [0, \lvert S_i \rvert) \times (-\varepsilon, \varepsilon) \to \mathbb R^2$, such that for any $i=1, \dots , r$, any $t,t_0 \in [0, \lvert S_i \rvert)$ we have 
\begin{enumerate}
\item $Dg_i(t,0) \in O(2)$, the orthogonal group,
\item $g_i$ is injective, it and its inverse have an $\varepsilon$-local Lipschitz constant of at most $2$  and for $i \neq i'$ the images of $g_i$ and $g_{i'}$ do not intersect, 
\item $s=0$ if and only if $g_i(t,s) \in \partial \Lambda$,
\item $s<0$ if and only if $g_i(t,s) \in \Lambda$, \label{s<0 case tnt}
\item $\operatorname{dist}(g(t,s), \partial \Lambda)= \lvert s \rvert$ and the image of $g_i$ is the $\varepsilon$-neighbourhood $D_\varepsilon(S_i)$, 
\item $\lVert Dg_i(t_0,0)-Dg_i(t,s) \rVert \le C_\Lambda( \lvert s \rvert + \lvert t- t_0 \rvert)$, \label{Dg is almost Lipschitz}
\item $\lVert Dg_i(0,0)-Dg_i(\lvert S_r \rvert-t ,s ) \rVert \le C_\Lambda(\lvert s \rvert + \lvert t \rvert)$.
\end{enumerate}
\end{lemma}

\begin{proof}
Each of the connected components $S_1, S_2 ,\dots, S_r$ is a closed loop in $\mathbb R^2$. For each such loop $S_i$, as $\partial \Lambda$ is $\mathsf{C}^2$-smooth, we can choose a parametrisation, that is, a periodic $\mathsf{C}^2$-smooth function  $ f_i \colon \mathbb R \to S_i$ which satisfies $\lVert f_i'(t) \rVert=1$ for all $t \in \mathbb R$. Thus, its period is $\lvert S_i \rvert$, the length of the loop and we can regard $f$ as an injective $\mathsf{C}^2$-smooth function on $T_i \coloneqq [0, \lvert S_i \rvert] / \{0, \lvert S_i \rvert\}$, the interval with identified endpoints. We proceed to define a function $g_i \colon T_i \times [-1,1] \to \mathbb R^2$ by setting
\begin{align}
g_i(t,s) \coloneqq f_i(t) + s \mathcal R f_i'(t)\,,
\end{align}
where $\mathcal R= \begin{pmatrix} 0 & 1 \\-1 & 0 \end{pmatrix} $ is the matrix associated to a $-\pi/2$ rotation. We know that $\mathcal R f_i'(t)$ has norm $1$, is continuous in $t$ and is always orthogonal to (the tangent line at) $S_i$. Thus, with respect to $\partial\Lambda$, it is either the inward normal vector for all $t \in T_i$ or the outward normal vector for all $t \in T_i$. We assume that it is the outward normal vector. We observe that  $g_i$ is $\mathsf{C}^1$-smooth and that
\begin{align}
Dg_i(t,0)= \left( f_i'(t) , \mathcal R f_i'(t) \right)
\end{align}
is an orthogonal matrix, as $\lVert f_i'(t) \rVert=1$ and thus the two column vectors form an orthonormal basis. In particular, $Dg_i(t,0)$ is always invertible. Thus, for each $t \in T_i$, there is an $\varepsilon_t>0$ such that $g_i$ is injective on $D_{\varepsilon_t}((t,0))\subset T_i \times [-1,1]$. This forms an open cover of the compact set $T_i \times \{0\}$, which means that there is a fixed $\varepsilon>0$, such that $g_i$ is injective on any disk of radius $7\varepsilon$. Assume $g_i(t,s)=g_i(t',s')$ for $(t,s),(t',s') \in T_i \times [-\varepsilon, \varepsilon]$. As $\lVert g_i(t,s)- f_i(t) \rVert \le \varepsilon$, we have $\lVert f_i(t)-f_i(t') \rVert \le 2 \varepsilon$. Since $f_i$ is injective, $\mathsf{C}^2$-smooth and $\lVert f_i'(t) \rVert=1$, for sufficiently small $\varepsilon>0$, we can conclude that $\operatorname{dist}_{T_i}(t,t') = \min\left( \lvert t-t' \rvert, \lvert S_i \rvert - \lvert t-t' \rvert \right)  \le 4 \varepsilon$. Thus, the distance between $(t,s)$ and $(t',s')$ is at most $6 \varepsilon$, which implies $(t,s)=(t',s')$, as $6<7$.

We have now proved that $g_i$ is injective on $T_i \times [-\varepsilon, \varepsilon]$ for sufficiently small $\varepsilon>0$. By choosing $\varepsilon_0>0$ even smaller, we can ensure that for any $\varepsilon<\varepsilon_0$ this holds for all $i=1, \dots , r$ simultaneously, that $g_i$ and $g_i^{-1}$ have $2$ as an $\varepsilon$-local Lipschitz constant\footnote{This is possible as $Dg_i(t,0) \in O(2)$ preserves the Euclidean  norm.}  and that $g_i(T_i \times (-\varepsilon, \varepsilon)) \cap g_{i'} (T_{i'} \times (-\varepsilon, \varepsilon))=\emptyset$ for $i \neq i'$. 

If $g_i(t,s) \in \partial \Lambda$, then, as there are $i', t'$ with $g_i(t,s)=g_{i'}(t',0)$, due to injectivity, we know that $(i,t,s)=(i',t',0)$ and thus $s=0$ if and only if $g_i(t,s) \in \partial \Lambda$. Because we choose $\mathcal R f_i'(t)$ to be the outward normal vector, we can conclude that for each fixed $t$, \eqref{s<0 case tnt} holds whenever $\lvert s \rvert$ is sufficiently small and thus, by continuity, it holds for all $\lvert s \rvert< \varepsilon$, as $g_i(t,s) \in \partial \Lambda$ implies $s=0$. 

For any  $x \in D_\varepsilon(\partial \Lambda)$, the closest boundary point $y \in \partial \Lambda$ has to be in some $S_i$ and thus $y=g_i(t,0)$. Furthermore, the line from $x$ to $y$ has to intersect $\partial \Lambda$ orthogonally at $y$. However, by definition, that is the line $s \mapsto g_i(t,s)$. Thus, as $\lvert s \rvert< \varepsilon$, the injectivity tells us that $g_i(t, \pm s)=x$, where the sign is negative if and only if $x \in \Lambda$. 

For the final two claims, we observe that
\begin{align}
\left \lVert Dg_i(t,s)- Dg_i(t_0,0) \right \rVert &= \left \lVert \left( f_i'(t)-f_i'(t_0)  + s \mathcal R f_i''(t) , \quad \mathcal R f_i'(t) - \mathcal R f_i'(t_0)   \right)  \right \rVert  \\
&\le 2 \lVert f_i'(t) - f_i'(t_0) \rVert + \lvert s \rvert \lVert f''(t) \rVert\,.
\end{align}
As $f''$ is uniformly bounded, due to the mean value theorem on one of the intervals $(t,t_0), (t_0,t)$ or $(t, \lvert S_r \rvert + t_0)$, the last one using periodicity of $f_i$, we can conclude the final two claims.
\end{proof}

\begin{lemma} \label{Jm domain dependence lem}
Let $\Lambda$ be a $\mathsf{C}^2$-smooth domain. Let $m,K \in \mathbb N$ and  $L\in \mathbb R^+$. Then, asymptotically as $L \to \infty$ and uniformly in $K$ as long as $K^2 =o(L)$, we have
\begin{align}
\mathcal J_m(L \Lambda, L \Lambda^\complement;K) = L \lvert \partial \Lambda \rvert \mathcal J_m \left( [0,1) \times \mathbb R^-, \mathbb R \times \mathbb R^+;K \right) + \mathcal O( K^2 + \ln(L)^4 )\,.
\end{align}  
\end{lemma}
%\begin{remark}
%\color{red} One can instead choose $\delta=L/(K(m+1))$. This increases the error term to $\mathcal O(L)$ and still needs the assumption $K^2<L$. It does, however, allow one to have no assumptions that bound $m$. Thus, one can achieve comparable estimates for polynomials of arbitrary order and even some power series. With the current choice of $\delta=10K/L$, we have upper bounds for $m$, which depend on $K,L,\Lambda$ (still allowing arbitrary large $m$, if $L$ is chosen large enough)
%\end{remark}
\begin{proof}  As mentioned, this proof is quite similar to the polygonal case, \autoref{reduction lem polygons}. Now, let $C_\Lambda, \varepsilon, (S_i,g_i)_{i=1}^r$ be given by \autoref{tnt}, $R\coloneqq \max(10K,\ln(L)^2), \delta \coloneqq R/L$  and  let $Q_R\coloneqq [0,R) \times (-R,0)$.\footnote{This square shows up a lot in this proof and we hope this notation improves readability.} 

The claimed error term is just $\mathcal O(R^2)$. As in the polygonal case, we can use \autoref{pf thm}, \autoref{Jm properties} and \autoref{Jm halfspace} to get the a-priori estimates
\begin{align}
\mathcal J_m(L \Lambda, L \Lambda^\complement;K) &\le CL \ln(L)\,, \\
L \lvert \partial \Lambda \rvert \, \mathcal J_m \left( [0,1) \times \mathbb R^-, \mathbb R \times \mathbb R^+ ;K\right)  &\le C L \ln(K)\,.
\end{align}
Thus, we can assume that $\delta < \varepsilon/(12(m+1))$ and $288C_\Lambda \delta (m+1)^2 <1/2$, as otherwise the error term is larger than both the main term and the actual result.

For $i=1, \dots , r$, we have $\delta <\varepsilon < \lvert S_i \rvert$ and we can thus choose $\delta_i \in [\delta, 2 \delta)$, such that $N_i \coloneqq \lvert S_i \rvert/\delta_i$ is an integer. For any $j \in \mathbb N$ with $0 <j \le N_i$, we define
\begin{align}
E_{ij} \coloneqq  L g_i\left( [\delta_i(j-1), \delta_i j) \times ( -\delta_i, 0 ) \right) \subset L \Lambda\,.
\end{align}
These correspond to the \emph{edge sets} in the polygonal case. Let $R_i \coloneqq L \delta_i$. Then, we have
\begin{align}
\sum_{i=1}^r N_i R_i =L \lvert \partial \Lambda \rvert\,. \label{sum of the Nis}
\end{align}
We now define our \emph{essential interior}. It is given by 
\begin{align}
E_{\text{int}} \coloneqq  &L \left( \Lambda \setminus \bigcup_{i=1}^r D_{\delta_i} (  S_i) \right)\subset L \Lambda \setminus D_R(L \partial \Lambda)\,,
\end{align}
since $D_R(L \partial \Lambda) \subset \bigcup_{i=1}^r L D_{\delta_i} (  S_i)$ as $\delta_i\ge \delta = R/L$. Its contribution to $\mathcal J_m$ can be estimated identically to the polygonal case, see \eqref{Jm int estimate.1}--\eqref{Jm int estimate.end}. This yields
\begin{align}
\mathcal J_m(E_{\text{int}}, L \Lambda^\complement;K) \le C\,. \label{Jm C2 int est eq}
\end{align}
As a simple consequence of the properties in \autoref{tnt}, we have $D_{R_i}(L S_i) \cap (L\Lambda)=L(D_{\delta_i}(S_i) \cap \Lambda ) = \bigcup_{i=1}^r \bigcup_{j=1}^{N_i} E_{ij}$ and thus
\begin{align}
L \Lambda = E_{\text{int}} \cup  \bigcup_{i=1}^r \bigcup_{j=1}^{N_i} E_{ij}\,,
\end{align}
being a disjoint union\footnote{In particular, we do not get any \emph{corner sets}, as $\mathsf{C}^2$-smooth domains do not have corners.}. 
Thus, as $\mathcal J_m(\blank, L \Lambda^\complement;K)$ is additive, we have
\begin{align}
\mathcal J_m(L \Lambda, L\Lambda^\complement;K) = \mathcal J_m(E_{\text{int}},  L\Lambda^\complement;K) + \sum_{i=1}^r \sum_{j=1}^{N_i} \mathcal J_m(E_{ij},  L\Lambda^\complement;K)\, . \label{Jm C2 domain split est eq}
\end{align}
We have already seen that $E_{\text{int}}$ can be absorbed into the error term. We will now show that 
\begin{align}
\mathcal J_m(E_{ij},L \Lambda^\complement;K)= R_i \mathcal J_m([0,1) \times \mathbb R^-, \mathbb R \times \mathbb R^+;K) + \mathcal O( \delta R^2)\,,
\end{align}
where the upper bound for the error is independent of $j$. In combination with \eqref{sum of the Nis}, this leads to the claim. Unlike in the polygonal case, the boundary curve is not straight. However, due to the assumption $K^2 \ll L$, we can approximate it sufficiently well by a smooth curve. Due to additivity and \autoref{Jm halfspace}, we know that
\begin{align}
\mathcal J_m(Q_{R_i},\mathbb R \times \mathbb R^+;K) &= \mathcal J_m([0,R_i)\times\R^-,\mathbb R \times \mathbb R^+;K) - \mathcal J_m([0,R_i)\times(-\infty,R_i),\mathbb R \times \mathbb R^+;K)
\\
& = R_i \mathcal J_m([0,1) \times \mathbb R^- , \mathbb R \times \mathbb R^+;K ) + \mathcal O(1)\,.
\end{align}
We now fix $i,j$ and take care that the error term bounds only depend on $i,R,L$ and the constants $\varepsilon, C_\Lambda$ in \autoref{tnt}.

By choosing an appropriate affine-linear unitary transformation $\mathcal A_{ij}$, we may assume $g_i(\delta_i(j-1),0)=0$ and $Dg_i(\delta_i(j-1),0)= \text{Id}$ (the $2 \times 2$ identity matrix) without changing the constants $\varepsilon, C$ in \autoref{tnt}.  As $g_i$ has a Lipschitz constant of at most $2$, we know $E_{ij} = L g_i((\delta_i(j-1),\delta_i j),(-R_i,0)) \subset LD_{3\delta_i}(0) =D_{3R_i}(0)$. We recall that $3R_i(m+1)< 6 R(m+1)<\varepsilon L/2$. Let $(E,E')$ be one of the set of pairs $(E_{ij}, L \Lambda^\complement)$ and $ \left( Q_{R_i},  \mathbb R \times \mathbb R^+ \right)$. Thus $E\subset D_{3R_i}(0)$. We utilize \autoref{cutoff box 2} to obtain
\begin{align}
\mathcal J_m(E,E';K) = \mathcal J_m(E,E'\cap D_{3R_i(m+1)}(0);K) + \mathcal O(m^5 \exp(-\beta R_i)) \,.
\end{align}
Using the Lipschitz-type property in \autoref{Jm properties}, we see that 
\begin{align}
&\left \lvert \mathcal J_m(E_{ij}, L \Lambda^\complement;K)- \mathcal J_m \left( Q_{R_i},  \mathbb R \times \mathbb R^+;K \right) \right \rvert \\
&\le \mathcal O(m^5 \exp(-\beta R_i))  + \frac 1 {2\pi} \left \lvert E_{ij} \Delta Q_{R_i} \right \rvert+ \frac m{2 \pi} \left \lvert \left( L \Lambda^\complement \cap D_{3R_i(m+1)}(0)\right)  \Delta \left(  \mathbb R \times \mathbb R^+ \cap D_{3R_i(m+1)} (0) \right) \right \rvert.
\end{align}
Thus, we have reduced our claim to a purely geometric estimate\footnote{This estimate is probably rather rough and could be the reason we require $K^2\ll L$. It only checks how close the boundary curve of $L \Lambda$ and $L\Lambda^\complement$ approach the same line locally, without fully using the fact that it is the same boundary curve. We would not be surprised if the error terms in this step cancel out to some order, but we have not yet found a better way to estimate them.}.

We define $\hat g_{ij} \colon D_{6R_i(m+1)}(0) \to \mathbb R^2$ by 
\begin{align}
 \hat g_{ij}(x_1,x_2)\coloneqq L g_i(x_1/L+\delta(j-1) , x_2/L)\,.
 \end{align}
 As $D_{6R_i(m+1)}(0)\subset L D_\varepsilon(\Lambda)$, we know that $\hat g_{ij}^{-1}$ is well-defined on $D_{6R_i(m+1)}(0)$. As $g_i$ and $g_i^{-1}$ have Lipschitz constants of at most $2$, so do the rescaled functions $\hat g_{ij},\hat g_{ij}^{-1}$. As $g_{ij}(0)=0$, we conclude that $D_{3R_i(m+1)}(0) \subset \hat g_{ij}(D_{6R_i(m+1)}(0))$ and  $\hat g_{ij} (D_{3R_i(m+1)}(0) )\subset D_{6R_i(m+1)}(0)$.  
This allows us to write
 \begin{align}
 E_{ij}&=\hat g_{ij} \left(Q_{R_i} \right) , \\
 L \Lambda^\complement \cap D_{3R_i(m+1)}(0)&= \hat g_{ij} \left(  \mathbb R \times \mathbb R^+ \cap D_{6R_i(m+1)}(0)\right) \cap D_{3R_i(m+1)}(0)\,.
 \end{align}
As $Dg_i(\delta (j-1), 0) =\text{Id}$, we have  for any $x \in D_{6R_i(m+1)}(0)$, 
\begin{align}
\lVert D \hat g_{ij} (x_1,x_2) - \text{Id} \rVert &=\lVert Dg_i(x_1/L+\delta(j-1), x_2/L) - Dg_i(\delta (j-1), 0) \rVert \\
&\le C_\Lambda (\lvert x_1/L \rvert + \lvert x_2/L \rvert) \le 4C_\Lambda (3R_i(m+1))/L \le 24C_\Lambda \delta (m+1)\,.
\end{align}
Let $ h_{ij}(x) \coloneqq \hat g_{ij}(x)-x$. Then, by the mean value theorem we can conclude for any $x \in D_{6R_i(m+1)}(0)$,
\begin{align}
\lVert \hat g_{ij}(x) - x \rVert &=\lVert  h_{ij}(x) \rVert = \lVert  h_{ij}(x)- h_{ij}(0) \rVert \le (24C_\Lambda \delta (m+1)) (6R_i(m+1)) 
\\
&\le  288C_\Lambda \delta^2 (m+1)^2 L<R/2\,.
\end{align}
The last inequality is based on the assumption $288C_\Lambda \delta(m+1)^2<1/2$, which is stated earlier in the proof.

 Let $r \coloneqq  288C_\Lambda \delta^2 (m+1)^2 L$ 
and let $\mathcal E \subset \mathbb R^2$. We want to estimate $\left \lvert \left( \mathcal E \Delta \hat g_{ij}(\mathcal E) \right) \cap D_{3R_i(m+1)}(0) \right \rvert$ for $\mathcal E =  Q_{R_i} $ and $\mathcal E= \mathbb R \times \mathbb R^+$.

 Let $y \in  \mathcal E \setminus \hat g_{ij}(\mathcal E)$ with $y \in D_{3R_i(m+1)}(0)$. Then, there is an $x \in D_{6R_i(m+1)}(0) \cap \mathcal E^\complement$ with $\hat g_{ij}(x)=y$. As $\lVert x-y \rVert =\lVert x- \hat g_{ij}(x) \rVert \le r$, we see that $y \in D_r(\mathcal E^\complement)$. Thus, as $y \in \mathcal E$, we see $y \in D_r(\partial \mathcal E)$. On the other hand, if $y \in \hat g_{ij}(\mathcal E) \setminus \mathcal E$ with $y \in D_{3R_i(m+1)}(0)$, there is an $x \in D_{6R_i(m+1)}(0)\cap \mathcal E$ with $\hat g_{ij}(x)=y$ and thus $\lVert x-y \rVert \le r$, which implies $y \in D_r(\partial \mathcal E)$, again. Thus, we have shown that 
 \begin{align}
  \left( \mathcal E \Delta \hat g_{ij}(\mathcal E) \right) \cap D_{3R_i(m+1)}(0) \subset D_r(\partial \mathcal E) \cap D_{3R_i(m+1)}(0)\,,
  \end{align}
  which leads to
  \begin{align}
  \left \lvert   \left( \mathcal E \Delta \hat g_{ij}(\mathcal E) \right) \cap D_{3R_i(m+1)}(0) \right \rvert \le \left \lvert  D_r(\partial \mathcal E) \cap D_{3R_i(m+1)}(0)\right \rvert\,.
  \end{align}
Using this for the half space, we easily get
 \begin{align}
 &\left \lvert \left( L \Lambda^\complement \cap D_{3R_i(m+1)}(0)\right)  \Delta \left(  \mathbb R \times \mathbb R^+ \cap D_{3R_i(m+1)}(0) \right) \right \rvert \\
 &\le   \left \lvert  D_r(\mathbb R \times \{0\}) \cap D_{3R_i(m+1)}(0)\right \rvert \\
 &\le 4r (3R_i(m+1)) \le 4 \cdot 288C_\Lambda \delta^2 (m+1)^2 L \cdot 6R(m+1) < Cm^3 \delta R^2 \,.
 \end{align}
 In the last step we used that $\delta=R/L$.
 
 For the square, we use that $r<R/2$, which means that $r$ is less than half of the side length of the square. This allows us to estimate
 \begin{align}
 \left \lvert E_{ij} \Delta Q_{R_i} \right \rvert \le  \left \lvert D_r \left( \partial Q_{R_i} \right) \right \rvert 
= 8Rr- (4 -\pi)r^2  < 8Rr=8R \cdot 288C_\Lambda \delta^2 (m+1)^2 L < C m^2 \delta R^2\,.
 \end{align}
Thus, we have completed the proof of the asymptotic expansion
\begin{align}
\mathcal J_m(E_{ij},L \Lambda^\complement;K)=R_i \mathcal J_m([0,1)\times \mathbb R^-, \mathbb R\times \mathbb R^+;K) + \mathcal O ( \delta R^2)\,,
\end{align}
where the upper bound for the error term is independent of $i,j$. Thus, we can sum this expression over $i,j$ and recalling \eqref{Jm C2 domain split est eq}, \eqref{Jm C2 int est eq} and \eqref{sum of the Nis} to observe
\begin{align}
\mathcal J_m(L \Lambda, L \Lambda^\complement;K) &= \sum_{i=1}^r \sum_{j=1}^{N_i} \mathcal J_m(E_{ij} , L \Lambda^\complement;K)   +\mathcal J_m(E_{\text{int}},L \Lambda^\complement;K)\\
&= \left( \sum_{i=1}^r \sum_{j=1}^{N_i} R_i \mathcal J_m([0,1)\times \mathbb R^-, \mathbb R\times \mathbb R^+;K) + \mathcal O ( \delta R^2) \right) + \mathcal O(1) \\
&= \left( \sum_{i=1}^r  N_iR_i \right)  \mathcal J_m([0,1)\times \mathbb R^-, \mathbb R\times \mathbb R^+;K) + \mathcal O(R^2 \sum_{i=1}^r N_i \delta_i ) + \mathcal O(1)\label{7.69} \\
&= L \lvert \partial \Lambda \rvert  \mathcal J_m([0,1)\times \mathbb R^-, \mathbb R\times \mathbb R^+;K)+ \mathcal O(R^2 L\lvert \partial \Lambda \rvert /L ) +\mathcal O(1) \\
&= L \lvert \partial \Lambda \rvert  \mathcal J_m([0,1)\times \mathbb R^-, \mathbb R\times \mathbb R^+;K)+\mathcal O(R^2)\,.
\end{align}
In \eqref{7.69}, we used $\delta< \delta_i$ in the error term to apply \eqref{sum of the Nis}. As $R^2 \le C (K^2+ \ln(L)^4)$, this finishes the proof of this lemma.
\end{proof}

Altogether, \autoref{reduction lem polygons} and \autoref{Jm domain dependence lem} finally prove \autoref{Jm domain reduction cor}.

%Jm = MI Beweis, recht lang. Beinhaltet auch appendix A
\section{\texorpdfstring{$K \ll L$}{K ^^e2^^89^^aa L}: Proof of \autoref{Jm=I thm}} \label{dependence on K section}

%\section{\texorpdfstring{$K \ll L$}{K ^^e2^^89^^aa L}: The dependence on $K$ in $\mathcal J_m(\cdot,\cdot;K)$: Proof of \autoref{Jm=I thm}} \label{dependence on K section}

In this section we will prove \autoref{Jm=I thm}. Let us recall that it states

\JMtoIthm*

In accordance with \cite[(2.11)]{Leschke2021}, we write 
\begin{align}
\mathcal K_K \coloneqq \sum_{\ell=0}^{K-1} |\psi_\ell \rangle \langle \psi_\ell |\,,
\end{align}
where $\psi_\ell$ are the Hermite functions (as in \cite[(2.9)]{Leschke2021}) given by
\begin{align}
\psi_\ell(s) \coloneqq \left(\sqrt \pi  2^\ell \ell ! \right) ^{-1/2}  H_\ell(s) \exp(-s^2/2) \,,\quad s\in\R\,.
\end{align}
For any polynomial $f$ with $f(0)=f(1)=0$, we define
\begin{align}\label{def M_K}
\mathsf{M}_{<K}(f) \coloneqq \frac{1}{2\pi}\int_{\mathbb R } \mathrm d p\, \operatorname{tr} f(1_{>p} \mathcal K_K 1_{>p}) \,,
\end{align}
where $1_{>p}\coloneqq 1_{(p,\infty)}$. This $\mathsf{M}_{<K}(f)$ agrees with $\mathsf{M}_{\le K-1}(f)$ defined in \cite[(2.12)]{Leschke2021}.

While there are asymptotic formulas for the Hermite polynomials directly, we did not find the exact statement we needed. However, as we already go into detail on the asymptotics for the Laguerre polynomials, it is convenient to reduce the Hermite polynomials to Laguerre polynomials and only go deep into the asymptotic expansion of one of these polynomials.

\begin{lemma}  \label{MK=Jm}
The asymptotic scaling coefficient for fixed $K$ agrees with the one shown in \cite{Leschke2021}. That is,
\begin{align}
\mathcal J_m([0,1) \times \mathbb R^-, \mathbb R \times \mathbb R^+;K )=  \frac 1 {\sqrt{K}}  \mathsf{M}_{<K} (t \mapsto t(1-t)^m )\,.
\end{align}
\end{lemma}

\begin{remark} 
We will present a full proof of this statement, which is based on a simple comparison of coefficients and a sketch of an alternative proof. 
\end{remark}

\begin{proof}
Let $f_m(t) \coloneqq t (1-t)^m  $ for any $t \in [0,1]$. Consider the domain $\Lambda_0=D_1(0)$, the unit disk. This is a $\mathsf{C}^\infty$-smooth domain and thus, we can apply both \autoref{Jm domain dependence lem} as well as \cite[Theorem 2]{Leschke2021} to it. We now fix $K \in \mathbb N$ and consider
\begin{align}
\lim_{L \to \infty} \frac 1 {2 \pi L}  \mathcal J_m(L \Lambda, L \Lambda ^\complement;K)\,.
\end{align}
On the one hand, according to \autoref{Jm domain dependence lem}, we observe
\begin{align}
\lim_{L \to \infty} \frac 1 {2 \pi L}  \mathcal J_m(L \Lambda, L \Lambda ^\complement;K)= \mathcal J_m([0,1)\times \mathbb R^-, \mathbb R \times \mathbb R^+;K)\,. 
\end{align}
On the other hand, according to \cite[Theorem 2]{Leschke2021}, we have
\begin{align}
\lim_{L \to \infty} \frac 1 {2 \pi L}  \mathcal J_m(L \Lambda, L \Lambda ^\complement;K) = \lim_{L \to \infty} \frac 1 {2 \pi L} \operatorname{tr} f_m(1_{L\Lambda} P_K 1_{L \Lambda}) = \frac{1}{\sqrt K}\, \mathsf{M}_{<K}(f_m)\,,
\end{align}
where we used that $K=1/B$. This completes the coefficient comparison proof.
\end{proof}
\begin{proof}[Sketch of an alternative proof]
One can also directly transform the two integral representations into one another. Most of this work has been executed in the proof of Lemmata 5 \& 6 in \cite{Leschke2021}. However, they use Roocaforte's  approximation and get an error term, that they bound using the exponential decay of $P_K$ for fixed $K$, while  our error term needs to be bounded in $K$. 

Let $E \coloneqq [0,1) \times \mathbb R^-$ and $E' \coloneqq \mathbb R \times \mathbb R^+$. According to Mercer's theorem, we get
\begin{align}
\mathcal J_m&\left( [0,1) \times \mathbb R^-, \mathbb R \times \mathbb R^+;K \right) \,=\, \operatorname{tr} 1_E (P_K 1_{E'})^m P_K 1_E \\
&= \int_E \mathrm dx \int_{E'} \mathrm d x_1 \int_{E'} \mathrm d x_2\cdots \int_{E'} \mathrm d x_m \, P_K(x,x_1) P_K(x_1,x_2) \cdots P_K(x_m,x) \,.
\end{align}

This integral now looks very similar to the one studied in \cite[Proof of Lemma 5]{Leschke2021}. The key advantage is that due to our choice of $E,E'$, the Roccaforte approximation is exact, which means basically their (3.7) does not carry an error term. However, as we are studying a different base polynomial ($f_m$ instead of $t\mapsto t^m$), it looks slightly different. The main idea is that the intersection of the offset sets depends entirely linearly  on the maximum offset in $e_2$ direction. 

They have performed all the remaining integral transformations to carry out the remaining parts of the proof.
\end{proof}

The symmetry of the coefficient $\mathsf{M}_{<K}(f)$ in the next lemma is of independent interest. It is inspired by the same and obvious symmetry relation of the functional $\mathsf{I}$ of \eqref{def I(f)}. But is is also useful from a technical point of view as it simplifies the construction of suitable intervals as in \autoref{def: intervals}.

\begin{lemma} \label{MK symmetry}
Let $K \in \mathbb N$ and let $f$ be a function with $f(0)=f(1)=0$ and $\lvert f(t) \rvert \le Ct^\alpha (1-t)^\alpha$ for any $t \in [0,1]$ and some $\alpha>0$. We define the function $g$ by $g(t)\coloneqq f(1-t)$ for any $t \in [0,1]$. Then, the following equalities hold,
\begin{align}
\mathsf{M}_{<K}(f)=\mathsf{M}_{<K}(g) = \frac 1 {2\pi} \int_0^\infty \mathrm d p\,\big[ \operatorname{tr} f( 1_{>p} \mathcal K_K 1_{>p} ) + \operatorname{tr} g(  1_{>p} \mathcal K_K 1_{>p})\big]  \,.
\end{align}
\end{lemma}

Due to the symmetry of $M_{\le K-1}(f)$ in $f$, meaning $\mathsf{M}_{<K}(f)=\mathsf{M}_{<K}(t \mapsto f(1-t))$, for any polynomial $f$ with $f(0)=f(1)=0$, it suffices to consider $\mathsf{M}_{<K}(f)$ for symmetric polynomials (meaning $f(t)=f(1-t)$). Thus, we can restrict it to the case $f(t)=(t(1-t))^m$ for some $m \in \mathbb N$. Using the same idea as in \eqref{tr(f).1}, we can establish
\begin{align}
 2\pi \,\mathsf{M}_{<K}(t \mapsto [t(1-t)]^m ) = \int_{\mathbb R} \mathrm d p\,  \operatorname{tr} \left(1_{<p} \mathcal K_K 1_{>p} \mathcal K_K 1_{<p} \right)^m \,.
\end{align}
The advantage of the expression is that each occurrence of $\mathcal K_K$ is flanked by $1_{< \xi}$ and $1_{\ge \xi}$, which should help with estimating the Hilbert--Schmidt norms of error terms in the asymptotic expansion of the kernel $\mathcal K_K$ for large $K$.
\begin{proof}[Proof of \autoref{MK symmetry}]
Let us first point out that $g$ also satisfies $\lvert g(t) \rvert \le C t^\alpha (1-t)^\alpha$ for any $t \in [0,1]$. Thus, \cite[Lemma 3]{Leschke2021} tells us that $\mathsf{M}_{<K}(f)$ and $\mathsf{M}_{<K}(g)$ are well-defined.

We recall that by definition
\begin{align}
 2\pi\, \mathsf{M}_{<K}(f) = \int_{\mathbb R } \mathrm d p\, \operatorname{tr} f(1_{>p} \mathcal K_K 1_{>p}) \,.
\end{align}
For any $\ell \in \{0,1,\dots,K-1\}$ and any $x,y \in \mathbb R$, we have $\psi_\ell(-x)=(-1)^{\ell} \psi_\ell(-x)$ and thus 
\begin{align}
\mathcal K_K (-x,-y)=\sum_{\ell=0}^{K-1} (-1)^{2\ell} \psi_\ell(x) \psi_\ell(y) = \mathcal K_K(x,y)\, .
\end{align}
Thus, $\mathcal K_K$ commutes with the reflection operator $\mathcal R$ on $\Lp^2(\mathbb R)$, which is defined by $\mathcal R(\phi)(x)\coloneqq \phi(-x)$.
 
For any projections $P,Q$, the eigenvalues (including multiplicities, except for $0$) of the operators $PQP$ and $QPQ$ agree, as they are both given by the squares of the singular values of $PQ$ (or equivalently $QP=(PQ)^*$). Thus, for any function $f$ with $f(0)=0$, we have $\operatorname{tr}f(PQP)=\operatorname{tr}f(QPQ)$, if the traces exist. In our case $Q=\mathcal K_K$ is finite rank and thus $f(QPQ)$ and $f(PQP)$ are both finite rank and in particular trace class. As $f(0)=f(1)=0$, we can also show
\begin{align}
f(Q-QPQ)=f(\mathds{1}-QPQ)\,. \label{Q-QPQ.1}
\end{align}
As the operators $Q-QPQ$ and $\mathds{1}-QPQ$ both commute with $Q$, it suffices to prove this on the eigenspaces of $Q$. Restricted to the eigenspace of eigenvalue one (the image) of $Q$, both operators agree. Restricted to the kernel of $Q$, however, the first operator becomes $0$ while the second one becomes $\mathds{1}$. As $f(0)=f(1)=0$, the identity still holds.

Now, we get
\begin{align}
 2\pi\, \mathsf{M}_{<K}(f)=&\int_{\mathbb R} \mathrm d p \, \operatorname{tr} f(1_{>p} \mathcal K_K 1_{>p}) \\
=&\int_{0}^\infty \mathrm d p \, \operatorname{tr} f(1_{>p} \mathcal K_K 1_{>p}) + \int_{-\infty}^0 \mathrm d p\, \operatorname{tr} f( \mathcal K_K 1_{>p} \mathcal K_K) \\
=&\int_{0}^\infty \mathrm d p \,\operatorname{tr} f(1_{>p} \mathcal K_K 1_{>p}) + \int_0^{\infty} \mathrm d p\, \operatorname{tr} f( \mathcal K_K 1_{>-p} \mathcal K_K) \\
=&\int_{0}^\infty \mathrm d p \,\big[\operatorname{tr} f(1_{>p} \mathcal K_K 1_{>p}) + \operatorname{tr} f( \mathcal K_K 1_{<p} \mathcal K_K)\big] \\
=&\int_{0}^\infty \mathrm d p \,\big[\operatorname{tr} f(1_{>p} \mathcal K_K 1_{>p}) + \operatorname{tr} f(\mathcal K_K - \mathcal K_K 1_{>p} \mathcal K_K)\big] \\
=&\int_{0}^\infty \mathrm d p \,\big[\operatorname{tr} f(1_{>p} \mathcal K_K 1_{>p}) +\operatorname{tr} f(\mathds{1}- \mathcal K_K 1_{>p} \mathcal K_K)\big] \\
=&\int_{0}^\infty \mathrm d p \,\big[\operatorname{tr} f(1_{>p} \mathcal K_K 1_{>p}) +  \operatorname{tr} g( \mathcal K_K 1_{>p} \mathcal K_K)\big] \\
=&\int_{0}^\infty \mathrm d p \,\big[\operatorname{tr} f(1_{>p} \mathcal K_K 1_{>p}) + \operatorname{tr} g(1_{>p} \mathcal K_K 1_{>p})\big] = 2\pi \, \mathsf{M}_{<K}(g)\,.
\end{align}

In the second step, we perform a substitution $p \mapsto -p$. Then, we conjugate the expression inside the trace with $\mathcal R$ and use that $\mathcal R 1_{>-p}\mathcal R = 1_{<p}$. Afterwards, we insert $1_{<p}+1_{>p}=1$ followed by \eqref{Q-QPQ.1}, which brings us to the closure.
\end{proof}

The following theorem is proved in \autoref{main term is MI proof section}. It relies on the study of the asymptotic behaviour of $\mathcal K_K$, which is based on the Laguerre asymptotics we studied in \autoref{kernel asymp section}.

\begin{theorem} \label{main term is MI}  
There are families of intervals $(I_p,J_p)_{p \in \R^+}$, depending on $K$, such that
\begin{align}
\frac 1 {\pi \sqrt K}   \int_0^\infty \mathrm d p \, \operatorname{tr}  \left \lvert  1_{I_p} \mathcal K_K 1_{J_p} \right \rvert^{2m} = \frac{2\sqrt{2}}{\pi}\, \mathsf{I}(t\mapsto [t(1-t)]^m) \ln(K)+ \mathcal O(\ln\ln(K))\,,
\end{align}
with $\mathsf{I}(f)$ defined in \eqref{def I(f)}.
\end{theorem}

This theorem essentially tells us that the restriction to some intervals $(I_p,J_p)$ already contains the expected main term. To deal with the error term resulting from this restriction, we will first show that this error term for arbitrary $m \in \mathbb N$ can be bounded by the one for $m=1$. This is achieved by the next lemma. The case $m=1$ corresponds to the function $f(t)=t(1-t)$, which we already studied in \autoref{pf thm}. This allows us to conclude that the error term for $m=1$ is small, which will finish this section.

\begin{lemma} \label{reduction to EK}
Let us define for $K\in\N$,
\begin{align}
\mathcal E(K) \coloneqq \frac 1 {\pi \sqrt K}   \int_0^\infty \mathrm d p\, \operatorname{tr} \left(  \left \lvert 1_{<p} \mathcal K_K 1_{>p} \right \rvert^2 -  \left \lvert  1_{I_p} \mathcal K_K 1_{J_p} \right \rvert^{2}\right)  \, .
\end{align}
Then,  $\mathcal E(K)>0$ and for any $m \ge 2$, we have
\begin{align}
\frac 1 {\sqrt{K}} \mathsf{M}_{<K}(t \mapsto [t(1-t)]^m )= \frac 1 {\pi \sqrt K}   \int_0^\infty \mathrm d p\, \operatorname{tr}  \left \lvert  1_{I_p} \mathcal K_K 1_{J_p} \right \rvert^{2m} +  \mathcal O(m \mathcal E(K)) \, ,
\end{align}
while for $m=1$ we have
\begin{align}
\frac 1 {\sqrt{K}}  \mathsf{M}_{<K}(t \mapsto t(1-t) )= \frac 1 {\pi \sqrt K}   \int_0^\infty \mathrm d p\, \operatorname{tr}  \left \lvert  1_{I_p} \mathcal K_K 1_{J_p} \right \rvert^{2m}  + \mathcal E(K) \, .
\end{align}
\end{lemma}

\begin{proof}
Let $p \in \mathbb R$ and  $m \in \mathbb N$. With $\operatorname{tr} \lvert A \rvert^{2m} = \operatorname{tr} \lvert A^* \rvert^{2m}$, we observe
\begin{align}
&\left \lvert \operatorname{tr} \left \lvert 1_{<p} \mathcal K_K 1_{>p} \right \rvert^{2m} - \left \lvert 1_{I_p} \mathcal K_K 1_{J_p} \right \rvert^{2m} \right \rvert \\
 &\le  \left \lvert \operatorname{tr}  \left \lvert 1_{<p} \mathcal K_K 1_{>p} \right \rvert^{2m}  -  \left \lvert 1_{<p} \mathcal K_K 1_{J_p} \right \rvert^{2m} \right \rvert + \left \lvert \operatorname{tr}  \left \lvert 1_{J_p} \mathcal K_K 1_{<p} \right \rvert^{2m}  -  \left \lvert 1_{J_p} \mathcal K_K 1_{I_p} \right \rvert^{2m}  \right \rvert  \\
&\le  m \left( \left \lVert 1_{<p} \mathcal K_K \left( 1_{>p} - 1_{J_p} \right) \mathcal K_K 1_{<p} \right \rVert_1 + \left \lVert  1_{J_p} \mathcal K_K \left( 1_{<p} - 1_{I_p} \right) \mathcal K_K 1_{J_p} \right \rVert_1  \right) \\
&=m \left( \operatorname{tr} 1_{<p} \mathcal K_K \left( 1_{>p} - 1_{J_p} \right) \mathcal K_K 1_{<p} + \operatorname{tr}  1_{J_p} \mathcal K_K \left( 1_{<p} - 1_{I_p} \right) \mathcal K_K 1_{J_p} \right) \\
&= m \left( \operatorname{tr}\left \lvert 1_{<p} \mathcal K_K 1_{>p} \right \rvert^2 - \operatorname{tr}\left \lvert 1_{<p} \mathcal K_K 1_{J_p} \right \rvert^2 +  \operatorname{tr}\left \lvert 1_{J_p} \mathcal K_K 1_{<p} \right \rvert^2 -  \operatorname{tr}\left \lvert 1_{J_p} \mathcal K_K 1_{I_p} \right \rvert^2 \right)  \\
&= m \operatorname{tr} \left( \left \lvert 1_{<p} \mathcal K_K 1_{>p} \right \rvert^2  - \left \lvert 1_{I_p} \mathcal K_K 1_{J_p} \right \rvert^2 \right)  \, .
 \end{align}
 The second step uses the intermediate estimate in \autoref{tele sum lem1}. As $I_p \subset (-\infty,p)$ and $J_p \subset (p,\infty)$, the operators inside the trace norms are positive. Thus, their trace norms are just their trace. 

Due to \autoref{MK symmetry}, we can conclude
\begin{align}
&\frac 1 {\sqrt{K}} \mathsf{M}_{<K}(t \mapsto [t(1-t)]^m )\\
&= \frac 1 {\pi \sqrt K}    \int_0^\infty  \mathrm d p\, \Big[\operatorname{tr} \left \lvert 1_{I_p} \mathcal K_K 1_{J_p} \right \rvert^{2m}  +  \operatorname{tr} \left( \left \lvert 1_{<p} \mathcal K_K 1_{>p} \right \rvert^{2m}  -  \left \lvert 1_{I_p} \mathcal K_K 1_{J_p} \right \rvert^{2m}  \right) \Big] \label{reduction to EK proof eq1} \\
&=\frac 1 {\pi \sqrt K}   \int_0^\infty  \mathrm d p\, \operatorname{tr}  \left \lvert  1_{I_p} \mathcal K_K 1_{J_p} \right \rvert^{2m}  +  \mathcal O(m \mathcal E(K)) \, .
\end{align}
The case $m=1$ is just \eqref{reduction to EK proof eq1}.
\end{proof}

Here comes the trick to bound the error term.
\begin{lemma}  \label{EK small trick lem}
For the function $\mathcal E$ defined in \autoref{reduction to EK}, we have $\mathcal E(K) \le C \ln\ln(K)$ for any $K\in\N$ with $K \ge 3$.
\end{lemma}
While one might be able to study the asymptotics of $\mathcal K_K$ directly and show this estimate, we found a significantly more elegant solution.
\begin{proof}
We consider two approximations for $(1/\sqrt{K}) \mathsf{M}_{<K} (t \mapsto t(1-t)) $. On the one hand, \autoref{main term is MI}, \autoref{reduction to EK} and $ \frac{2\sqrt{2}}{\pi}\,\mathsf{I}(t\mapsto t(1-t))= 1/(\sqrt 2 \pi^3)$ tell us that
\begin{align}
\frac1{\sqrt{K}}\, \mathsf{M}_{<K} (t \mapsto t(1-t)) &= \frac 1 {\pi \sqrt K}   \int_0^\infty  \mathrm d p\, \operatorname{tr}  \left \lvert  1_{I_p} \mathcal K_K 1_{J_p} \right \rvert^{2}   + \mathcal E(K) \\
&=\frac{1}{\sqrt 2 \pi^3}\, \ln(K)+ \mathcal O(\ln\ln(K)) +  \mathcal E(K)  \, .
\end{align}
On the other hand, \autoref{MK=Jm}, \autoref{Jm domain dependence lem}, \eqref{def Jm} and \autoref{pf thm} with $\Lambda=D_1(0)$ the unit disk imply
\begin{align}
\frac1{\sqrt{K}}\, \mathsf{M}_{<K} (t \mapsto t(1-t))&= \mathcal J_1 \left( [0,1) \times \R^-, \R\times \R^+ ;K \right) \\
&= \lim_{L \to \infty} \frac{1}{2 \pi L}\, \mathcal J_1( D_L(0), D_L^\complement(0);K ) \\
&= \lim_{L \to \infty} \frac{1}{2 \pi L}\,\operatorname{tr} 1_{D_L(0)} P_K 1_{D_L^\complement(0)} P_K 1_{D_L(0)} \\
&= \frac{1}{\sqrt 2 \pi^3}\, \ln(K) + \mathcal O(1) \, .
\end{align}
From these two approximations, we can infer 
\begin{align}
 \mathcal E(K) \le C \ln\ln(K) \, .
\end{align}
\end{proof}

We can now conclude this section with the
\begin{proof} [Proof of \autoref{Jm=I thm}]
We will first show that, for any polynomial $f$ with $f(0)=f(1)=0$, we have
\begin{align} \label{Jm=I thm proof eq1}
\frac{1}{\sqrt{K}}\, \mathsf{M}_{<K} ( f) = \frac{2\sqrt{2}}{\pi}\,\mathsf{I}(t \mapsto (f) + \mathcal O(\ln\ln(K)) \,.
\end{align}
Due to \autoref{MK symmetry}, we know $\mathsf{M}_{<K}(f)= \mathsf{M}_{<K} (t \mapsto f(1-t))$ and $\mathsf{I}(f)= \mathsf{I}(t \mapsto f(1-t))$ is obvious from \eqref{def I(f)}. Thus, we can assume without loss of generality that $f(t)=f(1-t)$ by replacing $f$ with $t \mapsto (f(t)+f(1-t))/2$. Due to linearity, it further suffices to show \eqref{Jm=I thm proof eq1} for a basis of all polynomials with $f(0)=f(1)=0$ and $f(t)=f(1-t)$. Such a basis is given by the polynomials $t \mapsto [t(1-t)]^m$ for $m \in \mathbb N$. Due to \autoref{reduction to EK} and \autoref{EK small trick lem}, we can conclude
\begin{align}
\frac{1}{\sqrt{K}}\, \mathsf{M}_{<K} ( t \mapsto [t(1-t)]^m) = \frac{2\sqrt{2}}{\pi}\,\mathsf{I}(t \mapsto [t(1-t)]^m ) + \mathcal O(\ln\ln(K)) 
\end{align}
and thus \eqref{Jm=I thm proof eq1}. To conclude the proof, we insert $f(t)=t(1-t)^m$ into \eqref{Jm=I thm proof eq1} and apply \autoref{MK=Jm}, which tells us
\begin{align}
\mathcal J_m([0,1) \times \R^- , \R \times \R^+ ;K)= \frac{2\sqrt{2}}{\pi}\, \mathsf{I}( t \mapsto t(1-t)^m )\, \ln(K) + \mathcal O(\ln \ln(K)) \,. 
\end{align}
\end{proof}

%\newpage

\appendix
\section{Two simple trace(-norm) inequalities}

\begin{lemma} \label{tele sum lem1}
Let $A,B$ be two Hilbert--Schmidt operators with $\lVert A \rVert \le 1, \lVert B \rVert \le  1 $ and let $m \in \mathbb N$. Then, we have
\begin{align}
\lVert (A^*A)^m - (B^*B)^m \rVert_1 \le m \lVert A^*A - B^*B \rVert_1 \le m (\lVert A \rVert_2 + \lVert B \rVert_2 ) \lVert A-B \rVert_2 \, . 
\end{align}
\end{lemma}
\begin{proof}
We just observe that 
\begin{align}
\lVert (A^*A)^m - (B^*B)^m \rVert_1 &\le \sum_{j=1}^m \left \lVert (A^*A)^{j-1} ( A^*A-B^*B) (B^*B)^{m-j} \right \rVert_1 \\
&\le m \lVert A^*A-B^*B \rVert_1 \\
&\le  m \left( \lVert A^*(A-B) \rVert_1 + \lVert (A-B)^*B \rVert_1 \right) \\
&\le m \left( \lVert A \rVert_2 \lVert A-B \rVert_2 + \lVert A-B \rVert_2 \lVert B \rVert_2 \right) \\
&= m (\lVert A \rVert_2 + \lVert B \rVert_2 ) \lVert A-B \rVert_2 \, ,
\end{align}
which completes the proof.
\end{proof}

\begin{lemma} \label{tele sum lem2}
Let $A_1, A_2$ be two Hilbert--Schmidt operators. Then for $m\in\mathbb N$, $m\ge2$,
\begin{align}
\lvert \operatorname{tr} (A_1^m -A_2^m ) \rvert \le m \max(\lVert A_1 \rVert , \lVert A_2 \rVert )^{m-2} \, \max(\lVert A_1 \rVert_2, \lVert A_2 \rVert_2)  \lVert A_1 - A_2 \rVert_2\,.\label{K>L op algebra est}
\end{align}
\end{lemma}

\begin{proof} The proof goes along the same line as the previous one. We write the difference $A_1^m-A_2^m$ in the form $\sum_{i=1}^{m} (A_1^{m-i+1}A_2^{i-1} - A_1^{m-i}A_2^i)$ and use the triangle and H\"older inequality. 
\end{proof}

\section{Sine-kernel asymptotics and the leading asymptotic coefficient} 
\label{main term is MI proof section}

In this rather long section we finally prove \autoref{main term is MI}. To this end, we start with results on the kernel $\mathcal K_K$ and prove the convergence to the sine-kernel on a global scale, see \autoref{hat K kernel est}. In the second subsection, this is used to deal with the asymptotics of an integral of certain traces involving this kernel for large $K$ and we evaluate the asymptotic coefficient $\mathsf{M}_{<K}(f)$ to leading order in $K$ for polynomials $f$. The proof is based on a seminal result by Landau and Widom \cite{LandauWidom} and an improvement by Widom \cite{Widom1982}.

\subsection{Sine-kernel asymptotics on a global scale}

\bigskip For the results in this section we were inspired by the results of Kriecherbauer, Schubert, Schüler, and Venker in \cite{Kriecherbauer2014}. While the convergence of $\mathcal K_K$ ``in the bulk" $(-(2K-1)\pi/4,(2K-1)\pi/4)$ to the sine-kernel is standard on a local scale of order one, this is not the case on the much larger scale of order $K$ needed here. To go this scale, we use the function $\eta_{2K\pm1}$. 

Recall that $\eta(s)= \xi(s^2)$, see \eqref{xi in eta}. The advantage of $\eta$ over $\xi$ is that $\eta'(s)=\sqrt{1-s^2}$ for $|s|\le1$. For any $\lambda \in \mathbb R^+$, let $\eta_\lambda$ be the rescaled function
\begin{align}\label{def eta_lambda}
\eta_\lambda(s)\coloneqq \lambda \eta (s/\sqrt \lambda)\,,\quad |s|\le \sqrt{\lambda}\,.
\end{align}
Thus, $\eta_\lambda'(s) \coloneqq \left( \eta_\lambda\right)'(s)=\sqrt{ \lambda(1-s^2/\lambda)}=\sqrt{\lambda-s^2}$ and $\eta_{\lambda}^{-1}(s)\coloneqq \left(\eta_\lambda\right)^{-1}(s)= \sqrt{\lambda} \eta^{-1}(s/\lambda)$.
\begin{lemma} \label{psi n asymp}
 Let $\varepsilon>0$ and $n \in \mathbb N$. Then, for any $s \in \mathbb R$ with $\lvert s \rvert \le (1-\varepsilon) \sqrt{2n+1}$ we have the asymptotic expansion as $n\to\infty$,
 \begin{align}
\sqrt{ \eta_{2n+1}'(s)}  \psi_n(s)= \sqrt{2 / \pi} \cos(\eta_{2n+1}(s) - n \pi /2 ) + \mathcal O(1/(1+n \varepsilon^{3/2}))\,.
 \end{align}
  
\end{lemma}

\begin{proof}
Let $n=2 \ell$ or $n=2\ell+1$. According to \cite[\href{http://dlmf.nist.gov/18.7\#E19}{(18.7.19)}, \href{http://dlmf.nist.gov/18.7\#E20}{(18.7.20)}]{NIST:DLMF}, the Hermite polynomials can be expressed in terms of Laguerre polynomials.  Thus, we have
\begin{align}
\psi_{2 \ell}(s) &=  \frac   { (-1)^\ell 2^{2\ell} \ell! } {\left(\sqrt \pi 2^{2\ell} (2\ell) ! \right) ^{1/2} }\Ln_{\ell}^{(-1/2)}(s^2) \exp(-s^2/2) \\
&= \frac   { (-1)^\ell 2^{2\ell} \ell! } {\left(\sqrt \pi 2^{2\ell} (2\ell) ! \right) ^{1/2} } \frac{ F_{\ell+1}^{(-1/2)}\left( s^2/(4\ell+1) \right) } { 2^{-1/2} \sqrt {4\ell+1} \lvert 1- s^2/(4\ell+1) \rvert^{1/4} }  \\
&= \frac   { (-1)^\ell 2^{\ell} \ell!  } {\left(\sqrt \pi  (2\ell) ! \right) ^{1/2} \sqrt {2\ell+1/2} } \frac{ F_{\ell+1}^{(-1/2)}\left( s^2/(4\ell+1) \right) } {  \lvert 1- s^2/(4\ell+1) \rvert^{1/4} } \,, \\
%Second equation starts here.
\psi_{2 \ell+1}(s) &=  \frac   { (-1)^\ell 2^{2\ell+1} \ell! } {\left(\sqrt \pi 2^{2\ell+1} (2\ell+1) ! \right) ^{1/2} }s \Ln_{\ell}^{(+1/2)}(s^2) \exp(-s^2/2) \\
&= \frac   { (-1)^\ell 2^{2\ell+1} \ell! } {\left(\sqrt \pi 2^{2\ell+1}   (2\ell+1) ! \right) ^{1/2} } \frac{s F_{\ell+1}^{(+1/2)}\left( s^2/(4\ell+3) \right) } { 2^{+1/2} \sqrt {4\ell+3} \frac s {\sqrt {4\ell+3}} \lvert 1- s^2/(4\ell+3) \rvert^{1/4} }  \\
&= \frac   { (-1)^\ell 2^{\ell} \ell!  } {\left( \sqrt \pi   (2\ell) ! \right) ^{1/2} \sqrt{2\ell+1}} \frac{ F_{\ell+1}^{(+1/2)}\left( s^2/(4\ell+3) \right) } { \lvert 1- s^2/(4\ell+3) \rvert^{1/4} } \,.
\end{align}
In both cases, the first factor depends only on $\ell$ and points to a Stirling approximation. The version of the Stirling approximation we are using is $m!= \sqrt {2 \pi m} \left( m/e\right)^m \left(1+ \mathcal O\left(1/m\right) \right)$. Thus, for $a \in \{1/2, 1\}$, we observe
\begin{align}
\frac   { (-1)^\ell 2^{\ell} \ell!  } {\left( \sqrt \pi   (2\ell) ! \right) ^{1/2} \sqrt{2\ell+a}}&= \frac   { (-1)^\ell 2^{\ell} \sqrt {2 \pi \ell}  (\ell/e)^\ell  } {\left( \sqrt \pi   \sqrt{4 \pi \ell}  (2\ell/e)^{2\ell}  \right) ^{1/2} \sqrt{2\ell}} \left( 1 + \mathcal O(1/\ell)\right) \\
&= \frac   { (-1)^\ell  \sqrt {\ell}    } {\left(   \sqrt{ \ell}  \right) ^{1/2} \sqrt{2\ell}} \left( 1 + \mathcal O(1/\ell)\right) \\
&= \frac   { (-1)^\ell   } {\left( 2n+1 \right)^{\frac 1 4}} \left( 1 + \mathcal O(1/n)\right).
\end{align}

%{\color{red} Plausibility check: Together with the asymptotics for $F_\ell^{(\alpha)}$, this tells us that the Hermite functions are of order $\mathcal O(1/\ell^{1/4})$ on an interval of size  $\mathcal O(\sqrt \ell)$ and decay fast outside that interval. This fits well with $\lVert \psi_\ell \rVert_{\Lp^2(\mathbb R)}=1$.}

So far, we have shown that
\begin{align}
\left( (2n+1) \lvert 1-s^2/(2n+1) \rvert \right)^{1/4} \psi_n(s) =& \begin{cases} 
	(-1)^{\ell} F_{\ell+1} ^{(-1/2)} (s^2 / (2n+1) ) (1+ \mathcal O(1/n))\,,  &\text{ if } n=2\ell\,, \\
	(-1)^{\ell} F_{\ell+1} ^{(+1/2)} (s^2 / (2n+1) )   (1+ \mathcal O(1/n))\,,  &\text{ if } n=2\ell +1 \,.
	\end{cases}
\end{align}
We will use the asymptotics for $F_K^{(\alpha)}$, which we developed in \autoref{?? for x<1/2} and \autoref{Laguerre to Airy}. The parameter $\nu$ appearing in the asymptotics of $F_K^{(\alpha)}$, is given by $\nu=4K+2(\alpha-1)$, where $K=\ell+1$. If $n=2\ell$, this simplifies to $4(\ell+1)+2(-1/2-1)=4 \ell +1 = 2n+1$. In the other case $n=2\ell +1$, it is $4(\ell+1)+2(1/2-1)=4\ell+3=2n+1$. Thus, independent of the parity of $n$, the parameter $\nu$ is given by $2n+1$. Thus, 
\begin{align}
\nu \xi(s^2/(2n+1))=(2n+1) \eta(s/\sqrt{2n+1})=\eta_{2n+1}(s)\,.
\end{align}
We observe that $\alpha=-1/2 \,(-1)^n$. We can now insert this into the asymptotics granted by \autoref{?? for x<1/2} and \autoref{Laguerre to Airy} for the case $\lvert s\rvert \le \sqrt{2n+1}$ and get
\begin{align}
&\sqrt{ \eta_{2n+1}'(s)} \sqrt{\pi/2}\, \psi_n(s) \\
&=\left( (2n+1) \lvert 1-s^2/(2n+1) \rvert \right)^{1/4} \sqrt{\pi /2} \, \psi_n(s)\\
&= (-1)^\ell  \cos ( \nu \xi(s^2/(2n+1) - \alpha \pi/2 - \pi/4 ) +\mathcal O\left(\frac 1{1+ (\ell+1) (1-s^2/(2n+1))^{3/2}} \right) + \mathcal O\left( \frac 1  n\right) \\
&= (-1)^\ell \cos(\eta_{2n+1}(s) - \pi/4 (1-(-1)^n)) +\mathcal O\left(\frac 1{1+ n(1-s^2/(2n+1))^{3/2}} \right) \\
&= \cos(\eta_{2n+1}(s) - \pi/4 (1-(-1)^n)-\ell \pi ) +\mathcal O\left(\frac 1{1+ n(1-(1-\varepsilon)^2)^{3/2}} \right)\\
&= \cos(\eta_{2n+1}(s) - n\pi/2 ) +\mathcal O\left(\frac 1{1+ n\varepsilon^{3/2}} \right).
\end{align}
The last step relies on the identity $n=2\ell + (1-(-1)^n)/2$. 
\end{proof}

The next lemma deals with the integral kernel of $\mathcal K_K$ under scaling, that is, a change of coordinates by $\eta_{2K-1}$.

\begin{lemma} \label{KK unitary equivalence}
Let $K \in \mathbb N$ and let $I ,J \subset (-\sqrt{2K-1},\sqrt{2K-1})\eqqcolon \Omega_1$ be intervals. We define for $s, t \in (-(2K-1)\pi/4, (2K-1)\pi/4)\eqqcolon \Omega_2$, 
\begin{align} \label{hat K}
\hat {\mathcal K}_K(s,t) \coloneqq \sqrt{ \left( \eta_{2K-1}^{-1} \right)'(s) \left( \eta_{2K-1}^{-1} \right)'(t)}\, \mathcal K_K\big(\eta_{2K-1}^{-1}(s),\eta_{2K-1}^{-1}(t)\big) \,.
\end{align}
Then, the operators $1_I \mathcal K_K 1_J$ and $1_{\eta_{2K-1}(I)} \hat{\mathcal K}_K 1_{\eta_{2K-1}(J)}$ are unitarily equivalent.
\end{lemma}
\begin{proof}
As $I,J \subset \Omega_1$, it suffices to consider $\mathcal K_K$ as an integral operator on $\Lp^2( \Omega_1)$ and as $\eta_{2K-1}(\Omega_1)=\Omega_2$, it suffices to consider $\hat {\mathcal K}_K$ as an integral operator on $\Lp^2(\Omega_2)$. We define the shorthand $\theta \coloneqq \eta_{2K-1}$. All we need to know about $\theta$ is that $\theta$ is a $\mathsf{C}^1$-smooth bijection from $\Omega_1$ to $\Omega_2$. We can now define the unitary operator
\begin{align}
A_\theta^{-1} \colon \Lp^2( \Omega_2) &\to \Lp^2(\Omega_1)\,,  &f \mapsto& (t \mapsto f(\theta(t)) \sqrt{\theta'(t)})\,, \\
A_\theta \colon \Lp^2( \Omega_1) &\to \Lp^2(\Omega_2)\,,  &g \mapsto& (t \mapsto g(\theta^{-1}(t)) \sqrt{\left(\theta^{-1}\right)'(t)})\,.
\end{align}

We see that $A_\theta 1_I A_\theta^{-1}= 1_{\theta(I)}$ and  $A_\theta 1_J A_\theta^{-1}= 1_{\theta(J)}$. Let $f \in \Lp^2(\Omega_2)$ and $a \in \Omega_2$. Then, we observe that
\begin{align}
\left(A_\theta \mathcal K_K A_\theta ^{-1} f \right)(a)&= \left( A_\theta \int_{\Omega_1} \mathcal K_K(\blank ,t) f(\theta(t))\sqrt{\theta'(t)} \mathrm d t \right) (a) \\
&= \int_{\Omega_1}\sqrt{\left(\theta^{-1}\right)'(a)}\, \mathcal K_K \left(\theta^{-1}(a), \theta^{-1}(\theta(t))\right) f(\theta(t)) 1/\sqrt{\left(\theta^{-1}\right)' (\theta(t))}\, \mathrm d t \\
&= \int_{\Omega_2} \sqrt{\left(\theta^{-1}\right)'(a)} \left(\theta^{-1}\right)' (b) \,\mathcal K_K \left(\theta^{-1}(a), \theta^{-1}(b)\right) f(b) 1/\sqrt{\left(\theta^{-1}\right)' (b)} \,\mathrm d b \\
&= \int_{\Omega_2} \hat {\mathcal K}_K (a,b) f(b)\, \mathrm d b\\
&= (\hat{\mathcal K}_K f)(a)\,.
\end{align}
Thus, $1_I \mathcal K_K 1_J  = A_\theta^{-1} 1_{\theta(I)} \hat{\mathcal K}_K  1_{\theta(J)} A_\theta $ and $1_{\theta(I)} \hat{\mathcal K}_K 1_{\theta(J)}$ are unitarily equivalent, which was the claim.
\end{proof}

The next result is the mentioned sine-kernel asymptotics of $\mathcal K_K$ in the bulk, at least in the most relevant portion of it.

\begin{theorem} \label{hat K kernel est}
Let $0<\varepsilon<1/2, K \in \mathbb N$ with $\varepsilon^6 K>C$. (This $C$ is fixed, but currently unknown.) For any $s,t \in \mathbb R$ with $(2K-1)\eta(-1/2) < s < t < (2K-1)\eta(1-\varepsilon)$ and $1/2  \le \lvert s-t \rvert \le K\varepsilon^6$, we have  with $\hat {\mathcal K}_K(s,t)$ defined in \eqref{hat K},
\begin{align} \label{sine kernel is here eq}
\hat {\mathcal K}_K(s,t) = \frac{ \sin(s-t)}{\pi (s-t)} + \mathcal O\left( \frac{\varepsilon} {\lvert s-t \rvert }\right).
\end{align}
\end{theorem}

\begin{remark}
That this asymptotics also holds for $\lvert s-t \rvert \le 1/2$ (local scale) is well-known (for instance in random matrix theory), but our approach would need estimates on $\psi_\ell '$, which we don't look into.  
\end{remark}

\begin{proof}
We define
\begin{align} \label{def D(s,t)}
D(s,t) \coloneqq &\sqrt{ \left( \eta_{2K-1}^{-1} \right)'(s) \left( \eta_{2K-1}^{-1} \right)'(t)} \\
&=\sqrt{\eta_{2K-1}'(\eta_{2K-1}^{-1}(s)) \eta_{2K-1}'(\eta_{2K-1}^{-1}(t)) }^{-1}\\
&= \sqrt{(2K-1) \eta'(\eta^{-1}(s/(2K-1)))  \eta'(\eta^{-1}(t/(2K-1)))}^{-1} \,.
\end{align}
We recall that by the Christoffel--Darboux formula 
\begin{align}
\mathcal K_K(x,y)=  \sqrt{K/2} \,\frac{ \psi_K(x)\psi_{K-1}(y)- \psi_{K-1}(x)\psi_K(y)}{x-y} \,.
\end{align}
Therefore, we write
\begin{align} \label{K hat split}
\hat{\mathcal K}_K(s,t)=&\overbrace{  \frac{ D(s,t)} { \eta_{2K-1}^{-1}(s) - \eta_{2K-1}^{-1}(t)} }^{\eqqcolon T_1(s,t)}  \sqrt{K/2} \, D(s,t) \\
&\times \Bigg[\underbrace{\frac{ \psi_K\big(\eta_{2K-1}^{-1}(s)) \psi_{K-1}(\eta_{2K-1}^{-1}(t)\big)}{D(s,t)}}_{\eqqcolon T_2(s,t)} -\underbrace{\frac{\psi_{K-1}\big(\eta_{2K-1}^{-1}(s)\big) \psi_{K}\big(\eta_{2K-1}^{-1}(t)\big)}{D(s,t)} }_{ \eqqcolon T_3(s,t)}\Bigg]\,.
\end{align}

We take a closer look at the condition $t \le(2K-1)\eta(1-\varepsilon)$. Together with $s \ge (2K-1)\eta(-1/2)$ and $s<t$, this tells us that $s/(2K-1)$ and $t/(2K-1)$ are in $\eta((-1,1))=(\eta(-1),\eta(1))$. Thus, the expressions $\eta_{2K-1}^{-1}(s) , \eta_{2K-1}^{-1}(t)$ are well-defined. Furthermore, as $\eta$ is increasing on $(-1,1)$, its inverse is increasing as well and thus, we can conclude
\begin{align} \label{bounds on eta inv}
-1/2 \le \eta^{-1}(s/(2K-1))< \eta^{-1}(t/(2K-1))\le  1 -  \varepsilon\,.
\end{align}
The assumption $\lvert s-t \rvert < K \varepsilon^6$ now leads to
\begin{align} \label{eta inv Lip}
\big\lvert \eta^{-1}(s/(2K-1)) - \eta^{-1}(t/(2K-1)) \big\rvert &\le \frac{\lvert s-t \rvert}{K \inf_{a \in (-1/2, 1-\varepsilon)} \eta'(a)} \\
&= \frac{\lvert s-t \rvert}{K \inf_{a \in (-1/2, 1-\varepsilon)} \sqrt{1-a^2}}< \frac {\lvert s-t \rvert }{K \sqrt \varepsilon} < \varepsilon^5\,.
\end{align}

Next, we need a technical result. Let $u,v \in (-1/2,1-\varepsilon)$ with $\lvert u- v \rvert < \varepsilon$. Then, we claim that 
\begin{align} \label{eta' log Lip}
\eta'(u)/\eta'(v) =1 + \mathcal O( \lvert u- v \rvert /\varepsilon)\,.
\end{align}
To this end, we consider the function $ a \mapsto 2\ln(\eta'(a))$ and take its derivative at some $a \in (-1/2,1-\varepsilon)$. Thus, we get
\begin{align}
\lvert 2\ln(\eta'(a))' \rvert = \lvert \ln(1-a^2)' \rvert = 2\lvert a \rvert/(1-a^2)< 1/\varepsilon\,.
\end{align}
Thus, $1/ \varepsilon$ is a Lipschitz constant for $a \mapsto \ln(\eta'(a))$. Since $\lvert u-v \rvert < \varepsilon$, our claim
\begin{align}
\eta'(u)/\eta'(v)=\exp(\mathcal O(\lvert u-v \rvert/\varepsilon)) = 1+ \mathcal O(\lvert u-v \rvert/\varepsilon)
\end{align}
follows.

We begin by showing that $T_1(s,t) \approx 1/(s-t)$. 

We recall that $\eta_{2K-1}^{-1}$ is smooth and strictly increasing on the interval $(s,t)$. Let $\hat s \in (s,t)$ be (uniquely) determined by the mean-value theorem for $\eta_{2K-1}^{-1}$ on $(s,t)$, as seen next. We consider the expression
\begin{align}
\frac{ (s-t) D(s,t) }{\eta_{2K-1}^{-1}(s) - \eta_{2K-1}^{-1}(t)}=& \frac{\sqrt{ \left( \eta_{2K-1}^{-1} \right)'(s) \left( \eta_{2K-1}^{-1} \right)'(t)}  }{\left( \eta_{2K-1}^{-1} \right)'(\hat s)} \\
&= \frac{ \sqrt{2K-1}\eta'(\eta^{-1}(\hat s/(2K-1)))} { \sqrt{(2K-1) \eta'(\eta^{-1}(s/(2K-1)))  \eta'(\eta^{-1}(t/(2K-1)))}} \\
&= 1 +\mathcal O \left( \varepsilon^5/\varepsilon \right) = 1 +\mathcal O \left( \varepsilon^4 \right) \,.
\end{align}
The final step relies on \eqref{eta' log Lip} and \eqref{eta inv Lip}. Thus, we have just shown that
\begin{align} \label{T1 final}
T_1(s,t)=\frac 1 {s-t} \left( 1 + \mathcal O(\varepsilon^4)\right).
\end{align}

Next, we consider $T_2$ and expand
\begin{align}
T_2(s,t)&=  \frac{ \psi_K\big(\eta_{2K-1}^{-1}(s)) \psi_{K-1}(\eta_{2K-1}^{-1}(t)\big)}{D(s,t)} \\
&= \frac{ \sqrt{\eta_{2K-1}'(\eta_{2K-1}^{-1}(s))} }{ \sqrt{\eta_{2K+1}'(\eta_{2K-1}^{-1}(s))} } \left(  \sqrt{\eta_{2K+1}'(\eta_{2K-1}^{-1}(s))}  \psi_K\big(\eta_{2K-1}^{-1}(s)) \right) \\ &\times \left(  \sqrt{\eta_{2K-1}'(\eta_{2K-1}^{-1}(t))}  \psi_{K-1}\big(\eta_{2K-1}^{-1}(t)) \right) .
\end{align}
The final two terms are already in the form of \autoref{psi n asymp}. For the first factor, we just observe
\begin{align}
\frac{ \eta_{2K-1}'(\eta_{2K-1}^{-1}(s)) }{ \eta_{2K+1}'(\eta_{2K-1}^{-1}(s)) } =& \frac{ \sqrt{2K-1} \eta'(\eta^{-1}(s/(2K-1))}{ \sqrt{2K+1} \eta'\left( \sqrt{(2K-1)/(2K+1)} \eta^{-1} (s/(2K-1)) \right) }\\
=& \left( 1+ \mathcal O(1/K) \right) \left(1 + \mathcal O(1/(K\varepsilon))\right) = 1 + \mathcal O(\varepsilon^5)\, . 
\end{align}
We used $\sqrt{(2K-1)/(2K+1)} =1+\mathcal O(1/K)$, \eqref{bounds on eta inv} and \eqref{eta' log Lip}.  Since $|\eta_{2K-1}^{-1}(s)|\le (1-\varepsilon) \sqrt{2K-1}$ (and the same with $t$ replacing $s$) by \eqref{bounds on eta inv} and $1/(K \varepsilon^{3/2}) < \varepsilon^4$ we can apply \autoref{psi n asymp} to see that
\begin{align}
T_2(s,t)&= \left( 1 + \mathcal O(\varepsilon^5)\right)\frac 2 \pi \left( \cos(\eta_{2K+1}(\eta_{2K-1}^{-1}(s)) - K \pi /2 ) + \varepsilon^4 \right) \\
&\times \left( \cos(\eta_{2K-1} (\eta_{2K-1}^{-1}(t))- (K-1) \pi/2 ) + \mathcal O(\varepsilon^4)\right)\\
&=- \frac 2 \pi \cos(\eta_{2K+1}(\eta_{2K-1}^{-1}(s)) - K \pi /2 )  \sin ( t- K \pi/2) + \mathcal O(\varepsilon^4)\,.
\end{align}

Our next goal is to understand $\eta_{2K+1}\circ \eta_{2K-1}(s)$. By definition, 
\begin{equation} \eta_{2K-1}(s) = (2K-1)\eta(s/\sqrt{2K-1}) 
\end{equation}
so that $\eta_{2K-1}^{-1}(s) = \sqrt{2K-1}\eta^{-1}(s/(2K-1))$. We also recall $\eta'(t) = \sqrt{1-t^2}$. Using the mean-value theorem and \eqref{eta' log Lip}, we can conclude
\begin{align} 
&\eta_{2K+1}\big(\eta_{2K-1}^{-1}(s)\big) \\
&= (2K+1)\eta\Big(\frac{\eta_{2K-1}^{-1}(s)}{\sqrt{2K+1}}\Big)
\\
&= (2K+1)\eta\Big( \eta^{-1}\left(s/(2K-1)\right) \Big(1-\frac{1}{2K+1}+\mathcal O\left( \frac 1 {K^2} \right)\Big)\Big)
\\
&= (2K+1) \Bigg[\frac s {2K-1}  - \eta'\left(\eta^{-1}\left( \frac{s}{2K-1} \right) \right) \frac{1+\mathcal O(1/(K \varepsilon))}{2K+1} \eta^{-1}\left( \frac{s}{2K-1} \right)+\mathcal O\left( \frac 1 {K^{2}} \right)\Bigg]
\\
&=(2K+1)\Big[\frac{s}{2K-1} - \sqrt{1-(\eta^{-1}(s/(2K-1)))^2} \,\frac{1+\mathcal O(\varepsilon^5)}{2K+1} \,\eta^{-1}(s/(2K-1))+\mathcal O(\varepsilon^5/K)\Big]
\\
&= s + 2s/(2K-1) -\eta^{-1}(s/(2K-1)) \sqrt{1-(\eta^{-1}(s/(2K-1)))^2} +\mathcal O(\varepsilon^5)
\\
&- \sin^{-1}(\eta^{-1}(s/(2K-1))) + \sin^{-1}(\eta^{-1}(s/(2K-1))) 
\\
&=s + 2s/(2K-1) - 2\eta(\eta^{-1}(s/(2K-1))) + \sin^{-1}(\eta^{-1}(s/(2K-1)))  +\mathcal O(\varepsilon^5)
\\
&=s + \sin^{-1}(\eta^{-1}(s/(2K-1))) +\mathcal O(\varepsilon^5)\,.
\end{align}
Thus, as the cosine is Lipschitz-continuous, we can conclude
\begin{align}
T_2(s,t)= -\frac 2 \pi \cos\Big(s+\sin^{-1}(\eta^{-1}(s/(2K-1)))   - K \pi /2 \Big)  \sin ( t- K \pi/2) + \mathcal O(\varepsilon^4)\,.
\end{align}
The identity $-2\cos(a)\sin(b)+\sin(a+b)=\sin(a-b)$ yields 
\begin{align}
&\pi T_2(s,t) + \sin\left(s+t +\sin^{-1}(\eta^{-1}(s/(2K-1)))-K \pi \right)\\
&=  \sin\left(s-t + \sin^{-1}(\eta^{-1}(s/(2K-1)))\right)  + \mathcal O(\varepsilon^4) \\
&= \sin(s-t) \cos\circ \sin^{-1} ( \eta^{-1}(s/(2K-1))) + \cos(s-t) \sin \circ \sin^{-1} (\eta^{-1}(s/(2K-1))) + \mathcal O(\varepsilon^4) \\
&= \sin(s-t) \eta'(\eta^{-1}(s/(2K-1))) + \cos(s-t) \eta^{-1}(s/(2K-1)) + \mathcal O(\varepsilon^4) \,.
\end{align}

This allows us to conclude
\begin{align} 
\pi T_2(s,t)- \pi T_3(s,t)& = \pi T_2(s,t)- \pi T_2(t,s) \\
&= \sin(s-t) \left( \eta'(\eta^{-1}(s/(2K-1))) + \eta'(\eta^{-1}(t/(2K-1))) \right)\\
& + \cos(s-t) \left( \eta^{-1}(s/(2K-1)) - \eta^{-1}(t/(2K-1)) \right) \\
&- \sin\left(s+t +\sin^{-1}(\eta^{-1}(s/(2K-1)))-K \pi \right) \\
&+\sin\left(s+t +\sin^{-1}(\eta^{-1}(t/(2K-1)))-K \pi \right)+ \mathcal O(\varepsilon^4) \\
&= \sin(s-t) \left( 2 \sqrt{\eta'(\eta^{-1}(s/(2K-1)))  \eta'(\eta^{-1}(t/(2K-1)))} \right) \\
&+ \mathcal O\left( \left \lvert \eta'(\eta^{-1}(s/(2K-1)))-\eta'(\eta^{-1}(t/(2K-1))) \right \rvert \right) \\
&+\mathcal O\left( \left \lvert \eta^{-1}(s/(2K-1))-\eta^{-1}(t/(2K-1)) \right \rvert \right) \\
&+ \mathcal O\left(\left \lvert \sin^{-1}(\eta^{-1}(s/(2K-1))) - \sin^{-1}(\eta^{-1}(t/(2K-1))) \right \rvert \right)+\mathcal O(\varepsilon^4)\,.
\end{align}
We used the rough estimates $a+b=2 \sqrt{ab} + \mathcal O(\lvert a-b\rvert )$ and $\lvert \sin(a)-\sin(b) \rvert \le \lvert a-b \rvert$. As $\sin^{-1}$ and $a \mapsto \eta'(a)=\sqrt{1-a^2}$ are $1/2$-Hölder continuous, \eqref{eta inv Lip} leads to
\begin{align}
\pi( T_2(s,t)-  T_3(s,t))&=\sin(s-t) \left( 2 \sqrt{\eta'(\eta^{-1}(s/(2K-1)))  \eta'(\eta^{-1}(t/(2K-1)))} \right) + \mathcal O(\varepsilon^2)\\
&= \frac{ 2\sin(s-t)}{\sqrt{2K-1} D(s,t)} + \mathcal O(\varepsilon^2) \\
&= \frac{ \sin(s-t)}{\sqrt{K/2} D(s,t)} + \mathcal O(\varepsilon^2) \,. \label{T2-T3 final}
\end{align}
Combining \eqref{K hat split}, \eqref{T1 final}, and \eqref{T2-T3 final}, we can conclude
\begin{align}
\hat {\mathcal K}_K(s,t)=& \frac 1 {s-t} \left( 1+ \mathcal O(\varepsilon^4) \right) \sqrt{K/2} D(s,t) \left(\frac{ \sin(s-t)}{\pi\sqrt{K/2} D(s,t)} + \mathcal O(\varepsilon^2) \right) \\
=& \frac{\sin(s-t)}{\pi(s-t)}+ \mathcal O\left( \frac{ \varepsilon^4+\varepsilon^2 \sqrt K D(s,t)} {\lvert s-t \rvert} \right).
\end{align}
We are left to show that $\varepsilon \sqrt K D(s,t)$ is bounded independent of $s,t,\varepsilon$ and $K$. Using \eqref{def D(s,t)}, we observe
\begin{align}
\varepsilon \sqrt K D(s,t) &= \frac {\varepsilon \sqrt K } {\sqrt{(2K-1) \eta'(\eta^{-1}(s/(2K-1)))  \eta'(\eta^{-1}(t/(2K-1)))}} \\
&\le \frac{ \varepsilon}{ \inf_{a,b \in (-1/2,1-\varepsilon)} \sqrt{\eta'(a) \eta'(b)}} \\
&\le  \frac{\varepsilon}{\inf_{\varepsilon< h <1} \sqrt {1-(1-h)^2}} \le \frac {\varepsilon}{\sqrt{\varepsilon}}<1\,.
\end{align}
This concludes the proof.
\end{proof}

\subsection{Evaluation of the leading asymptotic coefficient $\mathsf{M}_{<K}(f)$}

The above asymptotic results on the kernel $\mathcal K_K$ will now be applied to the deal with the trace of powers of the restricted operators $1_{<p} \mathcal K_K 1_{>p}$. Given \autoref{MK=Jm} and  \autoref{MK symmetry}, we are left to study
\begin{align}
\pi \,\mathsf{M}_{< K} (t \mapsto [t(1-t)]^m ) =& \int_0^\infty \mathrm d p  \,\operatorname{tr} \left \lvert 1_{<p} \mathcal K_K 1_{>p} \right \rvert^{2m}\,,
\end{align}
where $\lvert A \rvert^{2m}= (A^*A)^m$. \autoref{KK unitary equivalence} and \autoref{hat K kernel est} provide us with a good understanding of the operator $\mathcal K_K$ on certain pairs of non-intersecting intervals. We will now define a specific pair of intervals for each (sufficiently large) $K \in \mathbb N$ and $p \in \mathbb R^+$, which will lead to the main contribution. 

\begin{definition} \label{def: intervals}
Let $K \in \mathbb N$ with $K\ge 100$ and let $p \in \mathbb R^+$. With $\varepsilon \coloneqq 1/\ln(K)<1/4$ we define the intervals $\hat I_p$ and $\hat J_p$ as
\begin{align}
\hat I_p &\coloneqq \begin{cases} \left( \eta_{2K-1}(p)-K \varepsilon^6/2 , \eta_{2K-1}(p) -1/2\right) \quad & \text{ if } p\le  (1- 2 \varepsilon) \sqrt{2K-1}\,,\\ \emptyset &\text{ else,} \end{cases} \\
\hat J_p &\coloneqq \begin{cases} \left(\eta_{2K-1}( p),\eta_{2K-1}(p)+K \varepsilon^6/2 \right) \quad & \text{ if } p\le  (1- 2 \varepsilon) \sqrt{2K-1}\,, \\ \emptyset &\text{ else}. \end{cases}
\end{align}
Furthermore, let $I_p \coloneqq \eta_{2K-1}^{-1}(\hat I_p)$ and $I_p \coloneqq \eta_{2K-1}^{-1}(\hat J_p)$.
\end{definition}

\begin{remark} That the above intervals $I_p$ and $J_p$ are well-defined will become clear in the following proof. 
\end{remark}

With these families of intervals, we are finally ready for the
\begin{proof}[Proof of \autoref{main term is MI}]
We begin by taking a closer look at the intervals $\hat I_p,\hat J_p$. Using that $\hat I_p,\hat J_p$ are empty unless $0 < p \le (1-2 \varepsilon) \sqrt{2K-1}$, we observe
\begin{align}
\hat I_p \cup \hat J_p &\subset \left(\eta_{2K-1}(p)- K \varepsilon^6/2, \eta_{2K-1} (p)+ K \varepsilon^6/2 \right) \\
&\subset \left( - K \varepsilon^6/2, (2K-1) \eta(1-2 \varepsilon)+ K \varepsilon^6/2 \right) \\
&\subset \left(-(2K-1) \varepsilon^6/2, (2K-1)(\eta(1-2 \varepsilon)+\varepsilon^6/2)\right) \\
&\subset \left((2K-1)\eta(-1/2), (2K-1) \eta(1-\varepsilon) \right).
\end{align}
The last step relies on the estimates $\varepsilon^6/2 <(1/2)^7 < (\pi/4) /2 \le \eta(1/2)$ (see \eqref{eta lower linear bound}) and $\eta(1-2\varepsilon)+\varepsilon^6/2< \eta(1-\varepsilon)$, which can be seen as follows:
\begin{align}
\eta(1-\varepsilon)& > \eta(1-2\varepsilon)+ \varepsilon \inf_{h \in (\varepsilon,2\varepsilon)}\eta'(1-h) \\
&=\eta(1-2\varepsilon)+ \varepsilon \inf_{h \in (\varepsilon,2\varepsilon)}\sqrt{1-(1-h)^2} \\
&> \eta(1-2\varepsilon)+ \varepsilon \sqrt{\varepsilon} >\eta(1-2\varepsilon)+ \varepsilon^6/2\,.
\end{align}
Thus, $I_p=\eta_{2K-1}^{-1}(\hat I_p),J_p= \eta_{2K-1}^{-1}(\hat J_p)$ are well defined, see \eqref{def eta_lambda} for the definition of $\eta_{2K-1}$. For any pair $s \in I_p,t \in J_p$ we have that $(2K-1)\eta(-1/2) < s < t < (2K-1)\eta(1-\varepsilon)$ and $1/2 \le \lvert s-t \rvert \le K\varepsilon^6$, which are the assumptions of \autoref{hat K kernel est}. Furthermore, $I_p,J_p$ satisfy the conditions of \autoref{KK unitary equivalence}. Thus, \autoref{KK unitary equivalence} yields
\begin{align}
\frac 1 {\pi \sqrt K}   \int_0^\infty \mathrm d p  \operatorname{tr} \left \lvert  1_{I_p} \mathcal K_K 1_{J_p} \right \rvert^{2m} &= \frac 1 {\pi \sqrt K} \int_0^{\sqrt{2K-1}(1-2 \varepsilon)} \mathrm d p  \,\operatorname{tr} \left \lvert  1_{\hat I_p}\hat{ \mathcal K}_K 1_{\hat J_p} \right \rvert^{2m}\\
&= \frac 1 {\pi \sqrt K} \int_0^{\sqrt{2K-1}(1-2 \varepsilon)} \mathrm d p  \, \left \lVert  1_{\hat I_p}\hat{ \mathcal K}_K 1_{\hat J_p} \right \rVert_{2m}^{2m}\,.
\end{align}
Let us define the sine-kernel $T$ for any $s,t \in \mathbb R$ by
\begin{align}
T(s,t)=\frac{\sin(s-t)}{\pi(s-t)} \,.
\end{align}
Let $\lambda \coloneqq K \varepsilon^6/2$. Let $\transop_p$ be the (unitary) shift operator that sends $f \in \Lp^2(\R)$ to $t \mapsto f(t-\eta_{2K-1}(p))$. Thus, $\transop_p^{-1} 1_{\hat I_p} \transop_p =1_{(-\lambda,-1/2)}=1_{\hat I_0}$ and $\transop_p^{-1} 1_{\hat J_p} \transop_p =1_{(0,\lambda)}=1_{\hat J_0}$. Furthermore, $\transop_p$ commutes with $T$.

Let us define the integral operator $T_p: \Lp^2(\hat J_0)\to\Lp^2(\hat I_0)$ by setting for any $s \in I_0, t \in J_0$, 
\begin{align}
T_p(s,t)\coloneqq \hat{\mathcal K}_K(s+\eta_{2K-1}(p),t+\eta_{2K-1}(p))-T(s,t)\,.
\end{align}
This lets us write
\begin{align}
\transop_p^{-1} 1_{\hat I_p} \hat {\mathcal K}_K 1_{\hat J_p} \transop_p = 1_{\hat I_0} (T+T_p) 1_{\hat J_0}\,.
\end{align}
\autoref{hat K kernel est} tells us that 
\begin{align}
T_p(s,t)= \mathcal O \left( \frac \varepsilon {\lvert s-t \rvert} \right).
\end{align}
This leads to
\begin{align}
\left \lVert 1_{\hat I_0} T_p 1_{\hat J_0} \right \rVert_2^2 &\le C \varepsilon^2 \int_{-\lambda}^{-1/2} \mathrm d s\int_{0}^{\lambda} \mathrm d t   \,\frac 1 {\lvert s-t \rvert^2} \\
&\le C\varepsilon^2 \int_{-K/2} ^{-1/2} \mathrm ds \int_{0}^{\infty} \mathrm d t  \, \frac 1 {(s-t)^2} = C \varepsilon^2 \ln(K) =C/\ln(K)\,.
\end{align}
Similarly, from $T(s,t) \le 1/\lvert s-t \rvert$, we can also derive
\begin{align}
\left \lVert 1_{\hat I_0} T 1_{\hat J_0} \right \rVert_2^2 \le C \ln(K)\,.
\end{align}
We also get
\begin{align}
\left \lVert 1_{\hat I_0} T 1_{( \lambda,\infty)} \right \rVert_2^2 \le  \int_{-\lambda}^{-1/2} \mathrm ds \int_\lambda^\infty \mathrm d t \,\frac 1 {\lvert s-t \rvert^2 }<   \int_{-\lambda}^{0} \mathrm ds \int_\lambda^\infty \mathrm d t\, \frac 1 {\lvert s-t \rvert^2 } = \ln(2)\,.
\end{align}
Let us now employ \autoref{tele sum lem1} and observe that
\begin{align}
& \left \lVert \left( 1_{\hat I_0} (T+T_p) 1_{\hat J_0} (T+T_p) 1_{\hat I_0} \right)^m - \left( 1_{\hat I_0} T 1_{\hat J_0} T 1_{\hat I_0} \right)^m \right \rVert_1 \\
&\le  m\left \lVert 1_{\hat I_0} (T+T_p) 1_{\hat J_0}- 1_{\hat I_0} T 1_{\hat J_0}  \right \rVert_2 \left( \left \lVert 1_{\hat I_0} (T+T_p) 1_{\hat J_0} \right \rVert_2+ \left \lVert 1_{\hat I_0} T 1_{\hat J_0}\right \rVert_2 \right) \\
&\le   m \left(C/\sqrt{\ln(K)} \right) \left( C \sqrt{\ln(K)} \right) = Cm\,.
\end{align}
Using the intermediate estimate in \autoref{tele sum lem1}, we can conclude
\begin{align}
\left \lVert  \left( 1_{\hat I_0} T 1_{\hat J_0} T 1_{\hat I_0} \right)^m - \left( 1_{\hat I_0} T 1_{(0,\infty)} T 1_{\hat I_0} \right)^m \right \rVert_1 &\le m \left \lVert 1_{\hat I_0} T 1_{(\lambda, \infty)} T 1_{\hat I_0} \right\rVert_1\\
&= m\left \lVert 1_{\hat I_0} T 1_{(\lambda, \infty)} \right \rVert_2^2 \le Cm\,.
\end{align}
Thus, we have just seen that
\begin{align}
\operatorname{tr} \left( 1_{\hat I_0} (T+T_p) 1_{\hat J_0} (T+T_p) 1_{\hat I_0} \right)^m = \operatorname{tr}  \left( 1_{\hat I_0} T 1_{(0,\infty)} T 1_{\hat I_0} \right)^m  +\mathcal O(1) \,,
\end{align}
and we are left to calculate
\begin{align}
\operatorname{tr} \left( 1_{\hat I_0} T 1_{(0,\infty)} T 1_{\hat I_0} \right)^m\,,\quad \mbox{ with } \hat I_0 = (-\lambda,-1/2)\,,\quad \lambda = K\varepsilon^6/2\,.
\end{align}
This is a quite simple expression, as it only includes the sine-kernel and some intervals. Such expressions have been studied by Landau and Widom in \cite{LandauWidom}. We need to relate our notation to their notation before using their results to conclude our claim, which is
\begin{align} \label{LW corollary eq}
\operatorname{tr} \left( 1_{\hat I_0} T 1_{(0,\infty)} T 1_{\hat I_0} \right)^m = (1/2) \operatorname{tr} \left(1_{(0,\lambda)} T 1_{(0,\lambda)^\complement} T 1_{(0,\lambda)} \right)^m  +\mathcal O(1)\,.
\end{align}
On page 471 in \cite{LandauWidom}, one can see that in their notation, for any interval $I$, we have $P(I)=1_{I}$ and $Q(-1,1)=T$. Thus, in their notation, the left-hand side of the last equation equals
\begin{align}
\operatorname{tr}  \big[P(-\lambda,-1/2) Q(-1,1)P(0,\infty)Q(-1,1)P(-\lambda,-1/2)\big]^m\,.
\end{align}
Using the unitary transformations (rescaling and translation) \cite[(7)--(9)]{LandauWidom}, we can conclude that
\begin{align}
\operatorname{tr}&  \big[P(-\lambda,-1/2) Q(-1,1)P(0,\infty)Q(-1,1)P(-\lambda,-1/2)\big]^m
\\
&=\operatorname{tr} \big[P(1,2\lambda) Q(0,1) P(-\infty,0)Q(0,1)P(1,2\lambda)\big]^m \,.
\end{align}

In \cite{LandauWidom}, an operator is said to be $\mathcal O(1)$, if its trace norm is $\mathcal O(1)$ with respect to $\lambda \to \infty$. As we are dealing with products of projections, changing one factor by something $\mathcal O(1)$ changes the trace of the entire expression only by $\mathcal O(1)$. Thus, the second, unnumbered equation on page 475 in \cite{LandauWidom} tells us that
\begin{align}
\operatorname{tr} &\big[P(1,2\lambda) Q(0,1) P(-\infty,0)Q(0,1)P(1,2\lambda)\big]^m \\
&= \operatorname{tr} \big[P(0,2\lambda) Q(0,1) P(-\infty,0)Q(0,1)P(0,2\lambda)\big]^m+ \mathcal O(1)\,.
\end{align}
Again, using \cite[(7)--(9)]{LandauWidom}, we see that
\begin{align}
\operatorname{tr} &(P(0,2\lambda) Q(0,1) P(-\infty,0)Q(0,1)P(0,2\lambda))^m \\ 
&= \operatorname{tr} (P(0,\lambda) Q(-1,1) P(-\infty,0)Q(-1,1)P(0,\lambda))^m\\
&= \operatorname{tr} (P(0,\lambda) Q(-1,1) P(\lambda,\infty)Q(-1,1)P(0,\lambda))^m\,.
\end{align}
Thus, we can now conclude
\begin{align}
2 \operatorname{tr} & (P(1/2,\lambda) Q(-1,1)P(-\infty,0)Q(-1,1)P(1/2,\lambda))^m \\
&=  \operatorname{tr} (P(0,\lambda) Q(-1,1) P(-\infty,0)Q(-1,1)P(0,\lambda))^m  \\
&+ \operatorname{tr} (P(0,\lambda) Q(-1,1) P(\lambda,\infty)Q(-1,1)P(0,\lambda))^m + \mathcal O(1) \\
&= \operatorname{tr} (P(0,\lambda) Q(-1,1) (P(-\infty,0)+P(\lambda,\infty))Q(-1,1)P(0,\lambda))^m+ \mathcal O(1)\,,
\end{align}
where the last identity is derived from \cite[(13)]{LandauWidom} and the unitary equivalences \cite[(7)--(9)]{LandauWidom}. In conclusion, we have just proved \eqref{LW corollary eq}. 

Widom \cite{Widom1982} has also shown that
\begin{align}
\operatorname{tr} &\big[P(0,2\lambda) Q(0,1) (P(-\infty,0)+P(2\lambda,\infty))Q(0,1)P(0,2\lambda)\big]^m \\
&= \ln(\lambda/2) \frac{1}{\pi^2} \int_0^1 \mathrm dt\, \frac{(t(1-t))^m}{t(1-t)}  + \mathcal O(1) \\
& = 4\ln(K) \, \mathsf{I}(t \mapsto [t(1-t)]^m) + \mathcal O(\lvert \ln (\varepsilon^6) \rvert )\\
&= 4 \ln(K) \, \mathsf{I}(t \mapsto [t(1-t)]^m) + \mathcal O(\ln\ln(K))\, .
\end{align}
Thus, we can finally conclude that
\begin {align}
&\frac 1 {\pi \sqrt K}   \int_0^\infty \mathrm d p  \, \operatorname{tr} \left \lvert  1_{I_p} \mathcal K_K 1_{J_p} \right \rvert^{2m} \\
&=  \frac 1 {\pi \sqrt K} \int_0^{\sqrt{2K-1}(1-2 \varepsilon)} \mathrm d p \, \operatorname{tr}   \left \lvert  1_{\hat I_p}\hat{ \mathcal K}_K 1_{\hat J_p} \right \rvert^{2m} \\
&= \frac 1 {\pi \sqrt K} \int_0^{\sqrt{2K-1}(1-2 \varepsilon)} \mathrm d p \, \left[ (1/2)  \left(4\ln(K) \,\mathsf{I}(t \mapsto [t(1-t)]^m) + \mathcal O(\ln\ln(K) )+ \mathcal O(1) \right) \right] \\
&=  \frac {2\sqrt{2K-1}(1-2 /\ln(K))} { \pi \sqrt K} \ln(K) \,\mathsf{I}(t \mapsto [t(1-t)]^m) + \mathcal O(\ln\ln(K) ) \\
&= \frac { 2 \sqrt{2}}{\pi} \,\mathsf{I}(t \mapsto [t(1-t)]^m)  \ln(K) +\mathcal O(\ln\ln(K) )\,. 
\end{align}

\end{proof}

\section{Proof of the expansion \eqref{expansion F(s)}}\label{Appendix C}

In this section, we prove the expansion \eqref{expansion F(s)}, that is, we show that for any $\Lambda \subset \R^2$ with piecewise $\mathsf{C}^2$-smooth boundary $\partial\Lambda$
\begin{align}
F(s)=s \int_{L \Lambda} \mathrm d x \, \int_{0}^{2 \pi} \mathrm d \theta \, 1_{L \Lambda^\complement} ( x+ s (\cos(\theta), \sin(\theta))) = 2 s^2 L \lvert \partial \Lambda \rvert +\mathcal O(s^3)\, .
\end{align}

\begin{lemma}
For any piecewise $\mathsf{C}^1$-smooth domain $\Lambda \subset \R^2, f \in \Lp^1(\R^2) \cap \Lp^\infty(\R^2), v \in \R^2$, $\|v\|=1$, and any $s \in [0, \infty)$, we have the identity
\begin{align} \label{weird int trafo eq}
\int_{\R^2} \mathrm d x \, f(x) 1_{\Lambda^\complement} (x+sv)=  \int_{\Lambda^\complement}\mathrm d x\, f(x) - \int_0^s \mathrm d t \int_{\partial \Lambda} \mathrm d \mathcal H(y) \,f(y-tv) \, n(y) \cdot v \, ,
\end{align}
where $\mathcal H$ is the Hausdorff (surface) measure on $\partial\Lambda$ and where $n(y)$ is the outward unit normal vector of $\partial\Lambda$ at $y$, which is well-defined for $\mathcal H$ almost all $y$.
\end{lemma}

\begin{remark}
While this lemma is proved for piecewise $\mathsf{C}^1$-smooth domains $\Lambda$, the final result of this section requires a somewhat stronger condition and we are content with piecewise $\mathsf{C}^2$-smoothness.
\end{remark}

\begin{proof}
We note that the identity is trivial for $s=0$. 

Let us now assume that $f \in \mathsf{C}^1_{\text{c}} (\R^2)$. We differentiate the right-hand side by $s$ and arrive at
\begin{align}
\partial_s \left( \int_{\R^2} \mathrm d x \, f(x) 1_{\Lambda^\complement} (x+sv) \right) &= \partial_s \left( \int_{\R^2} \mathrm d x \, f(x-sv) 1_{\Lambda^\complement} (x) \right) \\
&= \int_{\Lambda^\complement} \mathrm d x \, \partial_s f(x-sv) \\
&= -\int_{\Lambda^\complement}  \mathrm d x  \, Df(x-sv) v \\
&=  \int_{\Lambda^\complement}  \mathrm d x  \,\operatorname{div} (f(x-sv) v) \\
&= -\int_{\partial \Lambda} \mathrm d \mathcal H(y) \,f(y-sv) \,v \cdot n(y) \, .
\end{align}
The final step is the divergence theorem applied to the Lipschitz domain $\Lambda^\complement \cap D_R(0)$, where $R$ is so large that $L \Lambda \subset D_{R/2}(0)$ and $\operatorname{supp} (f) \subset D_R(0)$. The minus sign appears as the outward normal vector of $ \Lambda^\complement$ at $y \in \partial \Lambda$ is $-n(y)$. This equals the differential in $s$ of the right-hand side of \eqref{weird int trafo eq} and thus, as both sides of \eqref{weird int trafo eq} agree for $s=0$ and their differentials in $s$ agree, we have shown the claim for $f \in \mathsf{C}^1_{\text{c}}(\R^2)$.

To conclude the statement for arbitrary $f \in \Lp^1(\R^2) \cap \Lp^\infty(\R^2)$, for any $\varepsilon_0>0$, we need to construct a function $\hat f \in \mathsf{C}^1_{\text{c}}$ such that both sides of \eqref{weird int trafo eq} are bounded by $\varepsilon_0$ for the function $f-\hat f$. To that end, let $\varepsilon_1 \in (0, \varepsilon_0)$ to be chosen later and let $\hat f$ be given as the convolution of $f$ with a mollifying kernel, such that
\begin{align}
\lVert f - \hat f \rVert_{\Lp^1(\R^2) } \le \varepsilon_1\,, \quad \lVert \hat f \rVert_{\Lp^\infty(\R^2)} \le \lVert f \rVert_{\Lp^\infty(\R^2)} \, .
\end{align}
The $\Lp^1$ estimate deals with the left-hand side and the first expression on the right-hand side of \eqref{weird int trafo eq}. The final expression requires a bit more attention. Let $\varepsilon_2>0$ satisfy
\begin{align}
\int_{\{ y \in \partial \Lambda \colon \lvert n(y) \cdot v \rvert \le \varepsilon_2\}} \mathrm d \mathcal H(y) \,\lvert n(y) \cdot v \rvert < \varepsilon_0/ \left(4 s\lVert f \rVert_{\Lp^\infty(\R^2)} \right) .
\end{align}
Such an $\varepsilon_2>0$ exists, as the expression on the left-hand side vanishes for $\varepsilon_2=0$ and is right-continuous. Let $\partial \Lambda_1\coloneqq \{ y \in \partial \Lambda \colon \lvert n(y) \cdot v \rvert \le \varepsilon_2\}$ and $\partial \Lambda_2 \coloneqq \partial \Lambda \setminus \partial \Lambda_1$. Thus, we can conclude that
\begin{align}
\int_0^s \mathrm dt \int_{\partial \Lambda_1} \mathrm d \mathcal H(y) \,\lvert (f-\hat f)(y-tv) \rvert  \,\lvert n(y) \cdot v \rvert  
\le \left( 2 s \lVert f \rVert_{\Lp^\infty(\R^2)} \right) \varepsilon_0/ \left(4 s\lVert f \rVert_{\Lp^\infty(\R^2)} \right) = \varepsilon_0/2 \, .
\end{align}

On the remaining set, we will use $\lVert f - \hat f \rVert_{\Lp^1(\R^2)} \le \varepsilon_1$. To that end, we estimate
\begin{align}
\int_0^s \mathrm dt \int_{\partial \Lambda_2} \mathrm d \mathcal H(y) \,\lvert (f-\hat f)(y-tv) \rvert \,\lvert n(y) \cdot v \rvert  &\le \int_{\partial \Lambda_2} \mathrm d \mathcal H(y) \int_{\R} \mathrm d t  \,\lvert (f-\hat f)(y-tv) \rvert \\
&\le \sup_{y \in \partial \Lambda_2} \left( \#\left\{(y- v \R) \cap \partial \Lambda_2 \right \} \right) \, \lVert f- \hat f \rVert_{\Lp^1(\R^2)} \\
&\le  \sup_{y \in \R^2} \left( \#\left\{(y- v \R) \cap \partial \Lambda_2 \right\} \right) \varepsilon_1 \, .
\end{align}
As we are still free to choose $\varepsilon_1 < \varepsilon_0$, we only need to show that the supremum is finite for any fixed $\Lambda$ and $\varepsilon_2$. Because $\partial\Lambda$ is piecewise $\mathsf{C}^1$-smooth, it is a finite union of $\mathsf{C}^1$-smooth paths. For $i=1, \dots , r$, let  $\Psi_i \colon [0,\lambda_i] \to \R^2$ be these paths with the normalization $\lVert \Psi_i'(t) \rVert=1$ for all $t \in [0,\lambda_i]$. For any $y \in \R^2$, we observe (recall, $\wedge$ is the wedge product in $\R^2$)
\begin{align}
\#\left\{(y- v \R) \cap \partial \Lambda_2 \right\}\le  \sum_{i=1}^r \# \big\{ t \in [0, \lambda_i] \colon \Psi_i(t) \wedge v = y \wedge v  \mbox{ and }  \lvert \Psi_i'(t) \wedge v \rvert > \varepsilon_2 \big\} \, .
\end{align}
Let $t_1,t_2$ be in this set for some $i\in \{1,2, \dots, r\}$. By Rolle's theorem applied to the function $t \mapsto \Psi_i(t) \wedge v$, there is a $t^* \in (t_1,t_2)$ with $\Psi_i'(t^*) \wedge v= 0$. Thus, as $t\mapsto \Psi_i'(t) \wedge v$ is uniformly continuous, we conclude that $\lvert t_1- t_2 \rvert > \lvert t_1-t^* \rvert > \delta_i=\delta_i(\varepsilon_2, \Psi_i)$. Consequently, there are at most $\lambda_i/\delta_i$ many points in the set, which implies
\begin{align}
 \sup_{y \in \R^2} \left( \#\left\{(y- v \R) \cap \partial \Lambda_2 \right \} \right)  < \infty
\end{align}
for any $\varepsilon_2>0$. Therefore, we can choose $\varepsilon_1>0$ such that 
\begin{align}
\int_0^s \mathrm dt \int_{\partial \Lambda_2} \mathrm d\mathcal H(y) \,\lvert (f-\hat f)(y-tv) \rvert \,\lvert n(y) \cdot v \rvert  \le \sup_{y \in \R^2} \left( \#\left\{(y- v \R) \cap \partial \Lambda_2 \right\} \right) \varepsilon_1  \le \varepsilon_0/2 \, ,
\end{align}
which completes the proof.
\end{proof}

 Having proved this lemma, we are ready to prove the expansion of the function $F$ defined in \eqref{def F}.

\begin{proof}[Proof of \eqref{expansion F(s)}]
Throughout the proof, we assume $L=1$, as the only geometrically relevant parameter is $s/L$, which we study as it approaches $0$.

Let $v \in \R^2, \|v\|=1$. Then, as $\partial\Lambda$ is piecewise $\mathsf{C^2}$-smooth and if we apply \eqref{weird int trafo eq} to $f=1_\Lambda$, we get the integral identity
\begin{align}
 \int_{ \Lambda} \mathrm d x \, 1_{ \Lambda^{\complement}} (x+sv) =\int_0^s \mathrm d t \int_{\partial \Lambda} \mathrm d \mathcal H(y)\, n(y) \cdot v \,1_{\Lambda}(y-tv)\, .
\end{align}
Thus, for $s>0$ and $F$ defined in \eqref{def F}  we can conclude by Fubini, 
\begin{align}
F(s)/s=  \int_0^s \mathrm d t \int_{\partial \Lambda} \mathrm d \mathcal H(y) \int_0^{2\pi} \mathrm d \theta \left( n(y) \cdot (\cos(\theta),\sin(\theta))\right) 1_{ \Lambda} \left( y- t (\cos(\theta),\sin(\theta))\right) . \label{F(s)/s eq}
\end{align}
This shows that $s \mapsto F(s)/s$ is differentiable with
\begin{align}
\left( F(s)/s \right)' = \int_{\partial \Lambda} \mathrm d y \int_0^{2\pi} \mathrm d \theta \,\left( n(y) \cdot (\cos(\theta),\sin(\theta))\right) 1_{ \Lambda} \left( y- s (\cos(\theta),\sin(\theta))\right) .
\end{align}
The integral over $\theta$ obviously takes values between $-2$ and $2$, as it is the integral of a sine function over some subset of its period. We intend to show that it is, in fact $+2+\mathcal O(s))$ for most $y\in \partial \Lambda$.

As $\partial\Lambda$ is piecewise $\mathsf{C}^2$-smooth, we proceed to split $\partial \Lambda$ into $A_1(s)$ and $A_2(s)$, where $A_1(s)$ consists of all points, where $D_s(y) \cap \partial \Lambda$ is a single $\mathsf{C}^2$-smooth curve and $A_2(s)\coloneqq \partial \Lambda \setminus A_1(s)$. We observe  $\mathcal H(A_2(s))\le C s$. 

Let $y \in A_1(s)$. After translation and rotation, we achieve $y=0, n(y)=(0,1)$. We are left to calculate
\begin{align}
\int_0^{2\pi} \mathrm d \theta \sin(\theta) 1_{ \Lambda}(- s (\cos(\theta),\sin(\theta)) \,.
\end{align}
Since $A_1(s) \cap D_s(0)$ is the graph of a $\mathsf{C}^2$-smooth function $\varphi \colon [-s,s] \to \R$ with $\varphi(0)=0$ and $\varphi'(0)=0$, we know that there is a $C>0$, independent of $y \in A_1(s)$ and $s \in (0,1)$ with $\lvert \varphi(t) \rvert \le C t^2$. Thus, for each $\theta \in (0,1)$ with $\lvert \theta \rvert \ge 2 Cs$, we see that
\begin{align}
s \lvert \sin ( \theta) \rvert  \ge (1/2) s\lvert \theta \rvert \ge C s^2 \cos(\theta)^2 \ge \lvert \varphi ( s \cos(\theta) \rvert \,.
\end{align}
This means that $-s (\cos(\theta), \sin(\theta)) \in \Lambda$ is equivalent to $\sin(\theta)>0$ for $\lvert \theta \rvert \ge Cs$. Thus, we can conclude
\begin{align}
&\int_0^{2\pi} \mathrm d \theta \,\sin(\theta) 1_{ \Lambda}(- s (\cos(\theta),\sin(\theta))  
\\
&= \int_{0} ^{\pi}  \mathrm d \theta \, \sin(\theta) + \mathcal O\left ( \int_{-Cs}^{Cs}\mathrm d\theta\, \lvert \sin(\theta)\rvert + \int_{\pi-Cs}^{\pi+Cs} \mathrm d\theta\, \lvert \sin(\theta)\rvert \right)  =  2+ \mathcal O(s^2) \, .
\end{align}
In combination with $\lvert A_2(s) \rvert \le Cs$, we arrive at
\begin{align}
\left( F(s)/s \right)'  &= \lvert A_1(s)\rvert  \left( 2 + \mathcal O(s^2) \right) +\mathcal O(\lvert A_2(s) \rvert )  \\
& =2  \lvert \partial \Lambda \rvert  + \mathcal O(\lvert A_2(s) \rvert + s^2 )= 2 \lvert \partial \Lambda \rvert + \mathcal O(s) \, .
\end{align}
In \eqref{F(s)/s eq}, we can clearly see that $\lim_{s \to 0} F(s)/s=0$. Therefore, by integrating $(F(s)/s)'$, we can finally conclude
\begin{align}
F(s)=2 s^2 \lvert \partial \Lambda\rvert + \mathcal O(s^3) \, .
\end{align}
\end{proof}

\section{Concluding remarks}\label{concluding remarks}

In this final section we make some concluding remarks. 

\begin{enumerate}
\item
We formulated our main results for polynomial test functions $f$ with $f(0) = 0$ and $f(1)=0$ and bounded domains $\Lambda$. While the vanishing of $f$ at 0 is necessary for the trace of the corresponding operators to exist, the vanishing of $f$ at 1 has the effect to cancel the ``volume" term of the order $L^2$. It is trivial to go back to functions that do not vanish at 1 by adding the linear function $-t f(1)$. However, since the number of Landau levels below some Fermi energy $\mu$ is a discontinuous (step) function, the leading volume term may contain lower order terms that would scramble with the surface term when taking the limits $K,L\to\infty$. We are not interested in these terms. Moreover, under the  additional condition $f(t) = f(1-t)$, we may replace $\Lambda$ by its complement without changing the result in the relevant traces, in particular in \autoref{main}. In other words, we may assume that $\Lambda$ or its complement is a bounded domain.

Secondly, the most important application is the entropy function $f$, which in the von Neumann case is the function $f(t) = -t\ln(t) -(1-t)\ln(1-t)$ for $t\in (0,1)$ and zero outside. While it is a standard exercise to extend our result from polynomials $f$ to continuous functions $f$ which are (one-sided) differentiable at the two end-points (see for instance \cite{Leschke2014} and references therein), it is a serious issue to extend our results to functions that are only H\"older continuous (with exponent $\alpha$ less than 1) at the two endpoints. The above entropy function is such an example. Given the length of the present paper, we did not dive into this question but we are confident that it  will be accomplished in a forthcoming paper. What is needed is a trace-norm estimate of the form $\|(1_{L\Lambda} P_K 1_{L\Lambda^\complement} P_K 1_{L\Lambda})^\alpha\|_1\le C L \min\{\ln(K),\ln(L)\}$ with a constant $C$ depending only on the domain $\Lambda$. Such estimates do not trivially follow from known estimates in \cite{Sobolev2014b} for the Laplacian and in \cite{Leschke2021} for the Landau Hamiltonian and require substantial work.

\item Except for the simple but important case of a quadratic function we are not able to cover the full range of parameters $K$ and $L$. The quadratic function is important since it is the first non-trivial example in the Widom conjecture (or Szeg\H{o} asymptotics) and $f(t) = 4\ln(2) t(1-t)$ serves as a lower bound to the entropy function. We can analyse the full range of parameters because the phase (caused by the magnetic field) cancels in the computation of the relevant trace. For higher order polynomials and $L\ll K$ this phase is a nuisance and we only have a rough bound to control it. %This leads for quickly decaying $B$ to the smaller range $L\le C K^{??}$. We expect also here the full range $L\le C K$ for the occurrence of the $L \ln(L)$ area law. 

On the other hand, if $K\ll L$ then the phase caused by the magnetic field is absolutely crucial. When the domain $\Lambda$ is a polygon we are able to cover essentially the full range of parameters $K\le CL$ and prove an area law of the order $L \ln(K)|\partial\Lambda|$. Going to general (piecewise $\mathsf{C}^2$) smooth domains $\Lambda$ we loose some range in order to control additional error terms. Still we believe that the above area law holds in the larger region (of parameters $K,L$) but we are not able to prove it.

For fixed $K \in \mathbb N$, \autoref{reduction lem polygons}, \autoref{Jm domain dependence lem} and \autoref{MK=Jm} prove that the leading order asymptotic expansion in \eqref{1.3} also holds for domains $\Lambda$, which are polygons or $\mathsf{C}^2$-domains, instead of  $\mathsf{C}^3$-domains as in \cite{Leschke2021}. This result can be extended to merely Hölder continuous functions $f$, as the estimate \cite[Theorem 13]{Leschke2021} only needs $\Lambda$ to be a Lipschitz domain.

\item Throughout the paper we worked with a constant (that is, translation invariant) magnetic field whose strength goes to zero. One can ask what happens with an arbitrary magnetic field $\lambda B(x)$ in the limit $\lambda\to0$, $L\to\infty$ and the function $B$ fixed. In \cite{pfeiffer2021stability}, one of the present authors has analysed the stability of the area law under a varying magnetic field. Roughly speaking, if the magnetic field is asymptotically (for large $\|x\|$) constant then we observe the same area law as for a constant magnetic field. It is reasonable to expect that the same stability holds in the joint limit $\lambda\to0,L\to\infty$ discussed here ``only" for a spatially constant magnetic field. 

The situation is unclear for example for a magnetic field $B(x)$ which tends to 0 as $\|x\|\to\infty$. Depending on the rate of convergence, it may create absolutely continuous spectrum or even change completely to an absolutely continuous spectrum for $H_B$ on $[0,\infty)$, see \cite{MillerSimon}. We do not know whether an area law or an enhanced area-law holds as the appearance of absolutely continuous spectrum may still yield an area law. As a warning, we have found ground states of Hamiltonians with purely absolutely continuous spectrum that display an area law. %On the other hand, \cite{PeterMueller} has constructed states that resemble half-filled states of the Landau model with an enhanced area law.

%So in this general case it is not clear to us what the range of parameters $(\lambda,L)$ will be to see an enhanced area law of the form $L \ln(L) |\partial\Lambda|$ ($\lambda \ll 1/L$?) and when to meet an area law of the form $L \ln(1/\lambda) |\partial\Lambda|$ ($\lambda\gg 1/L$?).

\end{enumerate}

\bibliography{mybib2020}{}
\bibliographystyle{abbrvurl}

\end{document}